\documentclass[11pt]{article}
\usepackage[IL2]{fontenc} 
\usepackage[a4paper,margin=2.5cm]{geometry}
\usepackage{todonotes}

\usepackage[T1]{fontenc} 
\usepackage{hyperref}
\usepackage{float}
\usepackage{diagbox}
\usepackage{comment}
\usepackage{amssymb}
\usepackage{amsmath}
\usepackage{amsthm}
\usepackage{xspace}
\usepackage[noend]{algpseudocode}
\usepackage{bm}
\usepackage[normalem]{ulem}
\usepackage[subrefformat=parens]{subcaption}
\usepackage{multicol}	
\usepackage{mdframed}
\usepackage{cleveref}
\usepackage[switch,mathlines]{lineno}
\usepackage{tikz}
\usetikzlibrary{automata, positioning, arrows.meta, bending}
\usetikzlibrary{calc}
\tikzset{
	invisible/.style={opacity=0},
	emph/.style={color=magenta},
	alert/.style={color=red},
	anotherhl/.style={color=blue},
	bold/.style={very thick},
	emph on/.style={alt={#1{emph}{}}},
	alert on/.style={alt={#1{alert}{}}},
	anotherhl on/.style={alt={#1{anotherhl}{}}},
	arrow on/.style={alt={#1{->}{}}},
	bold on/.style={alt={#1{bold}{}}},
	visible on/.style={alt={#1{}{invisible}}},
	alt/.code args={<#1>#2#3}{%
		\alt<#1>{\pgfkeysalso{#2}}{\pgfkeysalso{#3}} 
	},
}

\newtheorem{theorem}{Theorem}
\newtheorem{lemma}[theorem]{Lemma}
\newtheorem{proposition}[theorem]{Proposition}
\newtheorem{corollary}[theorem]{Corollary}
\newtheorem{remark}[theorem]{Remark}
\newtheorem*{remark*}{Remark}

\newcommand{\newextmathcommand}[2]{%
	\newcommand{#1}{\ensuremath{#2}\xspace}
}

\newcommand{\renewextmathcommand}[2]{%
	\renewcommand{#1}{\ensuremath{#2}\xspace}
}

\newextmathcommand{\PTIME}{{\mathsf{PTIME}}}
\newextmathcommand{\EXPTIME}{{\mathsf{EXPTIME}}}
\newextmathcommand{\NL}{{\mathsf{NL}}}
\newextmathcommand{\TwoNPDAk}{{\text{2NPDA$(k)$}}}
\newextmathcommand{\TwoNPDA}{{\text{2NPDA}}}
\newextmathcommand{\TwoNPDAclass}{{\mathsf{2NPDA}}}
\newextmathcommand{\TwoNPDAhyp}{{\mathsf{2NPDA}}}
\newextmathcommand{\TwoNPDAkhyp}{{\mathsf{2NPDA}}(k)}
\newextmathcommand{\boldTwoNPDA}{{\mathbf{2NPDA}}}
\newextmathcommand{\TwoDPDA}{{\text{2DPDA}}}

\newextmathcommand{\TwoNFAAcc}{\text{\textsc{2NFA Acceptance}}}
\newextmathcommand{\TwoNPDARec}{\text{\textsc{2NPDA Recognition}}}
\newextmathcommand{\DyckTwoReach}{\mathsf{D_2Reach}}
\newextmathcommand{\AllPairsDyckTwoReach}{\mathsf{All\mbox{-}Pairs\mbox{-}D_2Reach}}

\newextmathcommand{\instance}{\bm{x}}

\newextmathcommand{\Lang}{\mathcal L}
\newextmathcommand{\RLang}{\mathcal L'}
\newextmathcommand{\A}{\mathcal A}
\newextmathcommand{\Ak}{\mathcal A_k}
\renewextmathcommand{\AA}{\mathcal A'}
\newextmathcommand{\B}{\mathcal B}
\newextmathcommand{\Rp}{\mathcal R}
\newextmathcommand{\PDA}{\mathcal P}
\newextmathcommand{\PDAA}{\mathcal P'}
\newextmathcommand{\SAT}{\mathsf{SAT}}
\newextmathcommand{\SharpSAT}{\mathsf{\# SAT}}
\newextmathcommand{\TAUT}{\mathsf{TAUT}}
\newextmathcommand{\SETH}{\mathsf{SETH}}
\newextmathcommand{\NSETH}{\mathsf{NSETH}}
\newextmathcommand{\NTIME}{\mathsf{NTIME}}
\newextmathcommand{\CONTIME}{\mathsf{co\text{-}NTIME}}
\newextmathcommand{\DTIME}{\mathsf{DTIME}}
\newextmathcommand{\MATIME}{\mathsf{MATIME}}
\newextmathcommand{\coMATIME}{\mathsf{co\text{-}MATIME}}
\renewextmathcommand{\P}{\mathsf{P}}
\newextmathcommand{\NP}{\mathsf{NP}}
\newextmathcommand{\MA}{\mathsf{MA}}
\newextmathcommand{\PSPACE}{\mathsf{PSPACE}}
\newextmathcommand{\HardestA}{\A_0}
\newextmathcommand{\AlmostHardestA}{\A'_0}
\newextmathcommand{\prestar}{\mathrm{Pre}^*(R)}

\newextmathcommand{\eps}{\varepsilon}
\renewcommand{\epsilon}{\varepsilon}

\newextmathcommand{\partto}{\rightharpoonup}

\newextmathcommand{\opr}{\text{\textup{\texttt{(}}}}
\newextmathcommand{\clr}{\text{\textup{\texttt{)}}}}
\newextmathcommand{\ops}{\text{\textup{\texttt{[}}}}
\newextmathcommand{\cls}{\text{\textup{\texttt{]}}}}
\newextmathcommand{\BrTwo}{\{ \opr, \clr, \ops, \cls \}}
\newextmathcommand{\OpenBr}{\{ \opr, \ops \}}
\newextmathcommand{\CloseBr}{\{ \clr, \cls \}}

\newextmathcommand{\qfinal}{q_\textsc{f}}
\newcommand{\set}[1]{\{ #1 \}}
\newextmathcommand{\lmarker}{{\lhd}}
\newextmathcommand{\rmarker}{{\rhd}}
\newextmathcommand{\llmarker}{{\trianglelefteq}}
\newextmathcommand{\rrmarker}{{\trianglerighteq}}
\newextmathcommand{\shuf}{\parallel}

\newcommand{\twonpda}{\mathcal{A}}
\newextmathcommand{\pureSigma}{\Sigma \setminus \{\lmarker, \rmarker\}}

\newextmathcommand{\sdelim}{{\mathtt{;}}}
\newextmathcommand{\strack}{\diamond}
\newextmathcommand{\Bin}{\{0, 1\}}

\newextmathcommand{\aAND}{\textsc{and}}
\newextmathcommand{\aOR}{\textsc{or}}
\newextmathcommand{\hash}{\mathtt\#}
\newextmathcommand{\amdelima}{\star}
\newextmathcommand{\amdelimb}{*}

\newextmathcommand{\offset}{\mathsf{offset}}
\newextmathcommand{\ind}{\mathsf{index}}
\newextmathcommand{\vmark}{\mathbin{\#}}
\newextmathcommand{\edgesep}{*}
\newextmathcommand{\poly}{\mathrm{poly}}
\newextmathcommand{\polylog}{\mathrm{polylog}}

\newcommand*{\myaccept}{\textbf{accept}\ }

\newextmathcommand{\AllNonTerminals}{\overrightarrow{V^2}}
\newextmathcommand{\NonTerminals}{\mathsf{NT}}

\newextmathcommand{\pda}{P}
\newextmathcommand{\fsa}{A}
\newextmathcommand{\two}{\mathcal M}
\newextmathcommand{\EncAlph}{\Delta}
\newextmathcommand{\LangAlph}{\Sigma}
\newcommand{\langof}[1]{\mathcal{L}(#1)}
\newextmathcommand{\aut}{\mathcal M}
\newextmathcommand{\hrd}{\mathcal H}
\newextmathcommand{\pu}{P}
\newcommand{\pdanfa}[1]{\mathrm{PDA\cap NFA}^{\!#1}}
\newcommand{\pdadfa}[1]{\mathrm{PDA\cap DFA}^{\!#1}}
\newcommand{\dpdadfa}[1]{\mathrm{DPDA\cap DFA}^{\!#1}}
\newextmathcommand{\mainpdanfa}{\pdanfa{k-1}}
\newextmathcommand{\mainpdadfa}{\pdadfa{k-1}}
\newextmathcommand{\maindpdadfa}{\dpdadfa{k-1}}
\newcommand{\boldpdanfa}[1]{\mathbf{PDA \bm{\cap} NFA}^{\!#1}}
\newcommand{\boldpdadfa}[1]{\mathbf{PDA \bm{\cap} DFA}^{\!#1}}
\newextmathcommand{\boldmainpdanfa}{\boldpdanfa{\bm{k-1}}}
\newextmathcommand{\boldmainpdadfa}{\boldpdadfa{\bm{k-1}}}

\newcommand*{\br}{\nobreak\discretionary{}{\hbox{}}{}}

\newcommand{\Fix}{{Fix}:\quad}
\newcommand{\Input}{\uline{Input}:\quad}
\newcommand{\Decide}{\uline{Decide}:\quad}
\newcommand{\qaqed}{\hfill$\vartriangleleft$}

\newcommand{\mach}{\mathcal{M}}

\newcommand{\bala}[1]{}
\newcommand{\dch}[1]{}

\title{Pushdown Model Checking Above the Cubic Bottleneck}

\author{%
	A. R. Balasubramanian\thanks{\footnotesize{A part of the work was done when the author was at TUM, Germany}}\\
	Max Planck Institute for Software Systems\\
	(MPI-SWS), Kaiserslautern, Germany\\
	\texttt{bayikudi@mpi-sws.org}
	\and
	Dmitry Chistikov\\
	University of Warwick\\
	Coventry, United Kingdom\\
	\texttt{d.chistikov@warwick.ac.uk}
	\and
	Rupak Majumdar\\
	Max Planck Institute for Software Systems\\
	(MPI-SWS), Kaiserslautern, Germany\\
	\texttt{rupak@mpi-sws.org}
}

\begin{document}
	
	\date{}
	\maketitle

	\begin{abstract}
Many problems in program analysis and the verification of recursive programs can be reduced to \emph{pushdown model checking}. 
In this problem, we are given as input a pushdown automaton (PDA) over a constant-sized stack alphabet, representing the program, and a 
description of undesirable behaviors given by an intersection of NFAs, and the problem is to decide if there is a behavior of the PDA
that belongs to the set of undesirable behaviors. 
It is well-known that there is an algorithm for this problem that runs in time $O(n^{2k} |\Sigma| + n^{3k})$, where $n$ is the maximum number of states of the PDA and the NFAs, 
$\Sigma$ is the common alphabet of these machines, and $k-1$ is the number of NFAs used to specify the violations. 
Despite the importance of this problem, no better algorithm is known for it.

In this paper, we provide an explanation for this lack of progress using the lens of \emph{fine-grained complexity theory}. 
More precisely, we prove that if the $3k$-Clique hypothesis is true, then for any $\epsilon > 0$,
there is no algorithm that solves pushdown model checking in time $O((n^{(\omega-1)k} |\Sigma| + n^{\omega k})^{1-\epsilon})$ where $\omega$ is the matrix multiplication exponent. 
Similarly, under the combinatorial $3k$-Clique hypothesis, we show that no combinatorial algorithm 
can solve pushdown model checking in time $O((n^{2k} |\Sigma| + n^{3k})^{1-\epsilon})$ for any $\epsilon > 0$.
Hence, our result implies that any significantly faster algorithm for pushdown model checking than the existing ones would lead to a breakthrough for the $3k$-Clique problem.
Our lower bound applies even in the case when all the machines are deterministic,
and even when the PDA is simply a deterministic one-counter machine. 
Furthermore, using the same combinatorial $3k$-clique hypothesis, we also show that
pushdown model checking over constant-sized input alphabets cannot be solved by combinatorial algorithms in time $O(n^{3(k-1)-\epsilon})$ for any $\epsilon > 0$. 

Finally, we investigate the possibility of an $O(N^{3k-\epsilon})$ time algorithm for pushdown model checking where $N$ is the total bit size of the given input. 
We formulate a new hypothesis, the 2NPDA$(k)$ hypothesis, that helps  
explain the lack of $O(N^{3k-\epsilon})$ time algorithms for pushdown model checking.
To corroborate this hypothesis, we show a web of linear-time reductions between the 2NPDA$(k)$ hypothesis, pushdown  model checking, and other problems in formal language and automata theory. 
\end{abstract}

	\newpage


	
	\section{Introduction}
\label{sec:intro}


Many problems in program analysis, formal language theory and verification of recursive programs
are reducible to \emph{pushdown model-checking}: given as input a PDA (over a constant-sized stack alphabet) and 
a specification of bad behaviors given as a set of NFAs,
the goal is to check if there is a behavior in the language of the PDA that lies in the intersection of the languages of the NFAs. 
It is well-known that the \emph{model complexity} of the problem, which treats the PDA as the input and the NFAs as constant-sized,
is $O(n^3)$ time, where $n$ is the number of states of the given PDA \cite{DolevEK82,Yannakakis90,RHS95,BEM97,FinkelWW97,Pavlogiannis-survey}. 
This cubic algorithm simply does a product construction of the given NFAs with the PDA to get another PDA and then performs reachability analysis on this latter PDA. 

When we consider the \emph{combined complexity} of the problem, where both the PDA and the set of NFAs are considered part of the input,
the best algorithm runs in time $O(n^{2k}|\Sigma| + n^{3k})$ where $n$
is the maximum number of states in any of the given NFAs and the PDA, $\Sigma$ is the common input alphabet of all of the 
machines and $k-1$ is the number of NFAs used to encode the specification.
This algorithm uses the algorithm from the previous paragraph.
Intuitively, it takes time $O(n^{2k}|\Sigma|)$ to do the product construction of the NFAs with the given PDA to get another PDA with $O(n^k)$ states, 
and then it takes $O(n^{3k})$ time to perform reachability analysis on this new PDA. 
Going beyond constant-sized specifications is important in several program analysis applications, such as 
checking set constraints with regular annotations \cite{KodumalA07}, pointer analysis \cite{RountevMR01},
and model checking for certain security properties \cite{ChenW02}. 

Despite decades of research, no substantially better algorithm is known for either problem.
For model complexity, the best current bound is $O(n^3/\log n)$ time \cite{Chaudhuri}.
No truly subcubic algorithms are known; this is usually referred to as the ``cubic bottleneck'' for program analysis \cite{HeintzeMcAllester}. 
Similarly, for combined complexity, the trivial algorithm above is the best known.

In this paper, we explain this lack of progress through the lens of \emph{fine-grained complexity theory}.
This subfield of complexity theory is one of its recent successes:
it aims to explain the absence of faster algorithms (than existing ones) for various
polynomial-time solvable problems. More precisely, it attempts to explain why
a problem can be solved in time $O(n^c)$ but
not $O(n^{c-\epsilon})$ for any $\eps>0$.
By identifying a small number of relevant hypotheses and linking many such questions of the above kind with them, fine-grained complexity has provided tight conditional lower bounds for problems
in a variety of domains:
graph theory, stringology, formal language theory, databases, data structures, and dynamic algorithms~\cite{WilliamsSurvey,Bringmann19tutorial}.

Within this context, our contributions in this paper are the following:
\begin{enumerate}

\item
We prove that significantly faster algorithms for pushdown model checking
would lead to breakthrough algorithms for the $k$-Clique problem
(and thus would refute a well-established hypothesis).

\item
We show that the hypothesis that pushdown model checking requires
time $n^{3k}$ is closely related (by linear-time inter-reductions)
to another natural hypothesis, on the time complexity
of languages recognized by two-way nondeterministic pushdown automata with $k$~heads
($\TwoNPDA(k)$, for short).

\item
We give a language-theoretic consequence of our reductions:
a new construction of hardest languages for the class of languages
recognized by $\TwoNPDA(k)$, for each~$k$.
%

\end{enumerate}

We now expand upon each of these contributions in more detail.

\paragraph*{1. Conditional lower bounds from $k$-Clique.}
One of the most important hypotheses used in fine-grained complexity theory is the so-called
$k$-Clique hypothesis ($k \ge 3$):
For each $\epsilon > 0$, there is no algorithm that detects the existence of cliques of size~$k$ in graphs on $n$~vertices in time $O(n^{\omega k/3-\epsilon})$. 
(Here $\omega$ is the matrix multiplication exponent: the infimum of all $c$ such that
there is an $O(n^c)$ algorithm to multiply two $n \times n$ matrices.)
Closely related to this hypothesis is the \emph{combinatorial} $k$-Clique hypothesis
which asserts that, for each $\epsilon > 0$, there is no combinatorial\footnote{
	While the notion of \emph{combinatorial algorithms} is not rigorous, these are algorithms that do not use fast matrix multiplication techniques.
} 
algorithm that solves $k$-Clique in time $O(n^{k-\epsilon})$. 
These two hypotheses have been used to show conditional lower bounds for a variety of problems such as
context-free language recognition and RNA folding~\cite{AbboudBW18},
parsing tree adjoining grammars \cite{BringmannW17},
Klee's measure problem~\cite{Chan08},
maximum-weight box problem in computational geometry~\cite{BackursDT16}, etc.
In fact, the best conditional lower bounds for the \emph{model complexity} of pushdown model checking
are based on the 3-Clique hypothesis \cite{ChatterjeeCP18,AbboudBW18,MathiasenP21,ChistikovMS22};
see also~\cite{PotechinS20}.
However, these lower bounds do not imply anything for the \emph{combined complexity} of pushdown model checking. \bala{added this sentence based on reviewer 3's comments.}

Our first main result is to provide a tight conditional lower bound on the combined complexity of pushdown model checking based on the Clique hypotheses. 
We show that, unless the $3k$-Clique
hypothesis (resp. combinatorial $3k$-Clique hypothesis) is false, there is no (combinatorial) algorithm that solves pushdown model checking 
in time $O((n^{(\omega-1)k}|\Sigma| + n^{\omega k})^{1-\epsilon})$ 
(resp. in time $O((n^{2k}|\Sigma| + n^{3k})^{1-\epsilon})$) for any $\epsilon > 0$. 
Hence, our result proves that, unless the combinatorial $3k$-Clique hypothesis is false, no algorithm can be better than the known algorithm for this problem. 
In fact, our lower bound applies even to the special case of the problem
in which all the $k-1$ NFAs are DFAs, and the PDA is a deterministic  one-counter machine.

Our proof (of Theorem~\ref{thm:general-case-lower-bound} in Section~\ref{sec:lb-general})
generalizes the well-known encoding of triangle finding using PDAs (cf.~\cite{PotechinS20}), but requires
several new ideas, in order to go beyond triangles.
First, a large input alphabet (with up to $n^k$ letters) enables us to ``name'' $k$-cliques.
With the help of a counter (counting up to $n^k$), we can uniquely store and retrieve a $k$-clique.
We use the counter along with the NFAs to find three
$k$-cliques whose nodes are all neighbors to each other.
Our key trick here is to use small automata ($n$ states each) to increment and decrement 
the counter all the way to these large numbers, as well as check the neighborhood relation between
nodes of different $k$-cliques.

We then turn our attention to the case of pushdown model checking when the input alphabet of all the machines is constant-sized, i.e.,
when $|\Sigma|$ is a constant. 
This is an important special case: the complexity of language-theoretic problems is often studied under this assumption.
Because $|\Sigma|$ is a constant, we have that $O(n^{2k}|\Sigma| + n^{3k}) = O(n^{3k})$, which is the best known running time in this setting.
For this case (handled in Theorem~\ref{thm:constant-alphabet-lower-bound} in Section~\ref{sec:lb-constant}),
based on the $3k$-Clique hypothesis, we rule out the existence
of $O(n^{\omega (k-1) - \epsilon})$ time algorithms. 
Similarly, using the combinatorial $3k$-Clique hypothesis, we also prove that $O(n^{3(k-1)-\epsilon})$ time combinatorial algorithms cannot exist for this problem. 
Since the alphabet size is fixed, we can no longer use the trick of naming $k$-cliques.
To go beyond this, our construction initially uses linear-sized input and stack alphabets
in order to find three $k$-cliques whose nodes are all neighbors to each other.
For this purpose, various specialized gadgets are constructed to store and keep track of (multiple) $k$-cliques
in the stack as well as to check that nodes in different $k$-cliques are neighbors. 
We then carefully encode the linear-sized alphabet in binary
and convert the original construction into one over constant-sized alphabets, with only a logarithmic blowup in the state space.
Note that our results in this regime have a gap of $O(n^{3 + (3-\omega)k})$ between the upper and lower bounds in the general case and $O(n^3)$ in the combinatorial case. 

We also note that, for the special case of $k = 2$, that is, the intersection non-emptiness problem
for the language of 1~PDA and 1~NFA, a conditional lower bound follows from a result
in~\cite{AbboudBBK17}.
They consider the special case of the problem in which the PDA is replaced by
a straight-line program (a context-free grammar that generates a single
word only). The problem is to decide whether a given NFA accepts this \emph{compressed word}.
They prove that, unless the combinatorial $k$-Clique hypothesis is false, there is no $\epsilon > 0$ for which
there exists an algorithm for this problem running in time
$O(\min\{p q^3, N q^2\}^{1 - \epsilon})$, where
$p$ is the size of the compressed representation of the word
(think the number of states in the PDA),
$q$ is the number of states in the NFA, and $N \le 2^n$.
The $N q^2$ term matches a decompress-and-solve algorithm which is not available
for general PDA.
In the regime $p = q = n$, this lower bound is of order $n^4$.

\paragraph*{2. New hypotheses.}
It is not known whether fast matrix multiplication algorithms can be used
for faster pushdown model checking.
Standard existing hypotheses appear to be insufficient for explaining the hardness
even for $k = 1$ (language non-emptiness of pushdown automata).
The best (non-combinatorial) lower bound is $\Omega(n^\omega)$ from $k$-Clique~\cite{AbboudBW18}, and
it has been shown that 
the strong exponential-time hypothesis (SETH),
perhaps the most well-known hypothesis in fine-grained complexity,
cannot be used to beat this bound, unless breakthrough results in circuit complexity appear~\cite{ChistikovMS22}.
Note that the $3k$-Clique hypothesis only asserts the non-existence of algorithms with runtime as a function of the number of nodes and not of the whole input. 
Intuitively, this cannot help explain
whether pushdown model checking admits an algorithm with runtime $O(N^{3k-\epsilon})$ where $N$
is the overall bit size of the input: 
the input includes all transitions in the machines, and their number could be quadratic in the number of states. 

Thus, new hypotheses may be required to explain the absence of faster algorithms.
Indeed, for $k = 1$ the recent NFA acceptance hypothesis~\cite{BringmannGKL24-j}
gives an $n^3$ lower bound for pushdown model checking for \emph{dense} PDAs.


For $k > 1$,
we introduce (in Section~\ref{sec:2npdak}) a new \emph{\TwoNPDAkhyp hypothesis}.
It asserts that there is no $\epsilon > 0$ for which some algorithm running in time $O(|w|^{3k-\epsilon})$
can decide if a given word $w$ is accepted by a (fixed)
two-way non-deterministic pushdown automaton with $k$ heads (\TwoNPDAk). 
Intuitively, a \TwoNPDAk is a machine (see~\cite{HarrisonI68,Ibarra73})
which has access to a stack and its input is written on a read-only input tape. 
This machine has $k$ heads on the input tape using which it can query the letters of $k$ positions on the input tape. 
Based on this query, it can update its state, the positions of these $k$ heads and also its stack content. 
It is a folklore result that 
there is an algorithm for the \TwoNPDAk acceptance problem that runs in time $O(|w|^{3k})$ for any fixed 2NPDA($k$).
(For example, Rytter~\cite{Rytter82} refers to Aho, Hopcroft, and Ullman~\cite{AhoHopcroftUllman68}, even though only
 2NPDA(1) are considered there.)
However, no faster algorithm is known for this problem.
Based on this lack of progress, we introduce this new hypothesis as a generalization of the \TwoNPDAhyp hypothesis (i.e., \TwoNPDAkhyp with $k=1$),
introduced by Neal \cite{Neal} and Heintze and McAllester \cite{HeintzeMcAllester}.
Furthermore, if $k \ge 2$, $\TwoNPDA(k)$ language recognition is not known to admit even $O(|w|^{3 k} / \log |w|)$ algorithms,
unlike for $k = 1$~\cite{Chaudhuri,Rytter85}.
To the best of our knowledge, algorithms with this complexity are only known for the special case of loop-free automata~\cite{Rytter85}; see also related results in~\cite{Rytter82,Rytter83a}.

As a way to strengthen the believability of this hypothesis, we provide a web of linear-time reductions 
between the \TwoNPDAk language recognition problem and a variety of other problems in language theory and program analysis, 
one of which is the pushdown model checking problem. (For $k=1$, similar reductions appear in~\cite{Chaudhuri} and later in~\cite{ChistikovMS22}.)
So, under the \TwoNPDAkhyp hypothesis, we show that there is no algorithm running in time $O(N^{3k-\epsilon})$ for pushdown model checking.
Hence, this gives rise to a hierarchy of program analysis problems, one for each $k > 1$, that go beyond the famous cubic bottleneck that is established in the literature.

\paragraph*{3. Consequences and applications.}
The lens of fine-grained complexity provides some purely language-theoretic consequences.
Using our chain of linear-time reductions between $\TwoNPDA(k)$ recognition and pushdown model checking,
we get that, for each $k\geq 1$, there is a \emph{hardest} $\TwoNPDA(k)$ language, i.e., there
is a fixed language $L_0^{(k)}$ recognized by a $\TwoNPDA$ with $k$ heads with the following property: For every $\TwoNPDA(k)$ language $L$,
there is a homomorphism $h$ such that, for all non-empty words~$w$, $w \in L$ if and only if $h(w) \in L_0^{(k)}$.
Previous constructions of hardest languages used language-theoretic constructions \cite{Miyano83},
\cite{Rytter-hardest}.


This paper is an extended version of the conference paper~\cite{BalasubramanianCM25} published at LICS 2025. In comparison to that paper, this version contains full proofs of the results as well as diagrams illustrating the main gadgets used in the reductions.

\section{Preliminaries}
\label{s:pre}

\subsection{PDAs and NFAs} A pushdown automaton (PDA) 
consists of a finite set of control states and a stack into which it can push/pop elements. Initially, the PDA begins in some designated initial control state and reads the input word $w$ one letter at a time. 
As it reads each letter, its transition relation allows it to move from one control state to another
whilst pushing/popping elements from its stack. At the end of reading $w$, if the machine is in one of a designated set of final states and its stack content is empty, then it is said to accept $w$. The language of the machine is the set of all words that it accepts.

Formally, a PDA is a tuple $P = (Q,\Sigma,\Gamma,\delta,q_0,F)$ where $Q$ is a finite set of control states,
$\Sigma$ is the input alphabet, $\Gamma$ is the stack alphabet, $\delta \subseteq Q \times (\Sigma \cup \{\epsilon\}) \times \Gamma \times Q \times \Gamma^{*}$ is the transition relation, $q_0$ is the initial state and $F$ is a set of final states. A transition of the form
$(p,a,\alpha,q,\beta)$ dictates that, in order for this transition to be used, the PDA must be at state $p$, read the letter $a$ (or read no letter if $a = \epsilon$) and then it must pop $\alpha$ from the stack, move to state $q$ and push $\beta$ into the stack. 
\begin{remark*}
	Note that $\beta$ could be $\epsilon$, which means that the net effect is simply popping $\alpha$ from the stack. Similarly, $\beta$ could also be of the form $\alpha \zeta$ for some letter $\zeta$, which means that the net effect is simply pushing $\zeta$ into the stack.	
\end{remark*}
We assume that $\Gamma$ contains a designated ``end of stack'' symbol~$Z_0$
such that no transition of $P$ replaces $Z_0$ on the stack with a different symbol, or pushes $Z_0$ 
on the stack when the top of the stack is not $Z_0$, i.e.,
for any transition $(p, a, \alpha, q, \beta)$, $\beta$ contains $Z_0$ if and only if $\alpha = Z_0$ and $\beta = Z_0 \gamma$ for some $\gamma \in (\Gamma \setminus \{Z_0\})^*$.
Throughout the paper, we will assume that the stack alphabet $\Gamma$ is of constant size, i.e., $|\Gamma| = O(1)$ unless specifically stated otherwise.

A \emph{configuration} of the PDA is a pair $(p,\gamma)$ where $p$ is a state and $\gamma \in \Gamma^*$ is the stack content.
For any transition $t = (p,a,\alpha,q,\beta)$, there is a step from a configuration of the form
$(p,\gamma \alpha)$ to $(q,\gamma \beta)$, which represents the changes made to the stack
as dictated by $t$. We will denote this step by $(p,\gamma \alpha) \xrightarrow{t} (q,\gamma \beta)$ 
or simply $(p,\gamma \alpha) \xrightarrow{a} (q,\gamma \beta)$ when only the component $a \in \Sigma \cup \{\epsilon\}$ of the transition $t$ is important.

A transition $(p,a,\alpha,q,\beta)$ is called an $\epsilon$ transition if $a = \epsilon$.
We say that a configuration $(q,\gamma)$ can reach a configuration $(q',\gamma')$ by $\epsilon$ transitions if there are configurations $(q_0,\gamma_0), \dots, (q_k, \gamma_k)$ such that $(q,\gamma) = (q_0,\gamma_0) \xrightarrow{\epsilon} (q_1,\gamma_1) \xrightarrow{\epsilon} (q_2,\gamma_2) \xrightarrow{\epsilon}
\dots (q_k,\gamma_k) = (q',\gamma')$. 

The initial configuration is $(q_0,Z_0)$.
A run of the PDA on a word $w = w_1w_2\dots w_n \in \Sigma^*$ is a sequence of configurations 
of the form $(q_0,\gamma_0), (q_1,\gamma_1),\dots, (q_n,\gamma_n)$ such that $(q_0,\gamma_0)$ is the initial configuration and, for each $i \ge 0$, there exist configurations $(q_i',\gamma_i'), (q_i'',\gamma_i'')$ such that $(q_i,\gamma_i)$ can reach $(q_i',\gamma_i')$ by $\epsilon$ transitions,
$(q_i',\gamma_i') \xrightarrow{w_{i+1}} (q_i'',\gamma_i'')$ and $(q_i'',\gamma_i'')$ can reach
$(q_{i+1},\gamma_{i+1})$ by $\epsilon$ transitions.
A run is said to be accepting if $\gamma_n$ is empty and $q_n \in F$.
The \emph{language} of a PDA $P$, denoted by $\langof{P}$, is the set of all words that it accepts, i.e., the set of all words
on which it has an accepting run. PDA are known to accept exactly the set of context-free languages.

A one-counter automaton (OCA) is a PDA with a stack alphabet containing only one letter (say $\alpha$) apart from $Z_0$.
Note that the content of a stack is then uniquely determined by the number of times $\alpha$ appears in it. Hence, in this case the stack can be thought of as a counter, where pushing/popping $\alpha$ corresponds to incrementing/decrementing the counter, respectively; and popping $Z_0$ corresponds to testing whether the counter is zero. 
Hence, the accepting condition for a run in an OCA is that the state at the 
end is an accepting state and the counter has reached the value 0. 

A PDA is deterministic if its transition relation is a partial function (i.e., $\delta$ is of the form $Q \times (\Sigma \cup \{\epsilon\}) \times \Gamma \partto Q \times \Gamma^{*}$) and, moreover,
for each $q \in Q$, either $\delta(q,a,\alpha)$ is defined for all $a \in \Sigma$, $\alpha \in \Gamma$ and $\delta(q,\epsilon,\alpha)$ is undefined for all $\alpha \in \Gamma$, or
                           $\delta(q,a,\alpha)$ is undefined for all $a \in \Sigma$, $\alpha \in \Gamma$ but $\delta(q, \epsilon,\alpha)$ is defined for all $\alpha \in \Gamma$.
A deterministic PDA (resp. deterministic OCA) will be succinctly referred to as a DPDA (resp. DOCA).

An NFA is a PDA in which there are no $\epsilon$ transitions and no stack operations are performed,
i.e., no push or pop happens to the stack.
(Formally, $\alpha = \beta = Z_0$ in all transitions $(p, a, \alpha, q, \beta)$.)
Hence, the transitions of an NFA can simply be presented as a relation $\delta \subseteq Q \times \Sigma \times Q$. A DFA is an NFA for which the transition relation is actually a function. NFAs and DFAs accept exactly the set of regular languages.

\subsection{Intersection Non-Emptiness Problems}\label{subsec:intersect}

This paper focuses on the following class of problems, one for \emph{each fixed $k \ge 1$}: 
%
\begin{mdframed}
\textbf{\boldmainpdanfa Non-Emptiness}\\[1.0ex]
	\Fix Stack alphabet $\Gamma$.\\
	\Input
	PDA \pda and $k-1$~NFAs $\fsa_1, \ldots, \fsa_{k-1}$,
	all over a common input alphabet $\Sigma$. 
	\\
	\Decide
	Is the intersection
	$\langof{\pda} \cap \langof{\fsa_1} \cap \ldots \cap \langof{\fsa_{k-1}}$ non-empty?
\end{mdframed}

When $k = 1$, this is simply checking if the given PDA $\pda$ has a non-empty language,
which is the well-known PDA non-emptiness problem.

%
In this paper, we let the stack alphabet $\Gamma$ be fixed (as it is in
applications related to program analysis). In contrast, the common input alphabet of the automata is provided as part of the input, as it forms a part of the description of the PDA and each NFA.
In applications, the input alphabet needs to be rich enough to specify possible actions or events in the system.
In theory, it is often instructive to consider constant-sized alphabets too, e.g., $|\Sigma|=2$.
Furthermore, there are applications related to PDAs in which the underlying language is fixed, which automatically fixes the alphabet as well. For instance, for the CFL Reachability problem (which we generalize and study as the CFL $k$-Intersection Reachability problem in Section~\ref{subsec:equivalences} in our paper), the underlying context-free language (and hence the input alphabet) is fixed.


In general, the model checking problem may use multiple NFAs to encode a specification.
We give a small example from the domain of model checking for security properties \cite{ChenW02,ChenDW04}.
A program is modelled as a PDA, and security properties as a set of NFAs.
For example, a security property is ``a program should drop privileges from all its user IDs before
calling certain system functions.'' 
This property is expressed as multiple NFAs: one NFA tracking if certain system calls have been made,
and the others tracking which user IDs have root privilege.
The common alphabet consists of system calls, which are executed by the program (PDA) and also cause the property NFAs to change state. 
For example, a call to drop root privilege from a user ID moves the corresponding NFA to a state in which that user ID does not have root privilege.

Furthermore, sometimes an alternative way of encoding a single big specification (an NFA of size $n^k$ for some $k$) might be to decompose it into a product of multiple smaller NFAs (intersection of $k$ NFAs, each of size $n$). There are NFAs that cannot be decomposed in this manner, but, a priori, it might have been possible that if the specification NFA has a nice structure, it could be decomposed into that form and the problem could have been solved faster than the general case. Our results in the next section (Theorem~\ref{thm:general-case-lower-bound}) show that even this restricted setting is as difficult as the original version.

If the number of NFAs, $k$, is unbounded (not fixed), then
intersection non-emptiness becomes complete for \EXPTIME.
(Indeed, a $T(n)$-time Turing machine can be simulated by
an auxiliary pushdown automaton (AuxPDA) with $O(\log T(n))$ bits of storage~\cite{Cook71char}.
Language recognition for such an AuxPDA is reducible to
the intersection non-emptiness of a usual PDA with $O(\log T(n))$ DFAs, each
responsible for one cell of the storage.)
Shortest words in the intersection may in this scenario be
doubly exponentially long~\cite{AiswaryaMS24}.

Our main contributions are to provide conditional lower bounds for the \mainpdanfa non-emptiness problem under different settings. For some of the settings, our lower bounds match the known upper bounds, suggesting that no improvement
over the known algorithm is possible.
These conditional lower bounds are based
on popular hypotheses from the field of fine-grained complexity theory. We now proceed to describe these hypotheses and then our contributions.

\subsection{The $k$-Clique Hypothesis}

One of the central hypotheses from the field of
fine-grained complexity theory is the $k$-Clique hypothesis which is formally stated as follows: For each $\epsilon > 0$,
there is no algorithm that, given a graph $G$ on some $n$ vertices, correctly decides if $G$
has a clique of size $k$ in time $O(n^{\omega k/3 - \epsilon})$. A similar hypothesis is the so-called
combinatorial $k$-Clique hypothesis, which states that there is no combinatorial algorithm that solves $k$-Clique in time $O(n^{k-\epsilon})$ for any $\epsilon > 0$. As mentioned in Section~\ref{sec:intro},  while the notion of combinatorial algorithms is not rigorous, these are roughly taken to be algorithms
that do not use fast matrix multiplication techniques.

	\section{Lower Bounds for \mainpdanfa Non-Emptiness}
\label{sec:lb-general}

Before we present our main result, which is a conditional lower bound for the \mainpdanfa non-emptiness problem, 
let us first recall the known upper bounds for this problem.

\begin{theorem}[Upper Bounds]\label{thm:general-case-upper-bound}
	The \mainpdanfa non-emptiness problem over an input alphabet $\Sigma$ 
	and with $n$ the maximum number of states among all the given machines can be
	solved in time
	\begin{itemize}
		\item $O(n^{2k}|\Sigma| + n^{\omega k})$ if the given PDA is an OCA;
		\item $O(n^{2k}|\Sigma| + n^{3 k})$ if the given PDA is not an OCA or if only combinatorial algorithms are allowed.
	\end{itemize}
\end{theorem}

\begin{proof}
	Let $\pda$ be the given PDA and let $\fsa_1,\dots,\fsa_{k-1}$ be the given NFAs.
	The proof of the theorem follows from two observations: 
	First, PDA non-emptiness can be solved in $O(n^3)$ time~\cite{BEM97} and furthermore, if the PDA is an OCA, then
	it can be solved in $O(n^\omega)$ time~\cite{HansenKP21}. Second, for any PDA $P$ and any NFA $A$, in $O(n^4|\Sigma|)$ time, we can construct a PDA with $O(n^2)$ states and the same stack alphabet as $P$, such that this new PDA recognizes $\langof{P} \cap \langof{A}$. This is the usual Cartesian product construction between a PDA and an NFA (see, e.g., Hopcroft, Motwani, and Ullman's textbook~\cite[Section~7.3.4]{HopcroftMotwaniUllman}). Hence, by repeatedly doing the product construction, in time $O(n^{2k}|\Sigma|)$, we can construct 
	a PDA with $O(n^k)$ states having the same stack alphabet as $P$ which recognizes $\langof{\pda} \cap \langof{\fsa_1} \cap \ldots \cap \langof{\fsa_{k-1}}$. Then, we can perform a non-emptiness check on this PDA in time $O(n^{3k})$
	or in time $O(n^{\omega k})$ if it is an OCA.
\end{proof}


The above algorithm is essentially the best one that we know for this problem and no 
polynomial improvements have been made for this problem. 
We provide a conditional lower bound for this problem,
showcasing the difficulty of any improvement.
 

\begin{theorem}[Lower Bounds]\label{thm:general-case-lower-bound}
	If the $3k$-Clique hypothesis is true, 
	the \mainpdanfa non-emptiness problem over an input
	alphabet $\Sigma$ with $n$ the maximum number of states among all the given machines 
	cannot be solved in time 
	\begin{itemize}
		\item $O((n^{(\omega-1)k} |\Sigma| + n^{\omega k})^{1-\epsilon})$
		for any $\epsilon > 0$;
		\item $O((n^{2 k} |\Sigma| + n^{3 k})^{1-\epsilon})$
		for any $\epsilon > 0$, if only combinatorial algorithms are allowed.
	\end{itemize}
	Both lower bounds already hold when the given PDA is a DOCA
	and all the NFAs are DFAs.
\end{theorem}

This theorem provides a tight lower bound for combinatorial algorithms:
no improvements using such algorithms are possible unless the combinatorial $3k$-Clique hypothesis is false. 

\begin{remark}
Our lower bounds admit multivariate counterparts, where the number of states of PDA and NFAs
has different orders of magnitude.
Suppose in the input to the problem the PDA
has at most~$p$ states and each NFA has at most~$q$ states.
The existing algorithm from Theorem~\ref{thm:general-case-upper-bound} delivers
the upper bound $O(s^2 |\Sigma| + s^3)$, where $s = p q^{k-1}$ is the number of states in the
product automaton.
Thus, our lower bound implies that there is no combinatorial algorithm for this problem
running in time
\begin{equation*}
O\bigl(((p q^{k-1})^2 |\Sigma| + (p q^{k-1})^3)^{1 - \eps}\bigr)
\end{equation*}
for any $\eps > 0$. Indeed, even in the regime $p = q = n$ such an algorithm would
contradict Theorem~\ref{thm:general-case-lower-bound}.
Similar conclusions can be drawn from other lower bounds that we prove in this paper.
We will henceforth not make them explicit.
\end{remark}

\subsection{Proof Idea of Theorem~\ref{thm:general-case-lower-bound}}\label{subsec:proof-idea-general-case-lower-bound}
Before we proceed to describe the formal proof of Theorem~\ref{thm:general-case-lower-bound}, we give the main ideas and intuitions behind it.

\paragraph*{Recall: Triangle finding.}
The main idea behind this proof is a generalization of the idea used to detect triangles, i.e., 3-cliques in a graph by means of an OCA. Let us first recall this idea. Let $G$ be some graph and without loss
of generality, let the nodes of this graph be $\{0,1,\dots,n-1\}$. Furthermore, we 
can assume that the graph contains no self-loops, since removing self-loops does not destroy the property
of having a triangle.
We can check for the existence of a triangle in this graph by using the following 
DOCA $\mach$:
First $\mach$ reads a letter corresponding to some node $a$, moves to a state labelled by $(a,0)$ and increments the counter by $a$.
Then it reads a letter corresponding to some node $b$ and moves from $(a,0)$ to $(b,1)$ only if $b$ is a neighbor of $a$.
(Checking that $b$ is a neighbor of $a$ can be hard-coded into the transitions of $\mach$.) Then, from $(b,1)$, it reads a letter corresponding to 
some node $c$ and moves from $(b,1)$ to $(c,2)$ only if $c$ is a neighbor of $b$.
Then, it reads some letter $d$ and moves from $(c,2)$ to $(d,3)$ only if $d$ is a neighbor of $c$.
Finally, from $(d,3)$, it reads any letter and decrements the counter by $d$. 
If the counter is 0 at the end, then $\mach$ accepts, otherwise it rejects.

Note that any triangle $\{a,b,c\}$ in the graph can be converted into an accepting
path in $\mach$ by inputting the sequence $a, b, c, a$. 
On the other hand, any accepting path in $\mach$ is of the form: $(a,0), (b,1), (c,2), (d,3)$, where $(a,b), (b,c), (c,d)$ are edges in $G$. Further the only updates that happen along the way to the counter are an increment by $a$ at the beginning, and a
decrement by $d$ at the end. Recall that for a run to be accepted, the counter must be 0 at the end,
and so it follows that $d = a$ and so 
$\{a,b,c\}$ is a triangle in $G$.
Hence the language of $\mach$ is non-empty if and only if $G$
has a triangle. \bala{Commented out here}


We will now expand upon this idea and prove the lower bound for the general case.
More precisely, given a graph $G$ on $n$ nodes (without self-loops),
we will use a DOCA $\mach_0$, $k-1$ many DFAs $\mach_1,\dots,\mach_{k-1}$ (where $\mach_0,\mach_1,\dots,\mach_{k-1}$ will all have $O(n)$ states) and an input alphabet of size $O(n^k)$ to detect $3k$-cliques in graphs. 

\paragraph*{Finding three $k$-cliques.}
We can think of a $3k$-clique as 3 different $k$-cliques $C_1, C_2, C_3$ of size $k$ each
such that every node in each $C_i$ is connected to every node in each $C_j$ for $i \neq j$, 
i.e., $C_i \cup C_j$ for $i \neq j$ is a $2k$-clique. Our construction attempts to find such $k$-cliques
in the following manner: First, it finds a $k$-tuple of nodes $C_1 := (v_0,\dots,v_{k-1})$, stores each node of $C_1$ in one of the machines and also stores $C_1$ as a whole in the counter
of the OCA \emph{in a unique way}. Some questions arise at this point.

\textbf{Q: } How can we uniquely store a tuple of nodes of size $k$ as a single number?

\textbf{A: } We map each $k$-tuple of nodes $(v_0,\dots,v_{k-1})$ to the number
$n^{k-1} v_{k-1} + n^{k-2} v_{k-2} + \dots + v_0$. Note that no two tuples are mapped to the same number. \qaqed 

\textbf{Q: } The above representation can lead to numbers as high as $n^k - 1$. However, each machine
can only have $O(n)$ states. How can we increment a counter to that high a value?

\textbf{A: } For this purpose, we construct gadgets that serve as a base-$n$ counter over $k$ digits.
These gadgets will have two important properties:  
First, each gadget $G_i$ will have exactly $n$ states, one for each number from 0 to $n-1$, with transitions
which are either self-loops or only taking place between successor states (modulo $n$).
The second property is that, for each $G_i$,
a transition between successor states in $G_i$ can occur if and only if a run of length $n$ traversing all the states occurs in $G_{i-1}$. This means that if we execute the gadgets $G_0,G_{1},\dots,G_{i}$ and at some
point $G_i$ stops at the state $v_i$ and $G_{i-1},\dots,G_{0}$ all stop at the state 0,
then we have executed a path of length 
exactly $n^{i} v_i$. 

\dch{partially reinstated gadget desc}
Formally,
the input letters of each gadget $G_i$ will be $\{\#_0,\dots,\#_{k-1}\}$. Each state $j \in \{0,1,\dots,n-1\}$ in $G_i$
will stay at $j$ if it reads any letter from $\#_0,\dots,\#_{i-1}$, move to $j+1$ if $j < n-1$ and it reads $\#_i$
and finally move to $0$ if $j = n-1$ and it reads any letter from $\#_{i+1},\dots,\#_{k-1}$.
The start state of each gadget is $0$.

Suppose we now take $k$ copies of our gadgets and execute the 
first copy of $G_0,G_1,\dots,G_{k-1}$ until $G_{k-1}$ reaches some state $v_{k-1}$ and
all the other gadgets reach 0. At this point, suppose we stop
the first copy by reading a special letter, which can only be read when $G_0,\dots,G_{k-2}$
are at the state 0. (If read from any other state, the machines will move to a rejecting sink state.)
After reading this special letter, we ``freeze'' the value $v_{k-1}$ in $G_{k-1}$ (i.e., we will always remember this value in the remaining copies of $G_{k-1}$) and then execute the second copy of the gadgets $G_0,G_1,\dots,G_{k-2}$.
Then, once the second copy of $G_{k-2}$ reaches some state $v_{k-2}$ and all the other gadgets reach 0,
we stop the second copy by reading another special letter, freeze the value
$v_{k-2}$ in $G_{k-2}$ and move on to the third copies of $G_0,G_1,\dots,G_{k-3}$ and so on.
Hence, once we have finished executing all the $k$ copies, we must have a run 
of length exactly $n^{k-1} v_{k-1} + n^{k-2} v_{k-2} + \dots + v_0$ for some $v_{k-1}, \dots, v_0$
that are stored in the last copies of the gadgets. So, if we incremented
the counter every time a step is executed, at the end, the counter value
will be exactly $n^{k-1} v_{k-1} + n^{k-2} v_{k-2} + \dots + v_0$.\qaqed

Having found this tuple $C_1 := (v_0,\dots,v_{k-1})$, we now check that it is 
a $k$-clique of the graph $G$.

\textbf{Q:} How can we verify that $C_1$ is indeed a $k$-clique of $G$?

\textbf{A:} Recall that when the collection is found, each machine stores one node of $C_1$.
Furthermore, we are allowed to have an input alphabet of size $n^k$. Hence, for each 
$k$-clique we will have a letter in our input alphabet. Then, we force the $i^{th}$ machine
(which stores $v_i$) to read one of these letters from its current state only if the 
$i^{th}$ node in this letter is $v_i$. (If it reads some other letter, it will move to a rejecting sink state.) This ensures that if all the machines successfully read some letter,
then $C_1$ is a $k$-clique. \qaqed

Having now found a $k$-clique $C_1$, we now find another $k$-clique $C_2$ such that $C_1 \cup C_2$ is a $2k$-clique and store each node of $C_2$ in one of the machines.

\textbf{Q:} How can we find such a $C_2$?

\textbf{A:} Each node of $C_1$ is stored in some machine. Now, 
we force the $i^{th}$ machine to read a letter corresponding to some $k$-clique $C_2$ only if 
the node stored in the $i^{th}$ machine is a neighbor of every node in $C_2$. If this is indeed
the case, then the $i^{th}$ machine forgets its current node and starts storing the $i^{th}$ node of $C_2$.
Note that since no self-loops are present in $G$, if all the machines successfully read the same letter corresponding to some $k$-clique $C_2$, then we are guaranteed that $C_1 \cap C_2 = \emptyset$, 
	$C_1 \cup C_2$ is a $2k$-clique and all the machines now store nodes of $C_2$. \qaqed

Having now found a $k$-clique $C_2$, we now find another $k$-clique $C_3$ such that $C_2 \cup C_3$
is a $2k$-clique by the same method as before. Then, we once again find another $k$-clique
$C_1'$ such that $C_3 \cup C_1'$ is a $2k$-clique. Now, if we verify that $C_1 = C_1'$,
then we have successfully found a $3k$-clique.

\textbf{Q:} How can we verify that $C_1 = C_1'$?

\textbf{A:} Recall that after having incremented the counter to uniquely store $C_1$, we have not modified
it at all. Hence, the current value of the counter is the encoding of $C_1$. So, if 
we decrement the counter by the encoding of $C_1'$ and accept if the counter is zero,
then we would have verified that $C_1 = C_1'$.\qaqed

\textbf{Q:} How do we decrement the counter by the value corresponding to $C_1'$?

\textbf{A: } Recall that to increment the counter to the value of $C_1$, we constructed $k$ copies of the gadgets $G_0,\dots,G_{k-1}$ and executed them till the states in the last copy of each gadget $G_i$ 
stored the $i^{th}$ node of $C_1$. Hence, if we execute the copies of these gadgets \emph{in reverse}, beginning with the $i^{th}$ machine storing the $i^{th}$ node of $C_1'$ and decrement the counter every time we take a step, then this would decrement the counter by exactly the value corresponding to $C_1'$. \qaqed\\

This completes all the main ideas behind the proof of the theorem. All the machines described
above have $O(n)$ states (where the constant depends on $k$, but recall that $k$ is also a fixed constant) and the alphabet size is $O(n^k)$. Furthermore, each of these machines can be constructed in time $O(mn^{k-1})$, where $m$ is the number of edges of $G$.

Now, suppose \mainpdanfa non-emptiness can be solved
in time $O((M^{(\omega - 1)k}|\Sigma| + M^{\omega k})^{1-\epsilon})$ where $M$ is the maximum number of states among all the given machines and $\epsilon$ is some number strictly bigger than $0$.
Then, we can solve the $3k$-Clique problem in time $O(n^{\omega k (1-\epsilon)})$ 
as follows: Given a graph $G$, first construct
the machines $\mach_0,\dots,\mach_{k-1}$ described above in time $O(mn^{k-1})$ and then
solve the  \mainpdanfa non-emptiness problem on this instance. By the above reduction,
it follows this is a correct algorithm for solving the $3k$-clique problem.
The overall time taken for this procedure is $O(mn^{k-1}) + O((n^{(\omega-1)k}|\Sigma| + n^{\omega k})^{1-\epsilon}) = O(mn^{k-1}) + O(n^{\omega k (1 - \epsilon)}) = O(n^{\omega k (1 - \epsilon)})$,\footnote{Technically, if $k = 1$ and $\omega = 2$, then this last equality need not be true (for instance when $m = O(n^2)$). But in that case, we get a different contradiction: If $k = 1$ and $\omega = 2$, then by assumption, this means that PDA non-emptiness (over instances with $\Sigma = \{0,1\}$) can be solved in time $O(M^{2(1-\epsilon)})$ where $M$ is the number of states of the given PDA. However, if the given PDA is dense, i.e., its number of transitions is $O(M^2)$, then any such algorithm cannot even read the complete input and so cannot correctly solve PDA non-emptiness. This case will recur in other parts of the paper as well, and for the sake of brevity, we do not repeat this argument again in those instances.} which contradicts the $3k$-Clique hypothesis.
Similarly, any $O((n^{2k}|\Sigma| + n^{3k})^{1-\epsilon})$ time algorithm for
\mainpdanfa non-emptiness would contradict the combinatorial $3k$-Clique hypothesis.

We now move on to presenting all the formal details of this proof.


\subsection{Proof of Theorem~\ref{thm:general-case-lower-bound}}\label{subsec:proof-general-case-lower-bound}

Let $G$ be a graph (without self-loops) and $3k$ be a given number. 
Without loss of generality, let $\{0,1,\dots,n-1\}$ be the vertices of $G$.
We will now construct a DOCA $\mach_0$ and $k-1$ many DFAs $\mach_1,\mach_2,\dots,\mach_{k-1}$
over a common alphabet $\Sigma$ such that the intersection of $\mach_0, \mach_1, \dots, \mach_{k-1}$ is non-empty
if and only if $G$ has a $3k$-clique.

The high-level idea behind the construction of these machines has been described in Section~\ref{subsec:proof-idea-general-case-lower-bound} and so we concentrate here on the formal aspects. 
We will construct the machines in three stages. 
In the first stage, we will describe the gadgets necessary for uniquely storing a $k$-tuple $C_1 = (v_0,\dots,v_{k-1})$ of nodes of $G$ as a single number into the counter of the DOCA. In the second stage, we will describe the gadgets necessary for finding other $k$-tuples $C_2, C_3, C_1'$ such that
$C_1 \cup C_2, C_2 \cup C_3$ and $C_3 \cup C_1'$ are all $2k$-cliques. In the third stage, we will describe
the gadgets necessary for decrementing the counter of the DOCA by the unique number assigned to $C_1'$, thereby allowing us to verify that $C_1 = C_1'$. Before we describe these three stages, we make a small remark. 

\begin{remark}
	In each of the three stages, we will actually
	construct machines such that for each state $q$ and each input letter $a$, there is at most one outgoing
	transition for the pair $(q,a)$. Strictly speaking, these do not correspond to deterministic machines,
	because determinism mandates that there be exactly one outgoing transition for the pair $(q,a)$.
	However, it is easy to see that any such machine can be converted into a language-equivalent 
	deterministic machine
	by adding a special sink state to which all undefined outgoing transitions are diverted to.
	The reason we do not add this sink state to all of the machines in our construction is purely for 
	expository purposes as it makes the construction and proofs easier to formulate.	
\end{remark}

\subsubsection*{First stage: Machines for storing a $k$-tuple by incrementing the counter}

In this stage, we will introduce machines which will allow us to uniquely store a $k$-tuple of nodes
as a number in the counter of the DOCA. 
As mentioned in Section~\ref{subsec:proof-idea-general-case-lower-bound}, to every $k$-tuple $C = (v_0,v_1,\dots,v_{k-1})$ of nodes,
we can uniquely assign a number $n^{k-1} \cdot v_{k-1} + n^{k-2} \cdot v_{k-2} + \dots + v_0$.
In this stage, we will construct machines which will first force the counter of the DOCA
to reach a value of the above form for some collection of nodes $v_0,\dots,v_{k-1}$.

To this end, we will construct a DOCA $A_0$
and $k-1$ many DFAs $A_1,\dots,A_{k-1}$ in this first stage. 
The common set of input letters for these machines
will be $\#_0,\#_1,\dots,\#_{k-1}$ and $@_0,@_1,\dots,@_{k-1}$. 
Each machine $A_i \in \{A_0,\dots,A_{k-1}\}$ will have the following set of states:
For each $s \in \{0,\dots,n-1\}$ and each $\ell \in \{0,\dots,k-1\}$, 
$A_i$ will have a state $(s,\ell)^{i}$. Furthermore, for each $s \in \{0,\dots,n-1\}$, $A_i$ will also have a state $(s,done)^{i}$.
The $i$ in the superscript denotes the machine to which these states belong. The intuition behind these states
are that for each $\ell \in \{0,\dots,k-1\}$, the states $(0,\ell)^{i},(1,\ell)^{i},\dots,(n-1,\ell)^{i}$ correspond to the $(k-\ell)^{th}$ copy of the 
gadget $G_i$ described in Section~\ref{subsec:proof-idea-general-case-lower-bound}. The states $(0,done)^{i},(1,done)^{i},\dots,(n-1,done)^{i}$
denote that we have frozen the $i^{th}$ machine with the values $0, 1, \dots, n-1$ respectively.
For each $i$, the state $(0,k-1)^i$ will be the unique initial state of $A_i$.


Before we describe the transitions of each of these machines, we build some intuition.
Each machine $A_i$ will begin at $(0,k-1)^{i}$. Note that each state of $A_i$
is of the form $(s,p)^{i}$ for some $s \in \{0,\dots,n-1\}$ and some $p \in \{0,\dots,k-1,done\}$.
The first part $s$ will be called the \emph{score} of that state and the second
part $p$ will be called the \emph{phase} of that state.
If the phase is $\ell$ for some $\ell \in \{0,\dots,k-1\}$,
then we say that the state is active and otherwise, we say that it is done. 

The transitions that we will construct will always satisfy the
following property: Suppose while reading a word, the machines $A_0,A_1,\dots,A_{k-1}$
reach states with scores $s_0,\dots,s_{k-1}$ at some point. Then the value of the counter of $A_0$ at that point will be \emph{exactly} equal to $\sum_{0 \le i \le k-1} n^i s_i$. Furthermore, if some machine $A_i$ 
reaches a state of the form $(s_i,done)^{i}$, then this would intuitively mean that
we have finished the increments corresponding to the $n^i$ term in our representation, i.e.,
that we have decided on picking the $i^{th}$ node to be $s_i$. Intuitively, then in order 
to complete our task, we must force each $A_i$ to reach a state of the form 
$(s_i,done)^{i}$, i.e., we must force each $A_i$ to reach a done state.
This we will do by first forcing $A_{k-1}$ to reach a done state, then $A_{k-2}$, then $A_{k-3}$ and so on.
We now proceed to formally describe the transitions.

\paragraph*{Description of $A_0$.}
The machine $A_0$ will have the following transitions: For each active state $(s,\ell)^{0}$,
\begin{itemize}
	\item If it reads $\#_0$ and $s < n-1$, then it increments the counter by 1 and 
	moves to $(s+1,\ell)^{+0}$.
	\item If it reads any one of $\#_1,\#_2,\dots,\#_{k-1}$, and $s = n-1$, then it increments the counter by 1 and moves
	to $(0,\ell)^{+0}$.
	\item If it reads $@_\ell$ and $\ell > 0$ and $s = 0$, then it moves to $(0,\ell-1)^{+0}$.
	\item If it reads $@_\ell$ and $\ell = 0$, then it moves to $(s,done)^{+0}$.
\end{itemize}
See Figure~\ref{fig:A0} for a representation of the machine $A_0$.

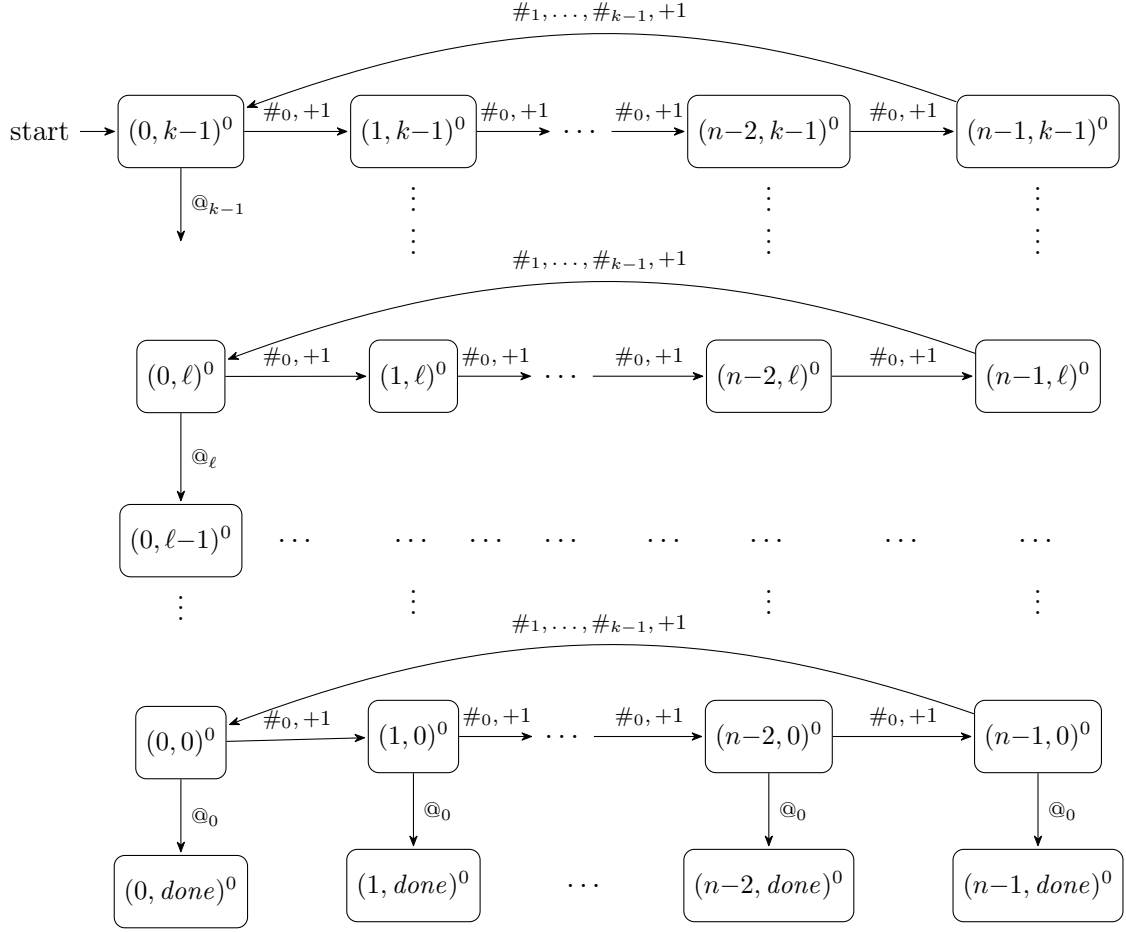
\begin{figure}[t]
	\centering
\usetikzlibrary{automata, positioning, arrows.meta, bending}
\begin{tikzpicture}[
	>={Stealth[round]},
	shorten >=1pt,
	auto,
	node distance=1.6cm and 1.4cm,
	every state/.style={
		draw,
		rounded corners=4pt,
		align=center,
		font=\small,
		rectangle
	},
	trans/.style={font=\scriptsize},
	dots node/.style={draw=none, minimum width=0pt, minimum height=0pt}
	]
	
	\node[initial, state] (s0k) {$(0,k{-}1)^0$};
	\node[state, right=of s0k] (s1k) {$(1,k{-}1)^0$};
	\node[dots node, right=1.0cm of s1k] (dotsRk) {$\dots$};
	\node[state, right=1.0cm of dotsRk] (sn2k) {$(n{-}2,k{-}1)^0$};
	\node[state, right=of sn2k] (sn1k) {$(n{-}1,k{-}1)^0$};
	
	\draw[->] (s0k)   edge[trans] node[above]{$\#_0,{+}1$} (s1k);
	\draw[->] (s1k)   edge[trans] node[above]{$\#_0,{+}1$} (dotsRk);
	\draw[->] (dotsRk) edge[trans] node[above]{$\#_0,{+}1$} (sn2k);
	\draw[->] (sn2k)  edge[trans] node[above]{$\#_0,{+}1$} (sn1k);
	\draw[->] (sn1k) to[bend right=20]
	node[above, trans]{$\#_1,\ldots,\#_{k-1},{+}1$} (s0k);

	\node[dots node, below=1cm of s0k] (vd1_0) {};
	\node[dots node, below=-0.1cm of s1k, align=center] (vd1_1) {$\vdots$\\$\vdots$};
	\node[dots node, below=-0.1cm of sn2k, align=center] (vd1_n2) {$\vdots$\\$\vdots$};
	\node[dots node, below=-0.1cm of sn1k, align=center] (vd1_n1) {$\vdots$\\$\vdots$};
	\draw[->] (s0k)  edge[trans] node[right] {$@_{k-1}$} (vd1_0);
	
	\node[state, below=1cm of vd1_0] (s0l) {$(0,\ell)^0$};
	\node[state, below=1cm of vd1_1] (s1l) {$(1,\ell)^0$};
	\node[dots node, right=1.0cm of s1l] (dotsRl) {$\dots$};
	\node[state, below=1cm of vd1_n2] (sn2l) {$(n{-}2,\ell)^0$};
	\node[state, below=1cm of vd1_n1] (sn1l) {$(n{-}1,\ell)^0$};
	
	\draw[->] (s0l)    edge[trans] node[above]{$\#_0,{+}1$} (s1l);
	\draw[->] (s1l)    edge[trans] node[above]{$\#_0,{+}1$} (dotsRl);
	\draw[->] (dotsRl) edge[trans] node[above]{$\#_0,{+}1$} (sn2l);
	\draw[->] (sn2l)   edge[trans] node[above]{$\#_0,{+}1$} (sn1l);
	\draw[->] (sn1l) to[bend right=20]
	node[above, trans]{$\#_1,\ldots,\#_{k-1},{+}1$} (s0l);
	
	\node[state, below=1.2cm of s0l] (s0lm1) {$(0,\ell{-}1)^0$};
	\draw[->] (s0l) -- node[right, trans]{$@_\ell$} (s0lm1);
	
	\node[dots node] at (s1l |- s0lm1)   (vd2_1)  {$\dots$};
	\node[dots node] at (dotsRl |- s0lm1) {$\dots$};
	\node[dots node] at (sn2l |- s0lm1)  (vd2_n2) {$\dots$};
	\node[dots node] at (sn1l |- s0lm1)  (vd2_n1) {$\dots$};
	\coordinate (midarrow01l) at ($(s0l)!0.5!(s1l)$);
	\coordinate (midarrow1dl) at ($(s1l)!0.5!(dotsRl)$);
	\coordinate (midarrowdnl) at ($(dotsRl)!0.5!(sn2l)$);
	\coordinate (midarrown2n1l) at ($(sn2l)!0.5!(sn1l)$);
	
	\node[dots node] at (midarrow01l   |- s0lm1) {$\dots$};
	\node[dots node] at (midarrow1dl   |- s0lm1) {$\dots$};
	\node[dots node] at (midarrowdnl   |- s0lm1) {$\dots$};
	\node[dots node] at (midarrown2n1l |- s0lm1) {$\dots$};

	\node[dots node, below=-0.1cm of s0lm1, align=center] (vd3_0)  {$\vdots$};
	\node[dots node, below=0.15cm of vd2_1, align=center]  (vd3_1)  {$\vdots$};
	\node[dots node, below=0.15cm of vd2_n2, align=center] (vd3_n2) {$\vdots$};
	\node[dots node, below=0.15cm of vd2_n1, align=center] (vd3_n1) {$\vdots$};
	
	\node[state, below=1cm of vd3_0] (s00) {$(0,0)^0$};
	\node[state, below=1cm of vd3_1] (s10) {$(1,0)^0$};
	\node[dots node, right=1.0cm of s10] (dotsR0) {$\dots$};
	\node[state, below=1cm of vd3_n2] (sn20) {$(n{-}2,0)^0$};
	\node[state, below=1cm of vd3_n1] (sn10) {$(n{-}1,0)^0$};
	
	\draw[->] (s00)    edge[trans] node[above]{$\#_0,{+}1$} (s10);
	\draw[->] (s10)    edge[trans] node[above]{$\#_0,{+}1$} (dotsR0);
	\draw[->] (dotsR0) edge[trans] node[above]{$\#_0,{+}1$} (sn20);
	\draw[->] (sn20)   edge[trans] node[above]{$\#_0,{+}1$} (sn10);
	\draw[->] (sn10) to[bend right=20]
	node[above, trans]{$\#_1,\ldots,\#_{k-1},{+}1$} (s00);

	\node[state, below=1cm of s00]  (s0d)  {$(0,\mathit{done})^0$};
	\node[state, below=1cm of s10]  (s1d)  {$(1,\mathit{done})^0$};
	\node[dots node, right=1.0cm of s1d] (dotsRd) {$\dots$};
	\node[state, below=1cm of sn20] (sn2d) {$(n{-}2,\mathit{done})^0$};
	\node[state, below=1cm of sn10] (sn1d) {$(n{-}1,\mathit{done})^0$};
	
	\draw[->] (s00) -- node[right, trans]{$@_0$} (s0d);
	\draw[->] (s10) -- node[right, trans]{$@_0$} (s1d);
	\draw[->] (sn20) -- node[right, trans]{$@_0$} (sn2d);
	\draw[->] (sn10) -- node[right, trans]{$@_0$} (sn1d);
\end{tikzpicture}
\caption{The DOCA $A_0$. The label $+1$ that appears in (some of) the transitions after a letter means that we increment the counter while taking those transitions.}
\label{fig:A0}
\end{figure}

\paragraph*{Description of $A_i$ for $i \ge 1$.} Each machine $A_i$ for $i \ge 1$ will have the following transitions: For each active state $(s,\ell)^{+i}$,
\begin{itemize}
	\item If it reads any one of $\#_0,\dots,\#_{i-1}$, then it stays at $(s,\ell)^{+i}$.
	\item If it reads $\#_i$, and $s < n-1$, then it moves to $(s+1,\ell)^{+i}$.
	\item If it reads any one of $\#_{i+1},\dots,\#_{k-1}$, and $s = n-1$, then it moves
	to $(0,\ell)^{+i}$.
	\item If it reads $@_\ell$ and $\ell > i$ and $s = 0$, then it moves to $(0,\ell-1)^{+i}$.
	\item If it read $@_\ell$ and $\ell = i$, then it moves to $(s,done)^{+i}$.
\end{itemize}
Furthermore, for each done state $(s,done)^{+i}$, if we read any one of $\#_0,\dots,\#_{i-1},@_{0},
\dots,@_{i-1}$, we stay there. See Figure~\ref{fig:Ai} for a representation of the machine $A_i$.

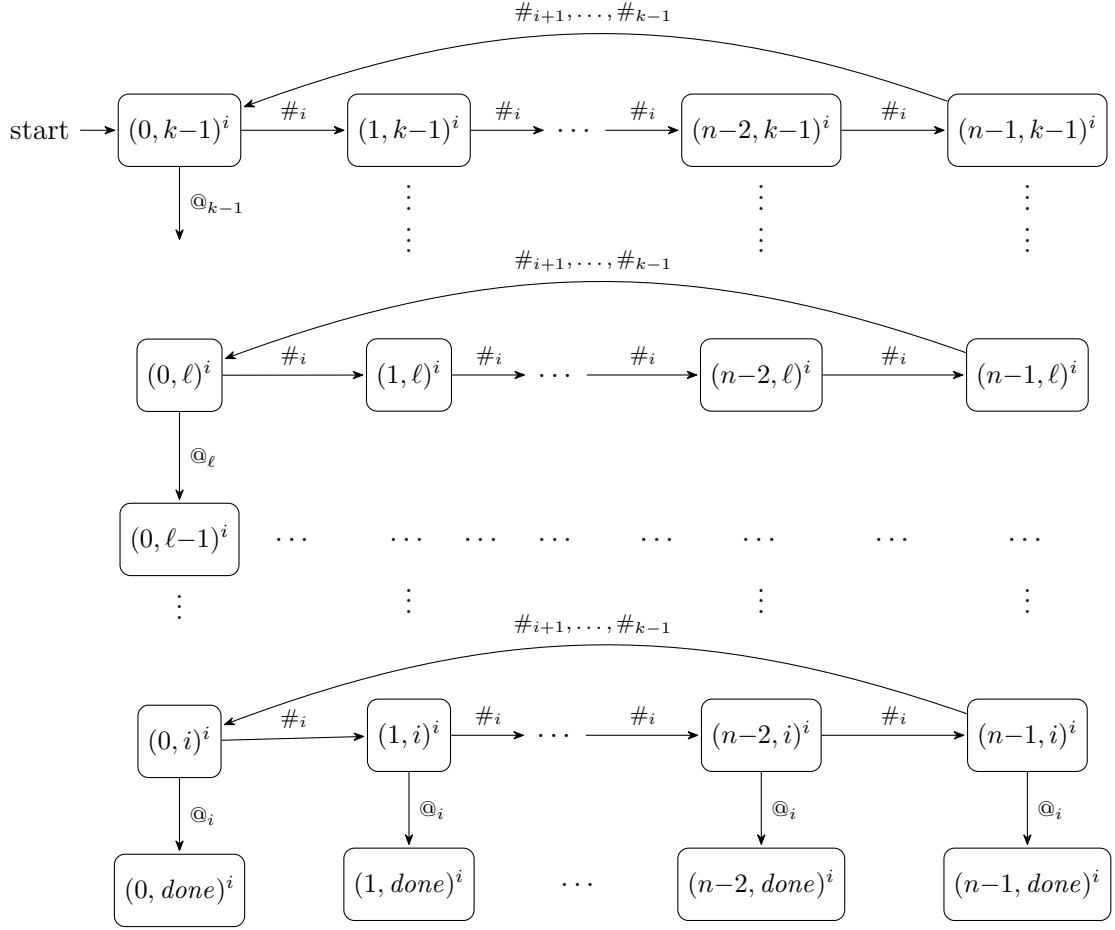
\begin{figure}[t]
	\centering
	\usetikzlibrary{automata, positioning, arrows.meta, bending}
\begin{tikzpicture}[
	>={Stealth[round]},
	shorten >=1pt,
	auto,
	node distance=1.6cm and 1.4cm,
	every state/.style={
		draw,
		rounded corners=4pt,
		align=center,
		font=\small,
		rectangle
	},
	trans/.style={font=\scriptsize},
	dots node/.style={draw=none, minimum width=0pt, minimum height=0pt}
	]
	
	\node[initial, state] (s0k) {$(0,k{-}1)^i$};
	\node[state, right=of s0k] (s1k) {$(1,k{-}1)^i$};
	\node[dots node, right=1.0cm of s1k] (dotsRk) {$\dots$};
	\node[state, right=1.0cm of dotsRk] (sn2k) {$(n{-}2,k{-}1)^i$};
	\node[state, right=of sn2k] (sn1k) {$(n{-}1,k{-}1)^i$};
	
	\draw[->] (s0k)    edge[trans] node[above]{$\#_i$} (s1k);
	\draw[->] (s1k)    edge[trans] node[above]{$\#_i$} (dotsRk);
	\draw[->] (dotsRk) edge[trans] node[above]{$\#_i$} (sn2k);
	\draw[->] (sn2k)   edge[trans] node[above]{$\#_i$} (sn1k);
	\draw[->] (sn1k) to[bend right=20]
	node[above, trans]{$\#_{i+1},\ldots,\#_{k-1}$} (s0k);
	
	\node[dots node, below=1cm of s0k] (vd1_0) {};
	\node[dots node, below=-0.1cm of s1k, align=center] (vd1_1) {$\vdots$\\$\vdots$};
	\node[dots node, below=-0.1cm of sn2k, align=center] (vd1_n2) {$\vdots$\\$\vdots$};
	\node[dots node, below=-0.1cm of sn1k, align=center] (vd1_n1) {$\vdots$\\$\vdots$};
	\draw[->] (s0k) edge[trans] node[right]{$@_{k-1}$} (vd1_0);
	
	\node[state, below=1cm of vd1_0] (s0l) {$(0,\ell)^i$};
	\node[state, below=1cm of vd1_1] (s1l) {$(1,\ell)^i$};
	\node[dots node, right=1.0cm of s1l] (dotsRl) {$\dots$};
	\node[state, below=1cm of vd1_n2] (sn2l) {$(n{-}2,\ell)^i$};
	\node[state, below=1cm of vd1_n1] (sn1l) {$(n{-}1,\ell)^i$};
	
	\draw[->] (s0l)    edge[trans] node[above]{$\#_i$} (s1l);
	\draw[->] (s1l)    edge[trans] node[above]{$\#_i$} (dotsRl);
	\draw[->] (dotsRl) edge[trans] node[above]{$\#_i$} (sn2l);
	\draw[->] (sn2l)   edge[trans] node[above]{$\#_i$} (sn1l);
	\draw[->] (sn1l) to[bend right=20]
	node[above, trans]{$\#_{i+1},\ldots,\#_{k-1}$} (s0l);
	
	\node[state, below=1.2cm of s0l] (s0lm1) {$(0,\ell{-}1)^i$};
	\draw[->] (s0l) -- node[right, trans]{$@_\ell$} (s0lm1);
	
	\node[dots node] at (s1l    |- s0lm1) (vd2_1)  {$\dots$};
	\node[dots node] at (dotsRl |- s0lm1)            {$\dots$};
	\node[dots node] at (sn2l   |- s0lm1) (vd2_n2) {$\dots$};
	\node[dots node] at (sn1l   |- s0lm1) (vd2_n1) {$\dots$};
	\coordinate (midarrow01l)   at ($(s0l)!0.5!(s1l)$);
	\coordinate (midarrow1dl)   at ($(s1l)!0.5!(dotsRl)$);
	\coordinate (midarrowdnl)   at ($(dotsRl)!0.5!(sn2l)$);
	\coordinate (midarrown2n1l) at ($(sn2l)!0.5!(sn1l)$);
	\node[dots node] at (midarrow01l   |- s0lm1) {$\dots$};
	\node[dots node] at (midarrow1dl   |- s0lm1) {$\dots$};
	\node[dots node] at (midarrowdnl   |- s0lm1) {$\dots$};
	\node[dots node] at (midarrown2n1l |- s0lm1) {$\dots$};
	
	\node[dots node, below=-0.1cm of s0lm1, align=center] (vd3_0)  {$\vdots$};
	\node[dots node, below=0.15cm of vd2_1,  align=center] (vd3_1)  {$\vdots$};
	\node[dots node, below=0.15cm of vd2_n2, align=center] (vd3_n2) {$\vdots$};
	\node[dots node, below=0.15cm of vd2_n1, align=center] (vd3_n1) {$\vdots$};
	
	\node[state, below=1cm of vd3_0] (s0i) {$(0,i)^i$};
	\node[state, below=1cm of vd3_1] (s1i) {$(1,i)^i$};
	\node[dots node, right=1.0cm of s1i] (dotsRi) {$\dots$};
	\node[state, below=1cm of vd3_n2] (sn2i) {$(n{-}2,i)^i$};
	\node[state, below=1cm of vd3_n1] (sn1i) {$(n{-}1,i)^i$};
	
	\draw[->] (s0i)    edge[trans] node[above]{$\#_i$} (s1i);
	\draw[->] (s1i)    edge[trans] node[above]{$\#_i$} (dotsRi);
	\draw[->] (dotsRi) edge[trans] node[above]{$\#_i$} (sn2i);
	\draw[->] (sn2i)   edge[trans] node[above]{$\#_i$} (sn1i);
	\draw[->] (sn1i) to[bend right=20]
	node[above, trans]{$\#_{i+1},\ldots,\#_{k-1}$} (s0i);
	
	\node[state, below=1cm of s0i]  (s0d)  {$(0,\mathit{done})^i$};
	\node[state, below=1cm of s1i]  (s1d)  {$(1,\mathit{done})^i$};
	\node[dots node, right=1.0cm of s1d] (dotsRd) {$\dots$};
	\node[state, below=1cm of sn2i] (sn2d) {$(n{-}2,\mathit{done})^i$};
	\node[state, below=1cm of sn1i] (sn1d) {$(n{-}1,\mathit{done})^i$};
	
	\draw[->] (s0i)  -- node[right, trans]{$@_i$} (s0d);
	\draw[->] (s1i)  -- node[right, trans]{$@_i$} (s1d);
	\draw[->] (sn2i) -- node[right, trans]{$@_i$} (sn2d);
	\draw[->] (sn1i) -- node[right, trans]{$@_i$} (sn1d);
	
\end{tikzpicture}
\caption{The NFA $A_i$. In addition to these transitions, in every state, there is also a self-loop transition for reading any one of $\#_0,\dots,\#_{i-1}$. Furthermore, in every done state, there is also a self-loop transition for reading any one of $@_0,\dots,@_{i-1}$}
\label{fig:Ai}
\end{figure}

This completes the description of all the machines. Note that in every machine, for each state $q$ and each letter $a$, there is at most 
one state $q'$ to which the machine can move to while reading the letter $a$
from $q$. Hence, for every machine $A_i$ and every word $w$, there can be at most one run of the machine $A_i$ on reading the word $w$.

\paragraph*{Properties of $A_0,\dots,A_{k-1}$.}

We now prove some properties of these machines and show that they conform with the intuitions that we had described above.
To this end, let us set up some notation. Suppose there exist runs
of the machines $A_0,\dots,A_{k-1}$ on some word $w$. We say that the runs
are active with phase $i$ for some $0 \le i \le k-1$ if at the end of the runs,
\begin{itemize}
	\item The states of the machines $A_0,\dots,A_{i}$ are all active with phase $i$,
	\item The states of the machines $A_{i+1},\dots,A_{k-1}$ are all done.
\end{itemize}
Finally, we say that this collection of runs is perfect if at the end of each run for each machine, the phase of the state of that machine is done.

Having stated these definitions, we state our first result. 
It proves that the collection of runs that we get for any word $w$ is either active or perfect.

\begin{lemma}[Incrementing Counter Soundness Lemma]\label{lem:inc-soundness}
	Let $\rho = (\rho_0,\dots,\rho_{k-1})$ be a collection of runs of the machines $A_0,\dots,A_{k-1}$ along some word $w$.
	Then $\rho$ is either active or perfect.
	Furthermore, if $s_0,\dots,s_{k-1}$ are the scores
	of the states of $A_0,\dots,A_{k-1}$ at the end of $\rho$, then the value of the counter of $A_0$
	at the end of $\rho$ is $\sum_{0 \le \ell \le k-1} n^\ell s_\ell$.
\end{lemma}

\begin{proof}
	We prove this by induction on the length of $w$. For the base case where $w$ is the empty word,
	the claim is easily seen to be true. 
	Suppose we have already proved this claim for some word $w$. We 
	would now like to prove this claim for the word $wa$ where $a$ is some letter. 
	
	To this end, let $\rho' = (\rho'_0,\dots,\rho'_{l-1})$ be a collection of runs of $A_0,\dots,A_{k-1}$ along the word $wa$. 
	Truncating each run in $\rho'$ at the last letter $a$, gives a collection of runs $\rho = (\rho_0,\dots,\rho_{k-1})$ 
	of $A_0,\dots,A_{k-1}$ along the word $w$. We now analyze the different possibilities for $\rho$.
	
	Suppose $\rho$ is perfect. Hence, the state reached in $A_0$ at the end of $\rho_0$ is a done state.
	However, there are no outgoing transitions from 
	any done state of $A_0$, which contradicts the fact that $\rho'_0$ is obtained from $\rho_0$ by letting $A_0$ read the letter $a$. 
	
	Hence $\rho$ must be active and so there is some $i \ge 0$ such that at the end of all of  the runs in $\rho$,
	\begin{itemize}
		\item The states of the machines $A_0,\dots,A_{i}$ are all active with phase $i$,
		\item The states of the machines $A_{i+1},\dots,A_{k-1}$ are all done.
	\end{itemize}
	Let $s_\ell$ be the score of the state of each $A_\ell$ at the end of reading $w$, i.e.,
	at the end of the run $\rho_\ell$. Combined with the above property, this means
	that the state of each $A_\ell$ at the end of  $\rho_\ell$ is either $(s_\ell,i)^\ell$ if $0 \le \ell \le i$ or $(s_\ell,done)^\ell$ if $i+1 \le \ell \le k-1$. Furthermore, by induction hypothesis, the value of the counter of $A_0$ at the end of $\rho_0$ is $\sum_{0 \le \ell \le k-1} n^\ell s_\ell$. We now consider two cases.
	
	\subparagraph*{Case 1: } Suppose there is a first index $e$ in $\{0,\dots,i\}$ such that $s_e < n-1$. We now consider all possible values that the letter $a$ can have.
	\begin{itemize}
		\item Suppose $a \notin \{\#_0,\#_1,\dots,\#_e,@_{i}\}$. By construction, $(s_\ell,i)^e$
		has no transitions for any such letter $a$, which contradicts the fact that $\rho'_e$ is obtained
		from $\rho_e$ by letting $A_e$ read the letter $a$ from  $(s_\ell,i)^e$. Therefore,
		this case is not possible.
		\item Suppose $a \in \{\#_0,\dots,\#_{e-1}\}$. Let $a = \#_x$ for some $x \in \{0,\dots,e-1\}$.
		By assumption on $e$, the machine $A_x$ is at the state $(n-1,i)^{x}$ at the 
		end of the run $\rho_x$. However, from this state there is no outgoing transition
		labelled by $\#_x$, which contradicts the fact that $\rho'_x$ is obtained from $\rho_x$ by letting $A_x$ read the letter $a$ from $(n-1,i)^x$. Therefore, this case is not possible as well.
		\item Suppose $a = \#_e$. 
		Then, upon reading the letter $a$ at the end of the runs in $\rho$, the machines $A_0,\dots,A_{e-1}$
		will move to $(0,i)^{0},\dots,(0,i)^{(e-1)}$, $A_e$ will move to $(s_e+1,i)^{e}$,
		all the other machines will remain where they are and the counter will be increased
		by 1. Hence, it follows that $\rho'$ is an active run with phase $i$.
		
		Now, the scores of the machines before reading $a$ were $\underbrace{n-1,\dots,n-1}_{e-1 \ \text{times}},s_e,s_{e+1},\dots,s_{k-1}$ and the scores
		of the machines after reading $a$ are $\underbrace{0, \dots, 0}_{e-1 \ \text{times}}, s_e + 1, s_{e+1},\dots,s_{k-1}$. 
		Since the counter was incremented by 1 upon reading $a$, using the induction hypothesis for $\rho$, we can now 
		conclude that the induction claim also holds for $\rho'$.

		\item Suppose $a = @_{i}$. Note that if there is some $x \in \{0,\dots,i-1\}$ such that $s_x \neq 0$, then
		there is no outgoing transition from $(s_x,i)^x$ upon reading $a$ in $A_x$. This contradicts the fact that
		$\rho'_x$ is obtained from $\rho_x$ by letting $A_x$ read the letter $a$ from the state  $(s_x,i)^x$. Therefore, $s_0 = s_1 = \dots = s_{i-1} = 0$.
		
		In this case, upon reading the letter $a$ at the end of the runs in $\rho$, the machines $A_0,\dots,A_{i-1}$ will move to the states
		$(0,i-1)^{0},(0,i-1)^{1},\dots,(0,i-1)^{(i-1)}$, 
		$A_{i}$ will move to the state $(s_{i},done)^{i}$ and 
		the states of all the other machines are unchanged. 
		Hence, it follows that $\rho'$ is an active run with phase $i-1$ if $i > 0$ or a perfect run if $i = 0$.
		Note that
		the scores of none of the states have changed and the counter value was undisturbed.
		Hence, we can conclude that the induction claim also holds for $\rho'$.
	\end{itemize}

	This finishes the proof of the induction step for Case 1 and therefore also concludes the proof of the lemma for Case 1.
	
	\subparagraph*{Case 2: } Suppose there is no index $e$ in $\{0,\dots,i\}$ such that $s_e < n-1$. Hence, 
	we have that $s_0 = s_1 = \dots = s_{i} = n-1$. We now consider all possible values that the letter $a$ can have.
	
	\begin{itemize}
		\item Suppose $a \in \{\#_{i+1},\dots,\#_{k-1},@_{i+1}\dots,@_{k-1}\}$. By assumption, the machine $A_{i+1}$ reaches a done state at the end of $\rho_{i+1}$ and done states of $A_{i+1}$ do not have an outgoing transition labelled by any letter from $ \{\#_{i+1},\dots,\#_{k-1},@_{i+1}\dots,@_{k-1}\}$. This contradicts the fact that $\rho'_{i+1}$ is obtained from $\rho_{i+1}$ by letting $A_x$ read the letter $a$. Hence, this case is not possible.
		\item Suppose $a \in \{\#_0,\dots,\#_{i}\}$. Let $a = \#_x$ for some $x \in \{0,\dots,i\}$. By construction, there is no
		outgoing transition from $(s_x,i)^x = (n-1,i)^x$ in $A_x$ labelled by $\#_x$, which leads to a contradiction for the same reason as above.
		\item Suppose $a \in \{@_0,\dots,@_{i-1}\}$. Then, by construction, there is no outgoing transition
		from $(s_{i},i)^{i} = (n-1,i)^{i}$ in $A_{i}$ labelled by $a$, which
		once again leads to a contradiction for the same reason as above.
		\item Hence, $a$ must be $@_{i}$. If $i > 0$, then by construction, there is no outgoing transition from $(s_{i-1},i)^{i-1} = (n-1,i)^{i-1}$ in $A_{i-1}$ labelled
		by $a$, which leads to a contradiction. Hence, $i = 0$, which means that at the end
		of $\rho$, the state of $A_0$ is active (with phase 0) and  all the other machines are at a done state. In this case, reading $a$ at the end of $\rho$, will make $A_0$ move to $(s_0,done)^0 = (n-1,done)^0$ and all the other machines will remain where they are. Hence, it follows that $\rho'$ is a perfect run. Furthermore, since none of the scores of the states nor the counter value was undisturbed, it follows that the induction claim also holds for $\rho'$. 
	\end{itemize}
	
	This finishes the proof of the lemma for Case 2 and thereby completes the proof of the lemma in its entirety.
\end{proof}

We now prove a lemma which acts as a sort of converse to the above lemma.
It shows that for any number $N = \sum_{0 \le i \le k-1} n^i s_i $ with $0 \le s_i \le n-1$ for each $i$,
there is a word $w$ with which we can force the counter value of $A_0$ to reach exactly $N$
whilst simultaneously guiding the machines $A_0,\dots,A_{k-1}$ to states with scores $s_0,s_1,\dots,s_{k-1}$.

\begin{lemma}[Incrementing Counter Completeness Lemma ]\label{lem:inc-completeness}
	Let $N = \sum_{0 \le i \le k-1} n^i s_i $ with $0 \le s_i \le n-1$ for each $i$. Then, there is a word $w$ satisfying the following property:
	There
	is a collection of perfect runs $\rho = (\rho_0,\dots,\rho_{k-1})$ for the machines 
	$A_0,\dots,A_{k-1}$ along the word $w$, with the counter 
	of $A_0$ reaching the value $N$ and each 
	$A_i$ reaching the state $(s_i,done)^{i}$.
\end{lemma}

\begin{proof}
	We prove this by induction on $N$. Note that for the base case of $N = 0$, it can be 
	easily verified that 
	$w := @_{k-1},@_{k-2},\dots,@_0$ satisfies the claim.
	
	Suppose we have proved the claim for some $N < n^k - 1$ and we now want to prove the claim for $N+1$.
	Let $w$ be the word that we obtain for $N$ and let $\rho = (\rho_0,\dots,\rho_{k-1})$ be the perfect collection of runs
	that we get out of reading $w$ from the machines $A_0,\dots,A_{k-1}$. 
	For each $i$, let $\Sigma_i$ denote the set $\{\#_0,\dots,\#_i\}$.
	
	Since $\rho$ is perfect, we claim that $w$ has to be a word of the form $\Sigma_{k-1}^* @_{k-1} \Sigma_{k-2}^* @_{k-2} \dots \Sigma_0^* @_0$. To see this, note the following facts regarding our construction:
	\begin{itemize}
		\item Fact 1: For every $i$, the only transitions that take an active state of $A_i$
		to a done state of $A_i$ are the ones labelled by $@_i$. Hence, the word $w$
		must contain at least one occurrence of $@_i$ for each $i$.
		\item Fact 2: For every $i$, once $A_i$ reaches a done state,
		upon reading any letter $a$, it either stays at the same state
		or has no transition for $a$ (if $a \in \{\#_i,\dots,\#_{k-1},@_i,\dots,@_{k-1}\}$). Hence, after the first occurrence
		of $@_i$, no more occurrences of $\#_i$ or $@_i$ can happen.
		\item Fact 3: For every $i$, the machine $A_i$ has no transitions if it is at an active state and reads $@_j$ for $j < i$. 
	\end{itemize}
	Hence Facts 1 and 2 imply that each $@_i$ occurs exactly once in $w$. Furthermore, Facts 1 and 3 together imply that $@_j$ can appear before $@_i$ if and only if $j > i$. Hence, among the letters in $\{@_{k-1},@_{k-2},\dots,@_0\}$, $@_{k-1}$ must appear first in $w$, then $@_{k-2}$, then $@_{k-3}$, and so on until $@_0$. Finally, Fact 2 implies that in between the occurrences of $@_{i+1}$ and $@_i$, only letters
	from $\Sigma_{i}$ can occur (and after $@_0$ no letter can occur). Putting all of this together, we get that $w$ must be of the form
	$\Sigma_{k-1}^* @_{k-1} \Sigma_{k-2}^* @_{k-2} \dots \Sigma_0^* @_0$. Hence, 
	we let $w = w_{k-1} @_{k-1} w_{k-2} @_{k-2} \dots w_0 @_0$ with each $w_i \in \Sigma_i^*$.
	
	Now, let $N = \sum_{0 \le j \le k-1} n^j s_j$ with each $s_j$ between 0 and $n-1$.
	Since $N < n^k - 1$, there must be a smallest $i$ such that $s_i < n-1$.
	Hence, $N+1 = \sum_{0 \le j < i} n^j \cdot 0 + n^i \cdot (s_i+1) + \sum_{i+1 \le j \le k-1} n^j s_j$.
	
	We now construct a new word $w'$ as 
	$$w' = w_{k-1} @_{k-1}w_{k-2} @_{k-2}\dots w_{i+1} @_{i+1}w_i w_{i-1}
	\dots w_0 \#_i @_i@_{i-1} \dots @_0$$ 
	Compared with $w$, we have pushed the letters $@_i, @_{i-1}, \dots, @_0$ to the
	very end and inserted a $\#_i$ in between $w_0$ and $@_i$. We now claim that 
	$w'$ is the required word for $N+1$ by showing that 
	each machine $A_j$ has a run $\rho_j'$ over the word $w'$
	ending at the state $(0,done)^j$ if $j < i$, $(s_j+1,done)^j$ if $j = i$,
	and $(s_j,done)^j$ if $j > i$.

	To prove this, we use the fact that $A_j$ already has a run $\rho_j$ over the word $w$. 
	Indeed, we split the run $\rho_j$ of $A_j$ over the word $w$ into $2k$ segments as follows:
	For any $0 \le \ell \le k-1$, let $(s_j^\ell,p_j^\ell)^j$ and $(t_j^\ell,q_j^\ell)^j$ 
	be the states reached in $\rho_j$ just before the beginning of the subword $w_\ell$ and just after the end of the subword $w_\ell$ respectively. Also, let $(s_j,done)^j := (s_j^{-1},p_j^{-1})^j$ be the state
	reached at the end of $\rho_j$.	By definition of the word $w$, this means that
	there must be a transition from $(t_j^\ell,q_j^\ell)^j$ to $(s_j^{\ell-1},p_j^{\ell-1})^j$
	reading the letter $@_\ell$	for every $0 \le \ell \le k-1$.

	We now claim that $s_j^{\ell-1} = t_j^{\ell}$, $p_j^\ell = q_j^\ell = \ell$ if $\ell > j$ and $p_j^\ell = q_j^\ell = done$ if $\ell < j$. To this end, we make the following observations regarding our construction:
	\begin{itemize}
		\item Observation 1: If there is a transition in $A_j$ from some state $\alpha$ to some state $\beta$ over a letter $x \in \{\#_0,\dots,\#_{k-1}\}$, then the phase of $\alpha$ and $\beta$ are the same.
		\item Observation 2: If there is a transition in $A_j$ from some state $\alpha$ to some state $\beta$ over a letter $x \in \{@_0,\dots,@_{k-1}\}$, then the score of $\alpha$ and $\beta$ are the same. Furthermore, the phase of $\beta$ is 
		\begin{itemize}
			\item Either one less than the phase of $\alpha$ if $x = @_\ell$ for $\ell > j$.
			\item Or done if $x = @_\ell$ for $\ell \le j$.
		\end{itemize}
	\end{itemize}

	Now, we use these two observations to show our claims. Observation 1 implies that $q_j^\ell = p_j^\ell$ for all $0 \le \ell \le k-1$. Next, Observation 2 implies that
	$s_j^{\ell-1} = t_j^\ell$ and furthermore $p_j^{\ell-1} = p_j^{\ell}-1$
	if $\ell  > j$ and $p_j^{\ell-1} = done$ if $\ell \le j$. Finally, by definition $(s_j^{k-1},p_j^{k-1})^j = (0,k-1)^j$. Putting these three together we get that 
	$s_j^{\ell-1} = t_j^{\ell}$, $p_j^\ell = q_j^\ell = \ell$ if $\ell \ge j$ and $p_j^\ell = q_j^\ell = done$ if $\ell < j$.

	We now use this to show that $A_j$ has a run $\rho_j'$ over the word $w'$. First, since $w'$ has the same prefix as $w$ till up to and including the occurrence of the letter $@_{i+1}$, it follows that $A_j$ has a run up till this prefix of $w'$. Moreover, this run will end at the state $(s_j^i,p_j^i)^j$. To extend the run 
	from here till the end of $w'$, we rely on the following observation, which immediately follows from the construction of $A_j$.
	\begin{itemize}
		\item Observation 3: For some phase $p$, there is a transition from $(s,p)$ to $(s',p)$ in $A_j$ reading a letter $x \in \{\#_0,\dots,\#_{k-1}\}$ if and only if for all phases $p$, there is a transition
		from $(s,p)$ to $(s',p)$ in $A_j$ reading the letter $x$.
	\end{itemize}
	
	We can use this observation as follows: We know that there for every $i \le \ell \le 0$, there
	is a run from $(s_j^\ell,p_j^\ell)^j$ to $(t_j^\ell,p_j^\ell)^j := (s_j^{\ell-1},p_j^\ell)^j$ in $A_j$. By Observation 3,
	this also implies that there is a run from $(s_j^\ell,p_j^i)^j$ to $(s_j^{\ell-1},p_j^i)^j$ in $A_j$.
	Therefore, it follows that from the state $(s_j^i,p_j^i)^j$ in $A_j$ we have a run upon reading the (sub)word $w_iw_{i-1}\dots w_0$ which leads to the state $(t_j^0,p_j^i)^j := (s_j^{-1},p_j^i)^j := (s_j,p_j^i)^j$.

	By our observation above, $p_j^i = done$ if $i < j$ and $p_j^i = i$ if $i \ge j$. 
	Furthermore, by assumption on $i$, $s_j = n-1$ if $i > j$ and $s_j < n-1$ if $i = j$. This implies that from $(s_j,p_j^i)^j$, upon reading the (sub)word $\#_i@_i@_{i-1}\dots@_0$, the machine $A_j$ will, by construction, do any of the following three operations:
	\begin{itemize}
		\item If $i < j$, it will continue to remain at $(s_j,done)^j$ by means of self-loop transitions.
		\item If $i > j$, it will first move to $(0,i)^j$ after reading $\#_i$ 
		and then move to the states
		$(0,i-1)^j,(0,i-2)^j,\dots,(0,j)^j,\underbrace{(0,done)^j,(0,done)^j,\dots,(0,done)^j}_{j \ \text{times}}$.
		\item If $i = j$, it will first move to $(s_j+1,i)^j$ after reading $\#_i$ 
		and then move to the states
		$\underbrace{(s_j+1,done)^j,(s_j+1,done)^j,\dots,(s_j+1,done)^j}_{j \ \text{times}}$.
	\end{itemize}

	It follows that we have the required run $\rho_j'$ for the machine $A_j'$ over the word $w'$.
	Note that the number of letters from $\Sigma_{k-1}$ appearing in $w'$ is one more than the number of letters from $\Sigma_{k-1}$ appearing in $w$. Since the machine $A_0$ increments its counter in a transition if and only if it reads a letter from $\Sigma_{k-1}$ and since the counter value of $A_0$ (on the run $\rho_0$) upon reading the word $w$ was $N$, it follows that the counter value of $A_0$ (on the run $\rho_0'$) upon reading the word $w'$ is $N+1$. This completes the induction step and therefore also concludes the proof of the lemma.
\end{proof}

The proof of these two lemmas also finishes the first stage of the reduction. 
We end this stage with an observation on the size of each $A_i$ as well as the time taken to construct it.
Its proof is immediate from the definition of each $A_i$.

\begin{proposition}\label{prop:size-ai-general}
	The number of states of each $A_i$ is $O(nk)$ and each $A_i$ can be constructed
	in time $O(nk^2)$.
\end{proposition}

\subsubsection*{Second stage: Machines for finding $2k$-cliques}
In this stage, we will construct machines which will help us find $2k$-cliques. 
More precisely, we will construct a DOCA $B_0$ and $k-1$ DFAs $B_1,\dots,B_{k-1}$
which are initially allowed to begin at some collection of states $S$.
The first task of the machines is to verify that this collection $S$ encodes a $k$-clique. From then on,
the machines will collectively ``hop'' from one $k$-clique to another, such that whenever they hop from one $k$-clique $S$
to another $k$-clique $S'$, they will ensure that $S \cup S'$ is by itself a $2k$-clique.

We now look at the formal aspects. The common set of input letters for these machines will be as follows:
For each $k$-clique $S$ of the given graph $G$, there will be a letter which we will also denote by $S$.
Further each machine $B_i$ will have the following set 
of states: For each $s \in \{0,\dots,n-1\}$ and $p \in \{check,\alpha,\beta,\gamma,\zeta\}$,
$B_i$ will have a state $(s,p)^{i}$. As before, the part $s$ will be called the score of the state
and the part $p$ will be called the phase of the state.

We now describe the transitions of each machine $B_i$. 
From a state $(s,check)^{i}$, if $B_i$ reads some letter $S = (v_0,v_1\dots,v_{k-1})$ 
corresponding to some $k$-clique then it moves to $(s,\alpha)^{i}$ if and only if $s = v_i$. 
Further for any state $(s,p)^{i}$ with $p \in \{\alpha,\beta,\gamma\}$, 
if $B_i$ reads some letter $S = (v_0,v_1,\dots,v_{k-1})$ corresponding to some $k$-clique 
then it moves to $(v_i,p')^i$ if and only if $s$ is adjacent to all of the nodes in $S$ (and so $(s,S)$ is a $(k+1)$-clique) and
$p'$ is $\beta$ or $\gamma$ or $\zeta$ depending on whether $p$ is $\alpha$ or $\beta$ or $\gamma$ respectively.

The following proposition immediately follows from the construction of each machine and from the fact
that there are no self-loops in the given graph $G$.

\begin{proposition}
	For any $i$ and any $s, s' \in \{0,\dots,n-1\}$, 
	the machine $B_i$ has a run between the states $(s,check)^i$ and 
	$(s',\zeta)^i$ over a word $w$ if and only if $w = S_1 S_2 S_3 S_1'$
	for some $k$-cliques $S_1, S_2, S_3, S_1'$ such that $s$ is the $i^{th}$ node of $S_1$, $s'$ is the
	$i^{th}$ node of $S_1'$ and 
	the $i^{th}$ node of $S_1, S_2, S_3$ is adjacent to every node in $S_2, S_3, S_4$ respectively.
\end{proposition}

As mentioned before, this gadget simply ``hops'' from one $k$-clique to another.
This intuition is made concrete by the following lemma, which easily follows from the
above proposition.

\begin{lemma}[Clique Finding Lemma]\label{lem:clique-finding}
	For any word $w$ and any $s_0,\dots,s_{k-1},s'_0,\dots,s'_{k-1} \in \{0,\dots,n-1\}$,
	each machine $B_i$ has a run between the states $(s_i,check)^i$ and 
	$(s'_i,\zeta)^i$ over $w$ if and only if $w = S_1 S_2 S_3 S_1'$, 
	for some $k$-cliques $S_1, S_2, S_3, S_1'$ such that $S_1 = (s_0,\dots,s_{k-1}), S_1' = (s_0',\dots,s_{k-1}')$ and $S_1 \cup S_2, S_2 \cup S_3, S_3 \cup S_1'$ are all
	$2k$-cliques.
\end{lemma}

This completes the second stage of the reduction. Similar to the first stage, we end this stage with a remark on the size of each $B_i$.

\begin{proposition}\label{prop:size-bi-general}
	The number of states of each $B_i$ is $O(n)$ and each $B_i$ can be constructed
	in time $O(mn^{k-1})$, where $m$ is the number of edges of the graph $G$.
\end{proposition}

\begin{proof}
	From the definition of each $B_i$, it is immediate that it has at most $O(n)$ states.
	It is also clear from the definition that each $B_i$ can be constructed in time that is linear in its number of states and transitions. So it suffices to count the number of transitions of each $B_i$.
	To this end, note that $B_i$ has a transition from a state $(s,p)^i$ upon reading a letter $S = (v_0,v_1,\dots,v_{k-1})$ only if either $s = v_i$ or $(s,v_i)$ is an edge in $G$. It follows then
	that the number of transitions of each $B_i$ is at most $O(mn^{k-1})$, which completes the proof.
\end{proof}

\subsubsection*{Third stage: Machines for retrieving a $k$-tuple by decrementing the counter}
In this stage, we will introduce machines which will allow us to retrieve a $k$-tuple
which is stored in the counter of the DOCA and check that that tuple is the same as the one that is
currently stored in the states of the machines.
These gadgets will simply be the ``reverse'' of the machines introduced in the first stage.

Formally, we first construct a OCA $C_0$ as follows: 
Let $A_0$ be the DOCA introduced in the first stage. $C_0$ is exactly like $A_0$, except its transitions are reversed. More precisely, $C_0$ has (a copy of) the states of $A_0$, i.e., for every state
of the form $(s,p)^0$ that $A_0$ has, $C_0$ has a state $(\overline{s,p})^0$. Furthermore,
if $((s,p)^0,a,u,(t,q)^0)$ is a transition in $A_0$ 
for some letter $a$ and some $u \in \{-1,0,+1\}$, then $C_0$ has a corresponding ``reverse'' transition
given by $((\overline{t,q})^0,a,-u,(\overline{s,p})^0)$

Similarly, we can now construct a finite-automaton $C_i$ for every $i > 0$. For every state of the form $(s,p)^i$ that $A_i$ has, $C_i$ has a state $(\overline{s,p})^i$. Furthermore,
if $((s,p)^i,a,(t,q)^i)$ is a transition in $A_i$ 
for some letter $a$, then $C_i$ has a corresponding ``reverse'' transition
given by $((\overline{t,q})^i,a,(\overline{s,p})^i)$.

We note that each machine $C_i$ is deterministic. This is because in each machine $A_i$, for each state $(s,p)^{i}$ and each letter $a$, we had at most one \emph{incoming transition} to the state $(s,p)^{+i}$ upon reading $a$. Hence, the machine $C_0$ is a DOCA and the machines $C_1,\dots,C_{k-1}$ are all DFAs.

Note that any run in any machine $C_i$ is the reverse of some run in the machine $A_i$ and vice versa.
Hence, we can define notions of reverse-perfect and reverse-active collections of runs in $C_i$,
as a collection of runs obtained by reversing perfect and active collections of runs in $A_i$.
This observation combined with the Incrementing Counter lemmas (Lemma~\ref{lem:inc-soundness} and Lemma~\ref{lem:inc-completeness}) that we proved
in the first stage immediately implies the following two lemmas.

\begin{lemma}[Decrementing Counter Soundness Lemma]\label{lem:dec-soundness}
	Let $\rho = (\rho_0,\dots,\rho_{k-1})$ be a collection of reverse-perfect runs of the machines 
	$C_0,\dots,C_{k-1}$ along some word $w$ such that $\rho_0$ starts
	at some state $(\overline{s_0,done})^0$ with some counter value $N$ and every other $\rho_i$
	starts at some state $(\overline{s_i,done})^i$. Then $N = \sum_{0 \le i \le k-1} n^i s_i$.
\end{lemma}

\begin{lemma}[Decrementing Counter Completeness Lemma]\label{lem:dec-completeness}
	Let $N = \sum_{0 \le i \le k-1} n^i s_i$ with $0 \le s_i \le n-1$ for each $i$. 
	Then, there is a word $w$ satisfying the following property: There
	is a collection of reverse-perfect runs $\rho = (\rho_0,\dots,\rho_{k-1})$ for the machines 
	$C_0,\dots,C_{k-1}$ along the word $w$, such that $\rho_0$ starts
	at $(\overline{s_0,done})^0$ with counter value $N$ and every other $\rho_i$
	starts at $(\overline{s_i,done})^i$. 
\end{lemma}

This finishes the third stage of the reduction. 

\subsubsection*{Putting the three stages together} Now we put the three stages together and
complete the reduction as follows: For each $i$, we have constructed three different machines $A_i, B_i, C_i$. Let us now combine
them together into one machine $\mach_i$ in the following manner: $\mach_i$ will have
all the states, letters and transitions of $A_i, B_i$ and $C_i$. In addition, it will have
two fresh letters $!,?$ and the following transitions:
\begin{itemize}
	\item From each state of $A_i$ whose phase is done, i.e., each state of the form
	$(s,done)^{i}$, upon reading the letter $!$, $A_i$ moves to the state $(s,check)^i$ of $B_i$.
	\item From each state of $B_i$ whose phase is $\zeta$, i.e., each state of the form
	$(s,\zeta)^{i}$, upon reading the letter $?$, $B_i$ moves to the state $(\overline{s,done})^{i}$ of $C_i$.
\end{itemize}

Let the initial state of each $\mach_i$ be $(0,k-1)^{i}$ and let the final state of each
$\mach_i$ be $(\overline{0,k-1})^{i}$. We now have the following lemma, which is a result of the lemmas that we proved in the previous stages.

\begin{lemma}\label{lem:clique-to-intersection}
	There is a word $w$ such that $w$ is accepted by each $\mach_i$ if and only if
	there is a $3k$-clique in the graph $G$.
\end{lemma}

\begin{proof}
	Suppose there is a $3k$-clique $S$ in the graph $G$. Hence, there are three $k$-cliques
	$S_1, S_2, S_3$ such that $S_1 \cup S_2, S_2 \cup S_3, S_3 \cup S_1$ are each $2k$-cliques.
	Let $S_1 = (s_0,\dots,s_{k-1})$. 
	
	By the Incrementing Counter Completeness lemma (Lemma~\ref{lem:inc-completeness}), there is a word $w_1$ such that each machine $\mach_i$, starting from $(0,k-1)^{i}$ 
	can read the word $w_1$ and reach the state $(s_i,done)^{i}$. 
	Furthermore, at the end of reading $w_1$, the 
	counter value of $\mach_0$ will be $N = \sum_{0 \le i \le k-1} n^i s_i$. 
	
	Afterwards by reading $!$, each machine $\mach_i$ will move from $(s_i,done)^{i}$
	to $(s_i,check)^i$.	After that, by the Clique Finding lemma (Lemma~\ref{lem:clique-finding}), upon reading the word $w_2 = S_1S_2S_3S_1$, each $\mach_i$ will move to $(s_i,\zeta)^i$. Then, by reading $?$, each
	$\mach_i$ will move to $(\overline{s_i,done})^{i}$.
	Finally, by the Decrementing Counter Completeness lemma (Lemma~\ref{lem:dec-completeness}), there is a word $w_3$ such that,
	after reading $w_3$,
	each $\mach_i$ will move to $(\overline{0,k-1})^{i}$, which is the final state of $\mach_i$. 
	Moreover, at the end of reading $w_3$, the counter value of $\mach_0$ will be 0.
	Hence, the word $w := w_1!w_2?w_3$ is accepted by each $\mach_i$. 
	
	Now, suppose there is a word $w$ that is accepted by each $\mach_i$. 
	By construction of $\mach_i$, it follows that $w$ has to be of the form $w_1!w_2?w_3$ for some $w_1, w_2, w_3$.
	Furthermore, for each $i$, the words $w_1, w_2$ and $w_3$ are read entirely in the parts of $\mach_i$ corresponding 
	to $A_i, B_i$ and $C_i$ respectively. Now, we note the following.

	\begin{itemize}
		\item For each $i$, the letter $!$ can only be read from a done state of the machine $A_i$.
		This means that the collection of runs of $A_0,\dots,A_{k-1}$ on the word $w_1$
		must be perfect. Let $(s_i,done)^{i}$ be the state visited
		by $A_i$ after reading $w_1$. By the Incrementing Counter Soundness lemma (Lemma~\ref{lem:inc-soundness}), 
		the value of the counter of $A_0$ at the end of reading $w_1$ is $N = \sum_{0 \le i \le k-1} n^i s_i$.
		
		\item For each $i$, reading the letter $!$ from $(s_i,done)^{i}$ leads to
		the state $(s_i,check)^i$ of $B_i$. Also, the letter $?$ could be read only from
		states whose phase is $\zeta$.
		By the Clique Finding lemma (Lemma~\ref{lem:clique-finding}), it follows that 
		$w_2$ must be of the form $S_1S_2S_3S_1'$ for some $k$-cliques
		$S_1, S_2, S_3, S_1'$ such that the $i^{th}$ node in $S_1$ is $s_i$, $S_1 \cup S_2, S_2 \cup S_3, S_3 \cup S_1'$ are all $2k$-cliques. Furthermore, each $B_i$, upon reading $w_2$ 
		has a run from $(s_i,check)^i$
		to $(s_i',\zeta)^i$ where $s_i'$ is the $i^{th}$ node in $S_1'$.
		\item For each $i$, reading the letter $?$ from $(s_i',\zeta)^i$ leads to the state
		$(\overline{s_i',done})^{i}$ of $C_i$. Note that since $B_0$ does not modify the counter value,
		the counter value after reading $?$ is still $N$.
		
		By assumption, $w$ is accepted by each $\mach_i$ and so this means
		that each $C_i$, starting at $(\overline{s_i',done})^{i}$, upon reading $w_3$, has a run which ends
		at $(\overline{0,k-1})^{i}$. Moreover, the value of the counter of $C_0$ at the end of reading $w_3$ must be 0. By the Decrementing Counter Soundness lemma (Lemma~\ref{lem:dec-soundness}),
		it follows that $N = \sum_{0 \le i \le k-1} n^i s_i'$. Since $N = \sum_{0 \le i \le k-1} n^i s_i$,
		it follows that $s_i = s_i'$ for each $i$ and so the clique $S_1'$ is actually $S_1$.
	\end{itemize}
	
	This then implies that $S_1 \cup S_2 \cup S_3$ is a $3k$-clique and so we are done.
\end{proof}

\subsubsection*{Running time of the reduction} Let us now analyse the running time taken by the reduction.
By Proposition~\ref{prop:size-ai-general}, it follows that each machine $A_i$ can be constructed in  $O(nk^2)$ time where $n$ is the number of nodes of $G$. Similarly each $C_i$ can also be constructed in $O(nk^2)$ time. Furthermore, by Proposition~\ref{prop:size-bi-general}, each $B_i$ can be constructed in time $O(mn^{k-1})$ where $m$ is the number of edges of $G$.
Since $k$ is a constant, by definition of $\mach_i$, it then follows that $\mach_i$ can be constructed in time $O(mn^{k-1})$. By Propositions~\ref{prop:size-ai-general} and~\ref{prop:size-bi-general}, 
it follows that the number of states of each $\mach_i$ is bounded by $O(n)$ and the alphabet size of each $\mach_i$ is bounded by $O(n^{k})$. 

Now, suppose the \mainpdanfa non-emptiness problem (even when the given machines are only a DOCA and $(k-1)$ DFAs) can be solved in time $O((M^{(\omega-1)k} |\Sigma| + M^{\omega k})^{1-\epsilon})$,
where $M$ is the maximum number of states among all the given machines and $\epsilon$ is any number strictly bigger than 0. Then, we can solve the $3k$-Clique problem in $O(n^{\omega k (1-\epsilon)})$ time
as follows: Given a graph $G$, first construct the machines $\mach_0,\dots,\mach_{k-1}$.
Then, run the algorithm for the  \mainpdanfa non-emptiness problem on $\mach_0,\dots,\mach_{k-1}$ and return the answer of this algorithm.
By Lemma~\ref{lem:clique-to-intersection}, this is a correct algorithm for deciding the $3k$-Clique problem.
Furthermore, its running time is $O(mn^{k-1} + n^{\omega k (1-\epsilon)}) = O(n^{\omega k (1-\epsilon)})$, which contradicts the 
$3k$-Clique hypothesis. The same argument proves a similar claim for combinatorial algorithms for the \mainpdanfa problem, thereby allowing us to conclude Theorem~\ref{thm:general-case-lower-bound}.

	\section{\mainpdanfa Non-Emptiness - The Case of Constant Alphabets}
\label{sec:lb-constant}

In the previous section, we gave conditional lower bounds for the \mainpdanfa non-emptiness problem in its full generality. This was accomplished by giving a reduction from the $3k$-Clique problem to the \mainpdanfa non-emptiness problem. Inspecting that reduction, we observe that it produces instances whose alphabet depends upon the given input graph $G$.

In this section, we turn our attention to instances of the \mainpdanfa non-emptiness problem where the input alphabet of the underlying machines is fixed. More precisely, we fix an alphabet $\Sigma$ in advance and only consider instances of the \mainpdanfa non-emptiness problem over this fixed alphabet $\Sigma$. 
As a consequence of Theorem~\ref{thm:general-case-upper-bound}, we get the following upper bound for this case. 

\begin{corollary}
	The \mainpdanfa non-emptiness problem over a fixed input alphabet with $n$ the maximum number of states among all the given machines can be solved in time
	\begin{itemize}
		\item $n^{\omega k}$ if the given PDA is an OCA
		\item $n^{3k}$ if the given PDA is not an OCA or if only combinatorial algorithms are allowed.
	\end{itemize}
\end{corollary}

No polynomial improvement over this algorithm is known in the literature. 
We now provide a lower bound that is almost tight in the case of combinatorial algorithms
and suggests that big improvements over this algorithm are unlikely.

\begin{theorem}\label{thm:constant-alphabet-lower-bound}
	If the $3(k-1)$-Clique hypothesis is true, 
	the \mainpdanfa non-emptiness problem over a fixed input alphabet 
	with $n$ the maximum number of states among all the given machines 
	cannot be solved in time 
	\begin{itemize}
		\item $O(n^{\omega (k-1) -\epsilon})$ for any $\epsilon > 0$.
		\item $O(n^{3(k-1) -\epsilon})$
		for any $\epsilon > 0$, if only combinatorial algorithms are allowed.
	\end{itemize}
	Both lower bounds already hold when the given PDA is a DPDA.
\end{theorem}

We note that the lower bound here is a factor of $n^{3 + (3-\omega)k}$ away in the 
general case and a factor of $n^3$ away in the combinatorial case from the respective upper bounds.

\subsection{Proof idea of Theorem~\ref{thm:constant-alphabet-lower-bound}}
\label{subsec:proof-idea-constant-alphabet-lower-bound}
Before we proceed to describe the formal proof of Theorem~\ref{thm:constant-alphabet-lower-bound}, we give the main ideas and intuitions behind it.

Let us fix a number $3(k-1)$. Let $G$ be some graph (without self-loops) over nodes $\{0,\dots,n-1\}$. 
We will construct a DPDA $\mach_0$ and $k-1$ many NFAs $\mach_1,\dots,\mach_{k-1}$ such that
$G$ has a $3(k-1)$-clique if and only if there is a word $w$ in the intersection of the languages of all $\mach_i$.
For the purposes of presentation, we will first describe this construction with \emph{linear-sized}
input and stack alphabets, i.e., the size of the input and stack alphabets will not be a constant. 
Then, by a careful analysis of the construction, we will convert the linear-sized alphabets to constant-sized alphabets with a logarithmic blowup in the state space. We now proceed
to the construction with the linear-sized alphabets.


The very high-level idea behind these machines is similar to the construction that we saw in Theorem~\ref{thm:general-case-lower-bound}. Intuitively, the machines will first find a collection of $k-1$ vertices $C_1$, check that $C_1$ is a $(k-1)$-clique, then find a collection of $k-1$ vertices $C_2$ and check that every node in $C_1$ is connected
to every node in $C_2$. Then, they will check that $C_2$ is a $(k-1)$-clique, find a collection 
of $k-1$ vertices $C_3$ and check that every node in $C_2$ is connected to every node in $C_3$.
Then they will do a similar procedure with $C_3$ and find another collection $C_1'$.
Then they will finally check that $C_1' = C_1$, which will prove that $C_1 \cup C_2 \cup C_3$ is a $3(k-1)$-clique. 
We stress that while the high-level idea behind this construction and the one given
in Theorem~\ref{thm:general-case-lower-bound} are the same, the actual implementation details
vary significantly. In particular, new ideas are needed in order to store the cliques into the
stack and circumvent the large alphabet size of the construction from Theorem~\ref{thm:general-case-lower-bound}.
We will present these ideas now.

At any given point, each machine will store either a node of the graph $G$ in its state
or store a special symbol $\lozenge$ denoting that it is not storing any node.	
The intuition behind the NFAs $\mach_1,\dots,\mach_{k-1}$ is that, at any point, each $\mach_i$
will store one of the nodes of the collection of $k-1$ nodes that is currently being examined, i.e.,
one of the nodes in either $C_1, C_2, C_3$ or $C_1'$. 
The intuition behind the PDA $\mach_0$ is two-fold. First, the stack of $\mach_0$ will help store
the cliques $C_1, C_2, C_3, C_1'$ of $G$, along with some other information. 
Second, the states of $\mach_0$ should be thought of as a ``scratchpad'', in that it will help store some auxiliary 
information that will be needed for the NFAs. 
\bala{Commented out here}

We will encode each node $i \in \{0,\dots,n-1\}$ by itself, i.e.,
the input and the stack alphabets will have as letters all the numbers between 0 and $n-1$.
In addition to these letters, the input alphabet will also have letters
of the form $\{\overline{i} : 0 \le i \le n-1\}$. The intuition is that,
whenever $i$ is read as an input letter, either the stack does not change or 
$i$ will be pushed onto the stack. Similarly, whenever $\overline{i}$ is read as an input letter, $i$ will be popped from the stack.
Furthermore, the input alphabet will have $\#$ and $@$ as two other additional letters.

We now describe the construction of the machines. 
We recall that each machine will store either a node
of the graph or $\lozenge$ in its state at all times. 
The machines will work together in three different parts and each part will itself comprise 
three different sub-parts. We begin by describing the first part, whose goal
is to check if a collection of nodes $C_1$ is a $(k-1)$-clique and if so,
find another collection $C_2$ such that every node in $C_1$ is connected to every node in $C_2$.
This is done in three sub-parts.

\paragraph*{Part 1, Sub-Part I: The Setup.} In the first sub-part, we will store $k-1$ nodes in the stack
of $\mach_0$ in a specific manner. 
The PDA $\mach_0$ begins this sub-part by remembering $\lozenge$ in its initial state.
Further, each NFA $\mach_i$ begins by remembering some node $x_i$ in its state. (This can be thought of as non-deterministically selecting a state for each $\mach_i$ with some node $x_i$ stored
in that state; later on we will see how this restriction can be removed). The goal of the 
first sub-part is to setup the stack in a specific way so that each $x_i$ is pushed into the stack 
exactly $i-1$ times.
This is done in the following manner. Since $x_2$ is stored in the state of $\mach_2$, 
we can force the input letter that is read at this point to be $x_2$. Indeed,
we only have to create a copy of the current state and have exactly one transition
which leads from the original state to the copy by reading $x_2$. This 
will ensure that $x_2$ is the only possible input letter that could be read at this point.
We also ensure that upon reading $x_2$, the PDA $\mach_0$ pushes it onto the stack.
Similarly, since $x_3$ is stored in the state of $\mach_3$, we can force the next \emph{two} input letters to be $x_3$, by adding two copies of the current state of $\mach_3$ and appropriate transitions.
We can also ensure that the letter $x_3$ is pushed into the stack twice. In this way,
we can ensure that each letter $x_i$ is pushed into the stack $i-1$ times. 
Once this sub-part is done, each $\mach_i$ will store the node $x_i$ in its state,
$\mach_0$ will store $\lozenge$ in its state, and the stack of $\mach_0$ (from the top) will contain $k-2$ many copies of $x_{k-1}$,
$k-3$ copies of $x_{k-2}$ and so on all the way till one copy of $x_2$.
This completes the first sub-part.

\paragraph*{Part 1, Sub-Part II: The Check. } 
In the second sub-part, we will use the special structure of the nodes $x_1,\dots,x_{k-1}$ that are stored in the stack
to check that these nodes indeed form a $(k-1)$-clique.
This is done in the following manner. Note that, at the end of the first sub-part, the node $x_1$ is stored
in $\mach_1$ and the node $x_{k-1}$ is stored at the top of the stack. Using just
this information, we will devise a gadget that checks that $x_1$ is a neighbor of $x_{k-1}$
in the following way.
Since $x_1$ is stored in $\mach_1$, we can force the next input letter to be $x_1$, similar
to how we did it in the first sub-part. 
Upon reading this input letter, we can make $\mach_0$ remember it in its state. (Hence at this point,
both $\mach_0$ and $\mach_1$ remember $x_1$ and $x_{k-1}$ is at the top of the stack.) 
Now, we force the next input letter to be the special letter $\#$. 
Upon reading $\#$,
$\mach_1$ will move to some neighbor $x'$ of $x_1$. (At this point, $\mach_0$ stores $x_1$, 
$\mach_1$ stores $x'$ and the top of the stack stores $x_{k-1}$.) We can now force the next
input letter to be $\overline{x'}$ and we also force $\mach_0$ to pop $x'$ from the top of the stack.
For both these things to simultaneously happen, it must be the case that $x' = x_{k-1}$ and hence that
$x_1$ and $x_{k-1}$ are neighbors. (If this successfully happens, then at this point, $\mach_0$ stores $x_1$, $\mach_1$ stores $x' = x_{k-1}$
and the stack now contains one fewer $x_{k-1}$ at the top.) Now, we force the input letter to be $\#$,
upon reading which $\mach_1$ moves to a state remembering $\lozenge$. From there, because $\mach_0$
remembers $x_1$, we can force the input letter to be $x_1$, 
upon reading which $\mach_1$ will move back to storing $x_1$ and $\mach_0$ will move back to storing $\lozenge$.
In this way, we have checked that $x_1$ and $x_{k-1}$ are neighbors and the only information that we lost along the way was a copy of $x_{k-1}$ from the stack.

By the structure of the first sub-part, it follows that we now have $k-3$ more
copies of $x_{k-1}$ remaining in the stack of $\mach_0$. Hence, we can now reformulate the same gadget 
from the above paragraph to check that $x_2$ and $x_{k-1}$
are neighbors, $x_3$ and $x_{k-1}$ are neighbors and so on. 

Note that, after exhausting all the copies of $x_{k-1}$ from the stack, we are left with $k-3$ copies of $x_{k-2}$. 
This is then sufficient to check that $x_{k-2}$ is a neighbor of $x_1, x_2, \dots, x_{k-3}$.
Then we do the same check for $x_{k-3}, x_{k-4}$ and so on all the way till $x_2$. 
This ensures that $x_1,\dots,x_{k-1}$ is a $(k-1)$-clique. 
Note that, at the end of this computation, each $\mach_i$ stores $x_i$ and $\mach_0$ stores $\lozenge$.
This completes the second sub-part.

\paragraph*{Part 1, Sub-Part III: The Exploration. } In the third sub-part, we will find $k-1$ more nodes 
$y_1,\dots,y_{k-1}$ and check that each $x_i$ is connected with each $y_j$. 
This is done in the following manner. 
Initially, we read some node $y_1$, store it in the state of $\mach_0$ and push it onto the stack. Since $y_1$ is now stored in $\mach_0$, we can ensure that the next $k-1$ input letters are all $y_1$ and also that all these $k$ input letters are pushed onto the stack. 
Then, by using the gadget from the second sub-part, we can check that
$x_1, x_2, \dots, x_{k-1}$ are all neighbors of $y_1$. By construction of this gadget, at the end of this check, each $\mach_i$ will still store $x_i$, $\mach_0$ will store $\lozenge$ and
the stack will contain one copy of $y_1$ (since we pushed $k$ many copies of $y_1$ and
only popped $k-1$ many copies).
We now repeat what we did before to read another node $y_2$ for $k$ times
and ensure that $x_1,\dots,x_{k-1}$ are all neighbors with $y_2$. 
This will end with the stack containing one copy of $y_2$ and then one copy of $y_1$.
Continuing this we can get $y_3,\dots,y_{k-1}$ such that each $x_i$ is a neighbor of each $y_j$,
and the stack contains one copy of $y_{k-1}$, one copy of $y_{k-2}$ and 
so on all the way till $y_1$. 

Now, by popping the nodes on the stack, we can ensure that the next $k-1$ input letters
are $y_{k-1},\dots,y_2, y_1$ in that order. 
While popping $y_i$ (which happens upon reading $\overline{y_i}$), we will store $y_i$ in
the state of $\mach_i$. Hence, at this point, each $\mach_i$ will store $y_i$ and $\mach_0$ will
store $\lozenge$. This completes the third sub-part and also the first part.

\paragraph*{Parts 2 and 3.} At the beginning of the first part, we started with a node $x_i$ in the state of each $\mach_i$. At the end of the first part, we have ensured that $x_1,\dots,x_{k-1}$ is a $(k-1)$-clique, found $k-1$ more nodes $y_1,\dots,y_{k-1}$ such that 
each $x_i$ is a neighbor of each $y_j$ and stored each $y_i$ in $\mach_i$.

The second and third parts are obtained by repeating the same procedure as the first part from 
where it stopped. More precisely, in the second part, the machines will check that 
$y_1,\dots,y_{k-1}$ is a $(k-1)$-clique, find $k-1$ more nodes $z_1,\dots,z_{k-1}$
such that each $y_i$ is a neighbor of each $z_j$ and then store each $z_i$ in $\mach_i$. 
Then, in the third part, the machines will check
that $z_1,\dots,z_{k-1}$ is a $(k-1)$-clique, find $k-1$ more nodes $x_1',\dots,x_{k-1}'$
such that each $z_i$ is a neighbor of each $x_j'$ and store $x_i'$ in each $\mach_i$.

Hence, at the end of the third part, each $\mach_i$ stores the node $x_i'$. 
By construction, it would then follow that if each $x_i' = x_i$,
then the nodes $x_1,\dots,x_{k-1},y_1,\dots,y_{k-1},z_1,\dots,z_{k-1}$ together form
a $3(k-1)$-clique. So, under the assumption that in each machine we store the same node 
at the end as the one that 
we started off with, there is a word in the intersection of the languages of all the machines
if and only if the given graph has a $3(k-1)$-clique. A natural question arises
at this point.

\textbf{Q: } How can we get rid of this assumption?

\textbf{A: } Before we begin the first part, 
we add a zeroth part (which we call the prologue) in which we read $k-1$ nodes, check that
the $i^{th}$ node read is the same as the node in the state of $\mach_i$ and push each node onto the stack as we read it.
In this way, the prologue ensures that when we begin the first part
with nodes $x_1,\dots,x_{k-1}$, they are already on the stack.

Then we proceed to execute the first part, second part and third part as mentioned above.
Note that nowhere in any of these parts did we ever need the stack to be empty to make a transition.
Hence, for the execution of these three parts, it does not matter what stack content we began with,
and so, even with the addition of the zeroth part, the execution of these three parts will be exactly the same as described before, except for the following fact: At the end of the three parts, 
each $\mach_i$ stores $x_i'$ and the stack of $\mach_0$ contains the same content as the end
of the prologue.

Now, we add another last part (which we call the epilogue) which reads $k-1$ letters and attempts
to pop the node from the stack corresponding to the input letter as it is read and ensures that the
$i^{th}$ letter that is popped is the same as the node stored in 
$\mach_{k-i}$. This will ensure that the nodes that were put in the stack during the 
prologue are the same as the nodes $x_1',\dots,x_{k-1}'$, which is what we wanted to verify.\qaqed

This completes the construction of our reduction except for the fact that we have linear-sized input and stack alphabets.

\textbf{Q: } How can we convert the linear-sized alphabets into constant-sized alphabets?

\textbf{A: } The idea is to encode each node $i$ using its binary representation.
So, now each node can be represented only by 0's and 1's. Furthermore, in the transitions, 
we replace pushing the node $i$ onto the stack with pushing its most significant bit first (msbf) representation onto the stack. Similarly, we replace popping the node $i$ from the stack with popping its least significant bit first (lsbf) representation from the stack.
When this encoding is done in a naive way, for each transition labelled with some node, 
this would incur an extra logarithmic amount of states. 
Overall, this would then give us an extra $O(m \log n)$ states where $m$ is the number of edges in the graph $G$.
Since $m$ could be $n^2$, this is not efficient enough for our reduction.

We now sketch how to circumvent this naive method with a more efficient procedure resulting
in only an extra $O(n \log n)$ states.
The crucial observation for this reduction in the state-space is the following one, stated here informally: \emph{Every} state $q$ storing a node $i$ in each of the machines 
in our construction obeys one of the following three conditions.
\begin{itemize}
	\item It only allows to 
	either read $i$ or $\overline{i}$ as input (but not both): In this case, we can only read $i$ or $\overline{i}$
	from that state.
	Hence, we only need $O(\log n)$ more states
	to encode the node $i$, i.e., we need $O(\log n)$ more states to check that we are 
	reading $i$ or $\overline{i}$ and then we can non-deterministically choose any of the outgoing transitions from that state.
	\item It allows for reading any letter from $\{i : 1 \le i \le n\}$ or $\{\overline{i} : 1 \le i \le n \}$ (but not both) and it goes to the same state irrespective of which letter is read: In this case even though we can read any node, the final state reached is the same. Hence, we only have to ensure that a valid binary 
	representation of \emph{some node} is read, i.e., some string in $\{0,1\}^{\log n}$ is read, 
	which can be ensured by having $O(\log n)$ more states. Once such a string is read,
	we know that we can move to exactly one state.
	\item It only allows to read $\#$ as input: In this case we do not need any more states,
	as we are only replacing the encoding of the nodes $i$.
\end{itemize}

This means that for any state storing some node, we only need $O(\log n)$ more states.
The second observation for this reduction is that the number of states which 
do not store any node, i.e., store $\lozenge$, is only a function of $k$. Furthermore,
any state that stores $\lozenge$ has only one outgoing transition for each node of the graph.
This means that we can afford to spend $O(\log n)$ states replacing each outgoing transition
from these states and in the end, we would end up spending $O(n \log n)$ states for each state storing $\lozenge$. Since the overall number
of states storing $\lozenge$ is a function of $k$ only, the total number
of states we introduce this way is still $O(n \log n)$ (because $3(k-1)$ was a fixed constant to begin with).\qaqed 

Finally, we end up with a construction in which each of the machines have $O(n \log n)$ states
and the input and stack alphabets are constant-sized. Furthermore, all of these machines
can be constructed in $O(n^2)$ time. Using this, we can then show
that if we can solve the \mainpdanfa non-emptiness problem in $O(M^{\omega (k-1) -\epsilon})$ time
(resp. in $O(M^{3(k-1)-\epsilon})$ time) for some $\epsilon > 0$ where $M$ is the maximum number of states of the given machines,
then $3(k-1)$-clique can be solved in $O(n^{\omega (k-1) -\epsilon/2})$ time
(resp. in $O(n^{3(k-1)-\epsilon/2})$ time). This proves Theorem~\ref{thm:constant-alphabet-lower-bound}.

We now move on to formalizing all of the ideas mentioned here. 


\subsection{Proof of Theorem~\ref{thm:constant-alphabet-lower-bound}}

Let $3(k-1)$ be a fixed number and $G$ be a graph over the nodes $\{0,\dots,n-1\}$ without self-loops.
We will now construct a DPDA $\mach_0$ and $k-1$ many NFAs $\mach_1,\dots,\mach_{k-1}$ over a common alphabet $\Sigma$ (of size linear in $n$) such that the intersection of $\mach_0,\dots,\mach_{k-1}$ is non-empty if and only if $G$ has a $3(k-1)$-clique. 

The high-level construction of these machines has already been discussed in Section~\ref{subsec:proof-idea-constant-alphabet-lower-bound}
and so we concentrate here on the formal aspects. We will construct the machines 
incrementally by using gadgets, each of which will correspond to one specific sub-part or part
mentioned in the proof idea. Then we will put together all the sub-parts and parts to get the
final machines. 

As mentioned in the proof idea, first we will use a linear-sized (input and stack) alphabet and then
describe how to replace it with one that is of constant size. To this end,
the input alphabet will be $\{0,\dots,n-1\} \cup \{\overline{0},\overline{1},\dots,\overline{n-1}\} \cup \{\#,@\}$ and the stack alphabet will be $\{0,\dots,n-1\}$ for all of the gadgets that we will construct.

As mentioned before, each machine $\mach_i$ will be constructed by first constructing gadgets and then composing them together in a specific manner. We will be doing this quite often (corresponding to each sub-part as well as the prologue and the epilogue parts) and hence it can become quite repetitive.
However, this act of composing together gadgets is uniform throughout and hence we define it formally here,
so that it can be reused (repeatedly) in the construction.

\subsubsection*{Gadgets and their composition}

For the purposes of this construction, a gadget to us will simply be any machine (PDA or NFA) $A$ whose states are of the form $(x,c,s)$ where
$x \in \{0,\dots,n-1\} \cup \{\lozenge\}, 0 \le c \le \max(5,k)$ and $s \in S_A$ where $S_A$ is some set called
the \emph{auxiliary set}.
We will often denote an element $(x,c,s)$ with $s \in S_A$ as $(x,c)^s$, which will be called 
the superscript notation.

For a gadget $A$, its value is the largest value $c$ such that a state of the form $(x,c)^s$ appears in $A$. The value of $A$ will be denoted by $val(A)$.
A copy of $A$ is another gadget $B$ which is exactly the same as the gadget $A$ obtained
by renaming the set $S_A$ to some fresh set $S_B$. 

Initial states of a gadget $A$ can only be states of the form $(x,0)^s$ for some $x$ and $s$.
Final states of a gadget $A$ can only be states of the form $(x,val(A))^s$ for some $x$ and $s$.
We will always have the constraint that for any $x$, there is \textbf{exactly one} initial and final state whose first entry is $x$. Hence, given $x$, we can abuse notation, and, for example, say that we consider the initial (resp. final) state $x$ of a gadget $A$ to mean the unique initial (resp. final) state
of $A$ that has $x$ as its first entry. 

We say that we compose a finite sequence of gadgets $A_1,A_2,\dots,A_\ell$ to get another 
gadget $B$ if $B$ is constructed from $A_1,\dots,A_\ell$ by taking all of their states and
transitions and adding the following new transitions: For each $i < \ell$, from each
final state $x$ of $A_i$ we add a transition to the initial state
$x$ of $A_{i+1}$ which reads the input letter $@$. 
The initial (resp. final) states of $B$ will be the initial (resp. final) states of $A_1$ (resp. $A_\ell$).
The composition $B$ intuitively corresponds to first executing $A_1$, then $A_2$ and so on
all the way till $A_\ell$.

Having stated all the necessary definitions regarding gadgets, we now move on to describing the reduction.
As discussed in Section~\ref{subsec:proof-idea-constant-alphabet-lower-bound}, our construction has five parts, a prologue, then parts 1, 2 and 3, and then finally an epilogue. We begin by describing the prologue part.

\subsubsection{Part 0: The Prologue} Recall the description of the prologue part: Each of the $(k-1)$ NFAs
begin at some states storing nodes $x_1,\dots,x_{k-1}$ respectively. The role of this part is to make sure that these nodes are pushed into the stack of the DPDA. To accomplish this, we construct gadgets
$P_0, P_1, \dots, P_{k-1}$ as follows.

Each gadget $P_i$ will have as its states $(x,c,p_i)$ where $x \in \{0,\dots,n-1\} \cup \{\lozenge\}$
and $c \in \{0,\dots,k\}$. Note that in this case the auxiliary set of each $P_i$ is simply $\{p_i\}$
and hence each state (in the superscript notation) is of the form $(x,c)^{p_i}$.

The transitions of $P_0$ are as follows: Upon reading some node $z$ from a state $(\lozenge,c)^{p_0}$ with $c < k-1$, it will push $z$ into the stack and move to $(\lozenge,c+1)^{p_0}$. Intuitively, this gadget
simply pushes $k-1$ nodes into the stack.

The transitions of each $P_i$ with $i > 0$ are as follows: Upon reading some node $z$
from a state $(x,c)^{p_i}$ with $c < k-1$, it will move to $(x,c+1)^{p_i}$ if $c \neq i-1$ or if $c = i-1$ and $x = z$. Intuitively, the machine $P_i$ will not really do anything until
the $i^{th}$ letter is read (which must be a node of the graph) and when that happens, it will check that that node
is exactly the node stored in its state. See Figures~\ref{fig:P0} and~\ref{fig:Pi} for a representation of $P_0$ as well of $P_i$ for any $i > 0$.

\begin{figure}[t]
	\centering
	\usetikzlibrary{automata, positioning, arrows.meta, bending}
	\begin{tikzpicture}[
		>={Stealth[round]},
		shorten >=1pt,
		auto,
		node distance=1.6cm and 1.75cm,
		every state/.style={
			draw, rounded corners=4pt,
			align=center, font=\small, rectangle
		},
		trans/.style={font=\scriptsize},
		dots node/.style={draw=none, minimum width=0pt, minimum height=0pt}
		]
		
		\node[state, initial] (c0)  {$(\lozenge, 0)^{p_0}$};
		\node[state, right=of c0]   (c1)  {$(\lozenge, 1)^{p_0}$};
		\node[dots node, right=of c1] (dots) {$\dots$};
		\node[state, right=of dots]  (ck2) {$(\lozenge, k{-}2)^{p_0}$};
		\node[state, accepting, right=of ck2]   (ck1) {$(\lozenge, k{-}1)^{p_0}$};
		
		\foreach \src/\tgt in {c0/c1, c1/dots, dots/ck2, ck2/ck1}{
			\draw[->] (\src) to[bend left=40]
			node[above, trans]{$0,\ \text{push}(0)$} (\tgt);
			\draw[->] (\src) edge[trans]
			node[above]{$1,\ \text{push}(1)$} (\tgt);
			\draw[->] (\src) to[bend right=40]
			node[below, trans]{$n{-}1,\ \text{push}(n{-}1)$} (\tgt);
		}
		
		\node[dots node] at ($(c0)!0.5!(c1)+(0,-0.25)$) {$\vdots$};
		\node[dots node] at ($(c1)!0.55!(dots)+(0,-0.25)$) {$\vdots$};
		\node[dots node] at ($(dots)!0.45!(ck2)+(0,-0.25)$) {$\vdots$};
		\node[dots node] at ($(ck2)!0.5!(ck1)+(0,-0.25)$) {$\vdots$};	
	\end{tikzpicture}
	\caption{The transitions of the PDA $P_0$ from the initial state $(\lozenge,0)^{p_0}$.}
	\label{fig:P0}
\end{figure}
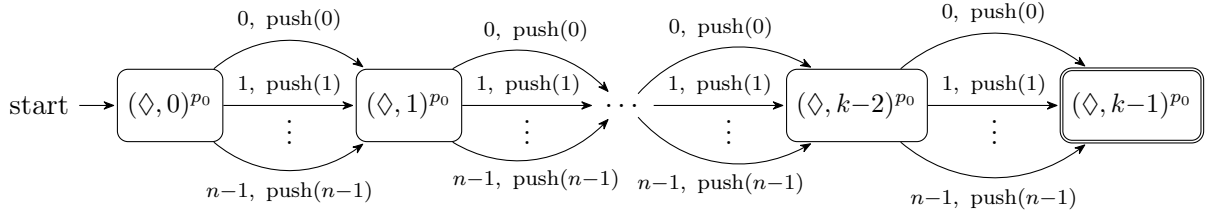

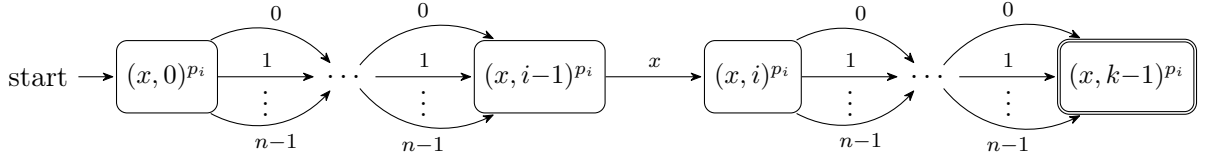
\begin{figure}[t]
	\centering
	\usetikzlibrary{automata, positioning, arrows.meta, bending}
	\begin{tikzpicture}[
		>={Stealth[round]},
		shorten >=1pt,
		auto,
		node distance=1.6cm and 1.8cm,
		every state/.style={
			draw, rounded corners=4pt,
			align=center, font=\small, rectangle
		},
		trans/.style={font=\scriptsize},
		dots node/.style={draw=none, minimum width=0pt, minimum height=0pt}
		]

		\node[state, initial] (ci0)   {$(x, 0)^{p_i}$};
		\node[dots node, right=1.3cm of ci0]   (dotsi1) {$\dots$};
		\node[state, right=1.3cm of dotsi1]    (ci1)   {$(x, i{-}1)^{p_i}$};
		\node[state, right=1.3cm of ci1]      (ci)    {$(x, i)^{p_i}$};
		\node[dots node, right=1.3cm of ci]   (dotsi2) {$\dots$};
		\node[state, accepting, right=1.3cm of dotsi2]    (cik1)   {$(x, k{-}1)^{p_i}$};
		
		\foreach \src/\tgt in {ci0/dotsi1, dotsi1/ci1, ci/dotsi2, dotsi2/cik1}{
			\draw[->] (\src) to[bend left=40]
			node[above, trans]{$0$} (\tgt);
			\draw[->] (\src) edge[trans]
			node[above]{$1$} (\tgt);
			\draw[->] (\src) to[bend right=40]
			node[below, trans]{$n{-}1$} (\tgt);
		}
		
		\node[dots node] at ($(ci0)!0.55!(dotsi1)+(0,-0.25)$) {$\vdots$};
		\node[dots node] at ($(dotsi1)!0.4!(ci1)+(0,-0.25)$) {$\vdots$};
		\node[dots node] at ($(ci)!0.55!(dotsi2)+(0,-0.25)$) {$\vdots$};
		\node[dots node] at ($(dotsi2)!0.4!(cik1)+(0,-0.25)$) {$\vdots$};
		
		\draw[->] (ci1) edge[trans]
		node[above]{$x$} (ci);	
		
	\end{tikzpicture}
	\caption{For any $i > 0$ and any $x \in \{0,\dots,n-1\}$, the transitions of the  NFA $P_i$ from the initial state $(x,0)^{p_i}$.}
	\label{fig:Pi}
\end{figure}

Recall our convention that for each element in $x \in \{0,\dots,n-1\} \cup \{\lozenge\}$, there
will be exactly one initial and final state whose first entry is $x$. (In the case of $P_i$,
the initial state by convention is $(x,0)^{p_i}$ and the final state by convention is $(x,k-1)^{p_i}$).
Now, from the construction of the gadgets, we can deduce the following theorem.
\begin{theorem}[The Prologue Theorem]\label{thm:prologue}
	Let $w$ be some word. Then, there is a run of $P_i$ on $w$ 
	between some initial state $x_i$ and some final state $x_i'$ for every $i$ if and only if 
	$w = x_1'x_2'\dots x_{k-1}'$, $x_0 = x_0' = \lozenge$ and for each $i > 0$, $x_i = x_i'$ is a node.
	Moreover, in every such collection of runs, 
	if the stack of $P_0$ initially stores a word $\gamma$, then at the end it stores the
	word $\gamma x_{1}'\dots x_{k-2}' x_{k-1}'$.
\end{theorem}

\begin{proof}
	Let us prove the right-to-left implication by constructing a run for each $P_i$. 
	Indeed, in this case, the machine $P_0$ upon reading $x_c'$ (for any $c$), moves from $(\lozenge,c-1)^{p_0}$ to $(\lozenge,c)^{p_0}$ and pushes $x_c'$ onto the stack.
	The machine $P_i$ upon reading $x_c$ works as follows: If $c < i$, it moves from $(x_i,c-1)^{p_i}$ to $(x_i,c)^{p_i}$. If $c = i$, since $x_i = x_i'$, 
	it can move from $(x_i,c-1)^{p_i}$ to $(x_i',c)^{p_i}$. If $c > i$, it moves from $(x_i',c-1)^{p_i}$ to $(x_i',c)^{p_i}$.
	This completes the desired construction of a run for each $P_i$. 
	
	Let us now prove the other direction. 
	Suppose for each $i$, there is a run of $P_i$ on $w$ from some initial state $x_i$ to some final state $x_i'$. Note that any transition from any state of the form $(x,c)^{p_i}$ takes it to a state of the form $(x,c+1)^{p_i}$. By construction of each $P_i$, it then follows that $|w| = k-1$, $x_0 = x_0' = \lozenge$ and $x_i = x_i'$ is a node for each $i > 0$. 
	
	Now, by the observation in the previous paragraph, it follows that for each $1 \le i,c \le k$,
	the machine $P_i$, before and after reading the $c^{th}$ letter of $w$, will be in states of the form $(x_i,c-1)^{p_i}$ and $(x_i,c)^{p_i}$ respectively.
	By construction, $P_i$ can take a transition from $(x_i,i-1)^{p_i}$ and reach $(x_i,c)^{p_i}$
	if and only if the letter that it reads is precisely $x_i$. It follows then that 
	the $i^{th}$ letter of $w$ is $x_i = x_i'$. By construction of $P_0$, it is then easily verified
	that the only change in the stack of $P_0$ during this run is that the letters $x_1',\dots,x_{k-1}'$
	are pushed into the stack in that order. This then completes the proof.
\end{proof}

Finally, we also note the following observation, which follows immediately from the construction given above.
\begin{proposition}[Size of $P_i$]\label{prop:size-pi}
	Each $P_i$ has $O(nk)$ states and can be constructed in $O(n^2k)$ time.
\end{proposition}

\subsubsection{Part 1} We will now describe the gadgets for the first part by 
designing gadgets for each of its sub-parts and then composing them together.

\paragraph*{Sub-Part I: The Setup. } Recall the description of this sub-part: After the Prologue, each of the $(k-1)$ NFAs are at some states storing nodes $x_1, \dots, x_{k-1}$ respectively. The objective of this sub-part is to push $x_2$ into the stack once, $x_3$ into the stack twice and so on all the way up till pushing $x_{k-1}$ into the stack $k-2$ times. We now construct gadgets for this purpose, by modifying the gadgets from the Prologue part.

For each $i \in \{0,\dots,k-1\}$ and $j \in \{2,\dots,k-1\}$ we will construct a machine $A_i^j$,
which will have as its states $(x,c)^{a_i^j}$ for $x \in \{0,\dots,n-1\} \cup \{\lozenge\}$ and $c \in \{0,\dots,j-1\}$. Now the transitions of these machines are as follows.

For each $j \in \{2,\dots,k-1\}$, the machine $A_0^j$ is a PDA, which 
upon reading some node $z$ from a state of the form $(\lozenge,c)^{a_i^j}$ with $c < j-1$,
pushes $z$ into the stack and moves to $(\lozenge,c+1)^{a_i^j}$. Intuitively, this machine simply pushes some $j-1$ nodes
into the stack.

For each $i \in \{1,\dots,k-1\}$ and $j \in \{2,\dots,k-1\}$, the machine $A_i^j$ is an NFA, which
upon reading some node $z$ from a state of the form $(x,c)^{a_i^j}$ with $c < j-1$,
moves to $(x,c+1)$ if $i \neq j$ or $i = j$ and $x = z$. Intuitively, the machine $A_i^j$ does not
really do anything unless $i = j$. If $i = j$, it will simply check that
the input consists of exactly $j-1$ letters all of which are exactly the same node that it had
stored in its state at the beginning. See Figures~\ref{fig:A0j},~\ref{fig:Aij-neq} and~\ref{fig:Aij-eq} for a representation of $A_0^j$ as well of $A_i^j$ when $i \neq j$ and $i = j$ respectively.

\begin{figure}[t]
	\centering
	\usetikzlibrary{automata, positioning, arrows.meta, bending}
	\begin{tikzpicture}[
		>={Stealth[round]},
		shorten >=1pt,
		auto,
		node distance=1.6cm and 1.75cm,
		every state/.style={
			draw, rounded corners=4pt,
			align=center, font=\small, rectangle
		},
		trans/.style={font=\scriptsize},
		dots node/.style={draw=none, minimum width=0pt, minimum height=0pt}
		]
		
		\node[state, initial] (c0)  {$(\lozenge, 0)^{a_0^j}$};
		\node[state, right=of c0]   (c1)  {$(\lozenge, 1)^{a_0^j}$};
		\node[dots node, right=of c1] (dots) {$\dots$};
		\node[state, right=of dots]  (ck2) {$(\lozenge, j{-}2)^{a_0^j}$};
		\node[state, accepting, right=of ck2]   (ck1) {$(\lozenge, j{-}1)^{a_0^j}$};
		
		\foreach \src/\tgt in {c0/c1, c1/dots, dots/ck2, ck2/ck1}{
			\draw[->] (\src) to[bend left=40]
			node[above, trans]{$0,\ \text{push}(0)$} (\tgt);
			\draw[->] (\src) edge[trans]
			node[above]{$1,\ \text{push}(1)$} (\tgt);
			\draw[->] (\src) to[bend right=40]
			node[below, trans]{$n{-}1,\ \text{push}(n{-}1)$} (\tgt);
		}
		\node[dots node] at ($(c0)!0.5!(c1)+(0,-0.25)$) {$\vdots$};
		\node[dots node] at ($(c1)!0.55!(dots)+(0,-0.25)$) {$\vdots$};
		\node[dots node] at ($(dots)!0.45!(ck2)+(0,-0.25)$) {$\vdots$};
		\node[dots node] at ($(ck2)!0.5!(ck1)+(0,-0.25)$) {$\vdots$};
	
	\end{tikzpicture}
	\caption{For any $j > 0$, the transitions of the PDA $A_0^j$ from the initial state $(\lozenge,0)^{a_0^j}$).}
	\label{fig:A0j}
\end{figure}
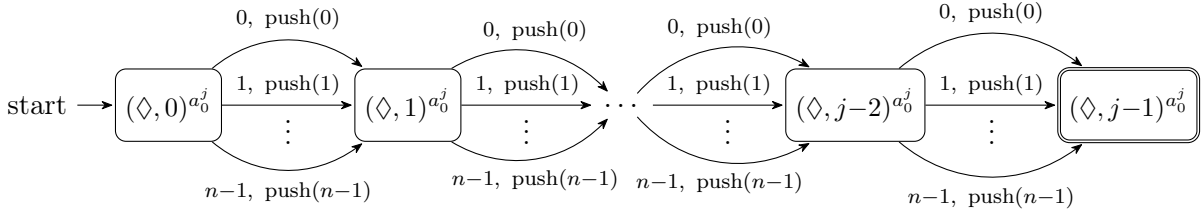
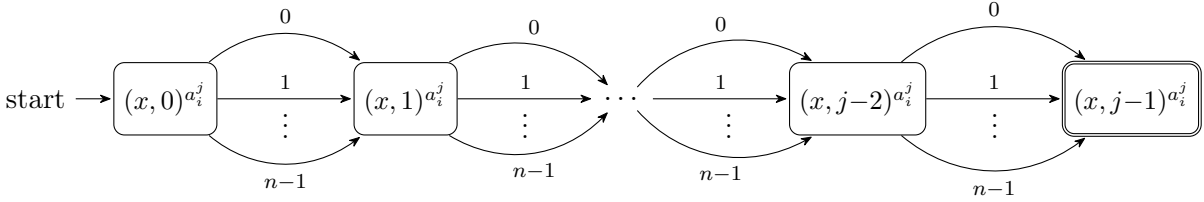
\begin{figure}[t]
	\centering
	\usetikzlibrary{automata, positioning, arrows.meta, bending}
	\begin{tikzpicture}[
		>={Stealth[round]},
		shorten >=1pt,
		auto,
		node distance=1.6cm and 1.8cm,
		every state/.style={
			draw, rounded corners=4pt,
			align=center, font=\small, rectangle
		},
		trans/.style={font=\scriptsize},
		dots node/.style={draw=none, minimum width=0pt, minimum height=0pt}
		]
		
		\node[state, initial] (cx0)  {$(x, 0)^{a_i^j}$};
		\node[state, right=of cx0]   (cx1)  {$(x, 1)^{a_i^j}$};
		\node[dots node, right=of cx1] (xdots) {$\dots$};
		\node[state, right=of xdots]  (cxk2) {$(x, j{-}2)^{a_i^j}$};
		\node[state, accepting, right=of cxk2]   (cxk1) {$(x, j{-}1)^{a_i^j}$};
		
		\foreach \src/\tgt in {cx0/cx1, cx1/xdots, xdots/cxk2, cxk2/cxk1}{
			\draw[->] (\src) to[bend left=40]
			node[above, trans]{$0$} (\tgt);
			\draw[->] (\src) edge[trans]
			node[above]{$1$} (\tgt);
			\draw[->] (\src) to[bend right=40]
			node[below, trans]{$n{-}1$} (\tgt);
		}
		\node[dots node] at ($(cx0)!0.5!(cx1)+(0,-0.25)$) {$\vdots$};
		\node[dots node] at ($(cx1)!0.55!(xdots)+(0,-0.25)$) {$\vdots$};
		\node[dots node] at ($(xdots)!0.45!(cxk2)+(0,-0.25)$) {$\vdots$};
		\node[dots node] at ($(cxk2)!0.5!(cxk1)+(0,-0.25)$) {$\vdots$};
		
	\end{tikzpicture}
	\caption{For any $i, j > 0$ with $i \neq j$ and any $x \in \{0,\dots,n-1\}$, the transitions of the NFA $A_i^j$ from the initial state $(x,0)^{a_i^j}$.}
	\label{fig:Aij-neq}
\end{figure}
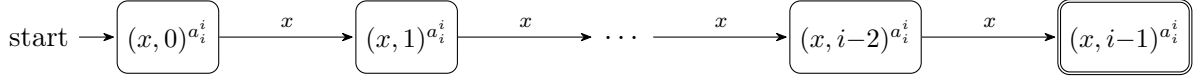
\begin{figure}[t]
	\centering
	\usetikzlibrary{automata, positioning, arrows.meta, bending}
	\begin{tikzpicture}[
		>={Stealth[round]},
		shorten >=1pt,
		auto,
		node distance=1.6cm and 1.8cm,
		every state/.style={
			draw, rounded corners=4pt,
			align=center, font=\small, rectangle
		},
		trans/.style={font=\scriptsize},
		dots node/.style={draw=none, minimum width=0pt, minimum height=0pt}
		]
		
		\node[state, initial, below = 2cm of cx0] (cy0)  {$(x, 0)^{a_i^i}$};
		\node[state, right=of cy0]   (cy1)  {$(x, 1)^{a_i^i}$};
		\node[dots node, right=of cy1] (ydots) {$\dots$};
		\node[state, right=of ydots]  (cyk2) {$(x, i{-}2)^{a_i^i}$};
		\node[state, accepting, right=of cyk2]   (cyk1) {$(x, i{-}1)^{a_i^i}$};
		
		\foreach \src/\tgt in {cy0/cy1, cy1/ydots, ydots/cyk2, cyk2/cyk1}{
			\draw[->] (\src) edge[trans]
			node[above]{$x$} (\tgt);
		}
	\end{tikzpicture}
	\caption{For any $i > 0$ and for any $x \in \{0,\dots,n-1\}$, the transitions of the NFA $A_i^i$, from the initial state $(x,0)^{a_i^i}$.}
	\label{fig:Aij-eq}
\end{figure}

The following lemma is immediate from the construction of the machines.
Its proof is exactly similar to the proof of the Prologue theorem (Theorem~\ref{thm:prologue}).

\begin{lemma}
	Let $w$ be some word and let $j \in \{2,\dots,k-1\}$. Then, there is a run of $A_i^j$ 
	on $w$ starting from some initial state $x_i$ and ending at some final state $x_i'$ for every $i$
	if and only if 
	$w = x_j^{j-1}$, $x_0  =  x_0' = \lozenge$ and for each $i > 0$, $x_i = x_i'$ is a node. 
	Moreover, in every such collection of runs, if the stack of $P_0$ initially stores a word $\gamma$, then
	at the end it stores the word $\gamma x_j^{j-1}$.
\end{lemma}

Now, for each $i$, let us compose the gadgets $A_i^2,A_i^3,\dots,A_i^k$ to get a new gadget
$A_i$. For any collection of nodes $x_1,\dots,x_{k-1}$, let $A(x_1,\dots,x_{k-1})$ be the word
given by $x_2@x_3^2@\dots@x_{k-1}^{k-2}$.
The following theorem now follows by using the definition of composition and repeatedly applying the above lemma. 

\begin{theorem}[The Setup Theorem]\label{thm:setup}
	Let $w$ be some word. Then, there is a run of $A_i$ on $w$ starting from some initial state $x_i$ and ending at some final state $x_i'$ for every $i$ if and only if 
	$w = A(x_1,\dots,x_{k-1})$, $x_0 = x_0' = \lozenge$ and for each $i > 0$, $x_i = x_i'$ is a node.
	Moreover, in every such collection of runs, if the stack of $P_0$ initially stores a word $\gamma$, then
	at the end it stores the word $\gamma x_2 x_3^2 \dots x_{k-2}^{k-3} x_{k-1}^{k-2}$.
\end{theorem}

\begin{proof}
	Let us first prove the right-to-left implication. Suppose $w$ is of the form $x_2@x_3^2@\dots@x_k^{k-1}$, $x_0 = x_0' = \lozenge$ and for each $i > 0$, $x_i = x_i'$ is a node. For each $i$, we can now construct a run of $A_i$ on $w$
	starting from the initial state $x_i$ and ending at the final state $x_i'$ as follows. Let $w_i = x_i^{i-1}$.
	By the previous lemma, for each $A_i^j$ with $j \in \{2,\dots,k\}$, there
	is a run of $w_i$ starting from the initial state $x_i$ of $A_i^j$ and ending with
	the final state $x_i$ of $A_i^j$. By using the definition of composition
	of the gadgets, it follows that we have a run of the desired form for $A_i$.
	
	Let us now prove the other direction. 	Suppose, for all $i$, there is a run of $A_i$ on $w$ starting from some initial state
	$x_i$ and ending at some final state $x_i'$.
	By definition of composition of gadgets, $w$ must be of the form $w_2@w_3@\dots@w_k$ such that for each $j \in \{2,\dots,k\}$,
	$A_i^j$ has an accepting run on $w_i$. Now, using the definition of composition and the previous lemma, we can conclude the proof of this direction.
\end{proof}

Finally, we observe that
\begin{proposition}[Size of $A_i$]\label{prop:size-Ai}
	Each $A_i$ has $O(nk^2)$ states and can be constructed in $O(n^2k^2)$ time.
\end{proposition}

\paragraph*{Sub-Part II: The Check.} Recall the description of this sub-part: After sub-part I,
each of the $(k-1)$ NFAs are at some states storing nodes $x_1,\dots,x_{k-1}$ respectively and the stack of the PDA contains $(k-2)$ copies of $x_{k-1}$ at the top, followed by $(k-3)$ copies of $x_{k-2}$ and so till a single copy of $x_2$. Using such a configuration, the goal of this sub-part is to check that the nodes $x_1,\dots,x_{k-1}$ form a $(k-1)$-clique, i.e.,  the goal of this sub-part is to check that any pair of such nodes are neighbors.  We shall now construct the necessary gadgets to check the neighborhood relation between such nodes. More precisely, given a node stored at the top of the stack and another node which 
is stored in some NFA, these gadgets will check whether these two nodes are neighbors.

For each $i \in \{0,\dots,k-1\}$ and $j \in \{1,\dots,k-1\}$, we will construct
a machine $B_i^j$ which will have as its states $(x,c)^{b_i^j}$ for $x \in \{0,\dots,n-1\} \cup \{\lozenge\}$ and $c \in \{0,1,2,3,4,5\}$. Now the transitions of these machines are as follows.

For each $j \in \{1,\dots,k-1\}$, the machine $B_0^j$ is a PDA which moves according to the following transitions:
\begin{itemize}
	\item From $(\lozenge,0)^{b_0^j}$, upon reading a node $x$, it moves to $(x,1)^{b_0^j}$.
	\item For any node $x$, from $(x,1)^{b_0^j}$, upon reading $\#$,  it moves to $(x,2)^{b_0^j}$.
	\item For any nodes $x, y$, from $(x,2)^{b_0^j}$, upon reading $\overline{y}$, it pops $y$ from the stack and moves to $(x,3)^{b_0^j}$.
	\item For any node $x$, from $(x,3)^{b_0^j}$, upon reading $\#$, it moves to $(x,4)^{b_0^j}$.
	\item Finally, for any node $x$, from $(x,4)^{b_0^j}$, upon reading $x$, it moves to $(\lozenge,5)^{b_0^j}$.
\end{itemize}

For each $i \in \{1,\dots,k-1\}$ and $j \in \{1,\dots,k-1\}$, the machine $B_i^j$ is an NFA 
which moves according to the following transitions: First, if $i \neq j$, then for any $x$ and any $c < 5$,
it moves from $(x,c)^{b_i^j}$ to $(x,c+1)^{b_i^j}$ upon reading any letter. Now, if $i = j$, then
its transitions are given as follows:
\begin{itemize}
	\item For any node $x$, from $(x,0)^{b_i^i}$, upon reading $x$, it moves to $(x,1)^{b_i^i}$.
	\item For any node $x$, from $(x,1)^{b_i^i}$, upon reading $\#$, it moves to $(y,2)^{b_i^i}$ where
	$y$ is any neighbor of $x$ in $G$.
	\item For any node $y$, from $(y,2)^{b_i^i}$, upon reading $\overline{y}$, it moves to $(y,3)^{b_i^i}$.
	\item For any node $y$, from $(y,3)^{b_i^i}$, upon reading $\#$, it moves to $(x',4)^{b_i^i}$ where
	$x'$ is any neighbor of $y$ in $G$.
	\item For any node $x'$, from $(x',4)^{b_i^i}$, upon reading $x'$, it moves to $(x',5)^{b_i^i}$.
\end{itemize}

See Figures~\ref{fig:B0j} and~\ref{fig:Bij-eq} for a representation of $B_0^j$ as well as of $B_i^i$ ($B_i^j$ for $i \neq j$ and $i > 0$ is quite trivial).
Now, for any two nodes $x$ and $y$, let $N(x,y)$ be the word $x\#\overline{y}\#x$.
The following lemma follows from an analysis of the constructed gadgets.

\begin{figure}[t]
	\centering
	\usetikzlibrary{automata, positioning, arrows.meta, bending}
	\begin{tikzpicture}[
		>={Stealth[round]},
		shorten >=1pt,
		auto,
		node distance=2.0cm and 0.9cm,
		every state/.style={
			draw, rounded corners=4pt,
			align=center, font=\small, rectangle,
		},
		trans/.style={font=\scriptsize},
		dots node/.style={draw=none, minimum width=0pt, minimum height=0pt}
		]
		
		\node[state, initial] (c0) {$(\lozenge, 0)^{b_0^j}$};
		
		\node[state, above right=2.0cm and 0.9cm of c0]  (x0c1) {$(0,1)^{b_0^j}$};
		\node[state, right=of x0c1]                (x0c2) {$(0,2)^{b_0^j}$};
		\node[state, right=2.5cm of x0c2]                (x0c3) {$(0,3)^{b_0^j}$};
		\node[state, right=of x0c3]                (x0c4) {$(0,4)^{b_0^j}$};
		
		\node[state, right=of c0]                  (x1c1) {$(1,1)^{b_0^j}$};
		\node[state, right=of x1c1]                (x1c2) {$(1,2)^{b_0^j}$};
		\node[state, right=2.5cm of x1c2]                (x1c3) {$(1,3)^{b_0^j}$};
		\node[state, right=of x1c3]                (x1c4) {$(1,4)^{b_0^j}$};
		
		\node[dots node, below=0.5cm of x1c1] (vd1) {$\vdots$};
		\node[dots node, below=0.5cm of x1c2] (vd2) {$\vdots$};
		\node[dots node, below=0.5cm of x1c3] (vd3) {$\vdots$};
		\node[dots node, below=0.5cm of x1c4] (vd4) {$\vdots$};
		
		\node[state, below right=2.0cm and 0.8cm of c0]  (xnc1) {\footnotesize$(n{-}1,1)^{b_0^j}$};
		\node[state, right=0.7cm of xnc1]                (xnc2) {\footnotesize$(n{-}1,2)^{b_0^j}$};
		\node[state, right=2.0cm of xnc2]                (xnc3) {\footnotesize$(n{-}1,3)^{b_0^j}$};
		\node[state, right=0.7cm of xnc3]                (xnc4) {\footnotesize$(n{-}1,4)^{b_0^j}$};
		
		\node[state, below right=2.0 cm and 0.9cm of x0c4, accepting] (c5) {$(\lozenge, 5)^{b_0^j}$};
		
		\draw[->] (c0) edge[trans] node[above left]{$0$}     (x0c1);
		\draw[->] (c0) edge[trans] node[above]{$1$}           (x1c1);
		\draw[->] (c0) edge[trans] node[below left]{$n{-}1$} (xnc1);
		
		\draw[->] (x0c1) edge[trans] node[above]{$\#$} (x0c2);
		\draw[->] (x0c2) to[bend left=30]
		node[above, trans]{$\bar{0},\ \text{pop}(0)$} (x0c3);
		\draw[->] (x0c2) edge[trans]
		node[above]{$\bar{1},\ \text{pop}(1)$} (x0c3);
		\draw[->] (x0c2) to[bend right=30]
		node[below, trans]{$\overline{n{-}1},\ \text{pop}(n{-}1)$} (x0c3);
		\node[dots node] at ($(x0c2)!0.5!(x0c3)+(0,-0.25)$) {$\vdots$};
		\draw[->] (x0c3) edge[trans] node[above]{$\#$} (x0c4);
		\draw[->] (x0c4) edge[trans] node[above right]{$0$} (c5);
		
		\draw[->] (x1c1) edge[trans] node[above]{$\#$} (x1c2);
		\draw[->] (x1c2) to[bend left=30]
		node[above, trans]{$\bar{0},\ \text{pop}(0)$} (x1c3);
		\draw[->] (x1c2) edge[trans]
		node[above]{$\bar{1},\ \text{pop}(1)$} (x1c3);
		\draw[->] (x1c2) to[bend right=30]
		node[below, trans]{$\overline{n{-}1},\ \text{pop}(n{-}1)$} (x1c3);
		\node[dots node] at ($(x1c2)!0.5!(x1c3)+(0,-0.25)$) {$\vdots$};
		\draw[->] (x1c3) edge[trans] node[above]{$\#$} (x1c4);
		\draw[->] (x1c4) edge[trans] node[above]{$1$} (c5);
		
		\draw[->] (xnc1) edge[trans] node[above]{$\#$} (xnc2);
		\draw[->] (xnc2) to[bend left=30]
		node[above, trans]{$\bar{0},\ \text{pop}(0)$} (xnc3);
		\draw[->] (xnc2) edge[trans]
		node[above]{$\bar{1},\ \text{pop}(1)$} (xnc3);
		\draw[->] (xnc2) to[bend right=30]
		node[below, trans]{$\overline{n{-}1},\ \text{pop}(n{-}1)$} (xnc3);
		\node[dots node] at ($(xnc2)!0.5!(xnc3)+(0,-0.25)$) {$\vdots$};
		\draw[->] (xnc3) edge[trans] node[above]{$\#$} (xnc4);
		\draw[->] (xnc4) edge[trans] node[below right]{$n{-}1$} (c5);
	\end{tikzpicture}
\caption{For any $j > 0$, the transitions of the PDA $B_0^j$ from the initial state $(\lozenge,0)^{b_0^j}$.}
\label{fig:B0j}
\end{figure}
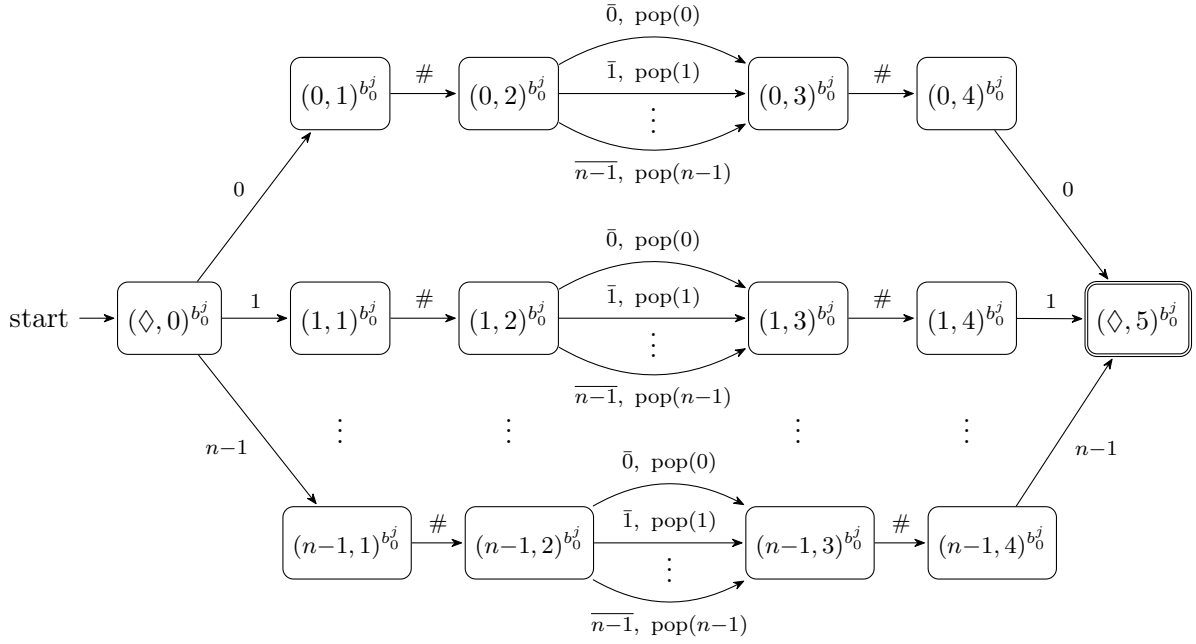

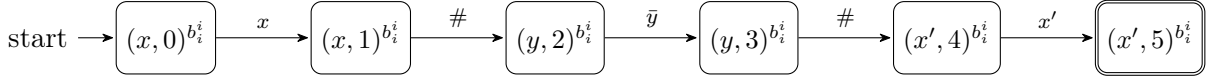
\begin{figure}[t]
	\centering
	\usetikzlibrary{automata, positioning, arrows.meta, bending}
	\begin{tikzpicture}[
		>={Stealth[round]},
		shorten >=1pt,
		auto,
		node distance=2.0cm and 1.25cm,
		every state/.style={
			draw, rounded corners=4pt,
			align=center, font=\small, rectangle,
		},
		trans/.style={font=\scriptsize},
		dots node/.style={draw=none, minimum width=0pt, minimum height=0pt}
		]
		
		\node[state, initial] (xc0) {$(x,0)^{b_i^i}$};
		\node[state, right=of xc0] (xc1) {$(x,1)^{b_i^i}$};
		\node[state, right=of xc1] (yc2) {$(y,2)^{b_i^i}$};
		\node[state, right=of yc2] (yc3) {$(y,3)^{b_i^i}$};
		\node[state, right=of yc3] (xpc4) {$(x',4)^{b_i^i}$};
		\node[state, right=of xpc4, accepting] (xpc5) {$(x',5)^{b_i^i}$};
		
		\draw[->] (xc0)  edge[trans] node[above]{$x$}                        (xc1);
		\draw[->] (xc1)  edge[trans] node[above]{$\#$} (yc2);
		\draw[->] (yc2)  edge[trans] node[above]{$\bar{y}$}                   (yc3);
		\draw[->] (yc3)  edge[trans] node[above]{$\#$} (xpc4);
		\draw[->] (xpc4) edge[trans] node[above]{$x'$}                        (xpc5);
		
	\end{tikzpicture}
\caption{For any $i > 0$ and any $x \in \{0,\dots,n-1\}$, the transitions of the NFA $B_i^i$ from the initial state $(x,0)^{b_i^i}$. Here $y$ is any neighbor of $x$ in the graph $G$ and $x'$ is any neighbor of $y$ in $G$.}
\label{fig:Bij-eq}
\end{figure}

\begin{lemma}\label{lem:Bij}
	Let $w$ be some word and let $j \in \{1,\dots,k-1\}$. Then, there is a run of $B_i^j$
	on $w$ starting from some initial state $x_i$ and ending at some final state $x_i'$ for every $i$ if and only if $x_j$ is a node, there is a neighbor $y$ of $x_j$ such that
	$w := N(x_j,y)$, $y$ is at the top of the stack of $B_0^j$ at the beginning of the run, $x_0 = x_0' = \lozenge$ and for each $i > 0$, $x_i = x_i'$ is a node.
	Moreover, in every collection of runs, if the stack of $B_0^j$ initially stores $\gamma y$, then at the end it stores $\gamma$.
\end{lemma}

\begin{proof}
	Let us first prove the right-to-left implication. In this case, first notice that all the machines
	$B_i^j$ except for $B_0^j$ and $B_j^j$  have the run $(x_i,0)^{b_i^j}, (x_i,1)^{b_i^j}, (x_i,2)^{b_i^j}, (x_i,3)^{b_i^j}, (x_i,4)^{b_i^j}, (x_i,5)^{b_i^j}$ on reading the word $w$. On the other hand, the machine $B_0^j$ upon reading the subword $x_j\#$ has the run
	$(\lozenge,0)^{b_0^j}, (x_j,1)^{b_0^j}, (x_j,2)^{b_0^j}$. From there it reads $\overline{y}$ and since $y$ is at the top of the stack,
	$B_0^j$ pops it and moves to $(x_j,3)^{b_0^j}$. From there it reads $\#x_j$ and has the run $(x_j,4)^{b_0^j},(\lozenge,5)^{b_0^j}$.
	Similarly, the machine $B_j^j$ upon reading $w$ has the run $(x_j,0)^{b_j^j}, (x_j,1)^{b_j^j}, (y,2)^{b_j^j}, (y,3)^{b_j^j}, (x_j,4)^{b_j^j}, (x_j,5)^{b_j^j}$.
	Hence, this direction of the claim is true.
	
	Let us now prove the other direction. Note that for any $i \notin \{0,j\}$, any run of $B_i^j$
	must be of the form $(x_i,0)^{b_i^j}, (x_i,1)^{b_i^j}, \dots, (x_i,5)^{b_i^j}$ with $x_i$ being a node. 
	Hence $x_i = x_i'$ is true for $i \notin \{0,j\}$. 
	
	Now let us analyze the runs of $B_0^j$ on $B_j^j$ on the word $w$. By construction,
	it is easy to see that $|w| = 5$. Hence $w = a_1a_2a_3a_4a_5$ for some letters $a_1,a_2,a_3,a_4,a_5$.
	Note that if $a_1 \neq x_j$ or $x_j$ is not a node, then there is no transition from $(x_j,0)^{b_j^j}$ 
	labelled by $a_1$.
	Hence, $a_1 = x_j$, $x_j$ is a node and the machines $B_0^j$ and $B_j^j$ move to the states $(x_j,1)^{b_0^j}$ and $(x_j,1)^{b_j^j}$
	respectively. Now, if $a_2 \neq \#$, then there are no transitions from either of these states.
	Hence, $a_2 = \#$ and $B_0^j$ and $B_j^j$ move to the states $(x_j,2)^{b_0^j}$ and $(y',2)^{b_j^j}$ for some $y'$
	which is a neighbor of $x_j$. For similar reasons as $a_1$ and $a_2$, $a_3$ must be $\overline{y'}$.
	Now, after reading $\overline{y'}$, $B_0^j$ attempts to pop $y'$ from the stack. By assumption on the stack, this is possible if and only if $y = y'$ and so $y$ and $x_j$ must be neighbors. Hence, after reading $a_3$, the machines
	now move to the states $(x_j,3)^{b_0^j}$ and $(y,3)^{b_j^j}$. Now $a_4$ must be $\#$ and
	the machines now move to the states $(x_j,4)^{b_0^j}$ and $(x',4)^{b_j^j}$ for some neighbor $x'$ of 
	$y$. Finally, $a_5$ must be $x_j$ and so $x'$ must also be equal to $x_j$
	and then the machines move to $(\lozenge,5)^{b_0^j}$ and $(x_j,5)^{b_j^j}$, thereby completing the proof.
\end{proof}

Now, for each $j$, let us first create $k-j-1$ many copies of $B_i^j$ and call them
$B_i^{j,k-1},\dots,B_i^{j,j+1}$. Then, for each $i$, we compose all of these gadgets in the following order to get the gadget $B_i$: 
$B_i^{1,k-1}, B_i^{2,k-1}, \dots, B_i^{k-2,k-1},B_i^{1,k-2},B_i^{2,k-2},\dots,B_i^{k-3,k-2},$ $B_i^{1,k-3},\dots,B_i^{1,2}$. 

The intuition behind this composition is as follows: At the end of sub-part I, for some nodes $x_2,\dots,x_{k-1}$, the stack contains 
contains $(k-2)$ copies of $x_{k-1}$ at the top, followed by $(k-3)$ copies of $x_{k-2}$ and so till a single copy of $x_2$ and the $(k-1)$ NFAs store the nodes $x_1,x_2,\dots,x_{k-1}$. If we now run each $B_i$ from this point, each $B_i$ would first
execute $B_i^{1,k-1}$ which ensures that $x_1$ and $x_{k-1}$ are neighbors and pops the topmost $x_{k-1}$ from the stack. Then, each $B_i$ would execute $B_i^{2,k-1}$ which ensures that $x_2$ and $x_{k-1}$
are neighbors and pops the next $x_{k-1}$ from the stack and so on all the way till $B_i^{k-2,k-1}$. This process will ensure that $x_{k-1}$ is a neighbor of every other node. 
At this point, the topmost part of the stack contains
$k-3$ copies of $x_{k-2}$. Now, each $B_i$ will execute $B_i^{1,k-2}$ which will ensure that
$x_1$ and $x_{k-2}$ are neighbors, then each $B_i$ will execute $B_i^{2,k-2}$, which will ensure that $x_2$ and $x_{k-2}$
are neighbors and so on all the way up till $B_i^{1,2}$. Hence, in this way we would have 
ensured that $x_1,x_2,\dots,x_{k-1}$ is a $(k-1)$-clique. This intuition is made more precise by the next paragraph.

For any collection of nodes $x_1,\dots,x_{k-1}$, let $B(x_1,\dots,x_{k-1})$ be the word 
\begin{align*}
	N(x_1,x_{k-1})@N(x_2,x_{k-1})@\dots @N(x_{k-2},x_{k-1})@N(x_1,x_{k-2})@\dots @N(x_{k-3},x_{k-2})\dots N(x_1,x_2)	
\end{align*}

The following theorem follows by applying the definition of composition to the machines $B_i$
along with the previous lemma. Its proof is similar to the proof of the Setup Theorem (Theorem~\ref{thm:setup}).

\begin{theorem}[The Check Theorem]\label{thm:check}
	Let $w$ be some word. Then, there is a run of $B_i$ on $w$ starting from some 
	initial state $x_i$ and ending at some final state $x_i'$ for every $i$
	if and only if for each $i > 0$, $x_i = x_i'$ is a node,
	$w = B(x_1,\dots,x_{k-1})$, the stack of $P_0$ at the beginning of the run is of the form $\gamma x_2 x_3^2 \dots x_{k-2}^{k-3} x_{k-1}^{k-2}$ for some $\gamma$, $x_0 = x_0' = \lozenge$ and 
	$x_1,\dots,x_{k-1}$ is a $(k-1)$-clique. Moreover, in every such collection of runs, the stack
	at the end stores $\gamma$.
\end{theorem}

We conclude with a discussion on the number of states and the time taken to construct
each $B_i$. Each $B_i$ is obtained by composing $O(k^2)$ many gadgets, each of which have
$O(n)$ states and each of which can be constructed in $O(n^2)$ time.
Hence, it follows that 
\begin{proposition}[Size of $B_i$]\label{prop:size-Bi}
	Each $B_i$ has $O(nk^2)$ states and can be constructed in $O(n^2k^2)$ time.
\end{proposition}

\paragraph*{Sub-Part III: The Exploration. } Recall the description of this sub-part: After sub-part II, each of the $(k-1)$ NFAs are at some states storing nodes $x_1,\dots,x_{k-1}$ respectively such that
$\{x_1,\dots,x_{k-1}\}$ is a $(k-1)$-clique. The role of this sub-part is to find another set of $k-1$ nodes such that every $x_i$ is connected to every $y_j$. 
For this purpose, we will create four different types of gadgets and then compose them all together.
All these gadgets will be similar to the gadgets that we have seen in the previous sub-parts (and the prologue), but with slight modifications. We will now describe the first type of gadgets, whose
purpose will be to simply push some node into the stack $k$ times.

For each $i \in \{0,\dots,k-1\}$, we will create a machine $D_i$, which will have as its states
$(x,c)^{d_i}$ for $x \in \{0,\dots,n-1\} \cup \{\lozenge\}$ and $c \in \{0,\dots,k\}$.
The machine $D_0$ is a PDA which moves according to the following transitions:
\begin{itemize}
	\item From $(\lozenge,0)^{d_0}$, upon reading a node $x$, it pushes $x$ onto the stack and moves to $(x,1)^{d_0}$.
	\item For any node $x$ and any $c \in \{1,\dots,k-2\}$, from $(x,c)^{d_0}$, upon reading $x$, it pushes
	$x$ onto the stack and moves to $(x,c+1)^{d_0}$. 
	\item For any node $x$, from $(x,k-1)^{d_0}$, upon reading $x$, it pushes $x$ onto the stack
	and moves to $(\lozenge,k)^{d_0}$.
\end{itemize}

Intuitively, this gadget remembers the first node that is read and then ensures
that the next $k-1$ input letters are the same as the first node. After ensuring this, it forgets
the node that it remembered. Further, each time it reads a node
it simply pushes it into the stack. See Figure~\ref{fig:D0} for a representation of $D_0$.

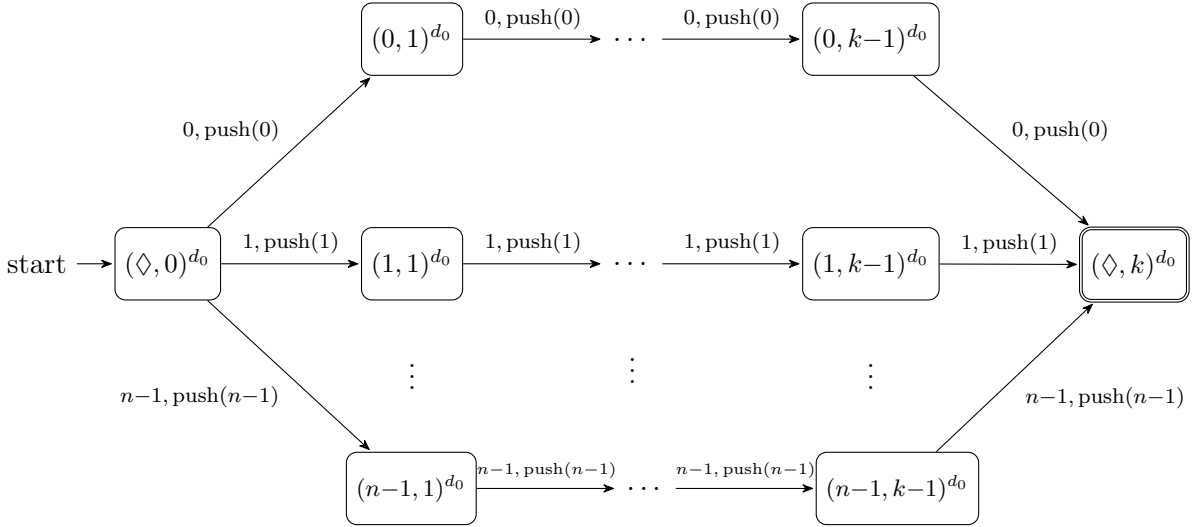
\begin{figure}[t]
	\centering
	\usetikzlibrary{automata, positioning, arrows.meta, bending}
	\begin{tikzpicture}[
		>={Stealth[round]},
		shorten >=1pt,
		auto,
		node distance=2.0cm and 1.85cm,
		every state/.style={
			draw, rounded corners=4pt,
			align=center, font=\small, rectangle,
		},
		trans/.style={font=\scriptsize},
		dots node/.style={draw=none, minimum width=0pt, minimum height=0pt}
		]
		
		\node[state, initial] (c0) {$(\lozenge, 0)^{d_0}$};
		
		\node[state, above right=2.0cm and 1.85cm of c0]  (x0c1) {$(0,1)^{d_0}$};
		\node[dots node, right=of x0c1]             (x0dots) {$\dots$};
		\node[state, right=of x0dots]               (x0ck1) {$(0,k{-}1)^{d_0}$};
		
		\node[state, right=of c0]                          (x1c1) {$(1,1)^{d_0}$};
		\node[dots node, right= of x1c1]             (x1dots) {$\dots$};
		\node[state, right=of x1dots]               (x1ck1) {$(1,k{-}1)^{d_0}$};
		
		\node[dots node, below=0.5cm of x1c1]  (vd1) {$\vdots$};
		\node[dots node, below=0.75cm of x1dots]  (vd2) {$\vdots$};
		\node[dots node, below=0.5cm of x1ck1] (vd3) {$\vdots$};
		
		\node[state, below right=2.0cm and 1.65cm of c0]   (xnc1)  {\footnotesize$(n{-}1,1)^{d_0}$};
		\node[dots node, right=of xnc1]              (xndots) {$\dots$};
		\node[state, right=of xndots]                (xnck1) {\footnotesize$(n{-}1,k{-}1)^{d_0}$};
		
		\node[state, below right=2.0cm and 1.85cm of x0ck1, accepting] (ck) {$(\lozenge, k)^{d_0}$};
		
		\draw[->] (c0) edge[trans] node[above left]{$0, \text{push}(0)$}     (x0c1);
		\draw[->] (c0) edge[trans] node[above]{$1, \text{push}(1)$}           (x1c1);
		\draw[->] (c0) edge[trans] node[below left]{$n{-}1, \text{push}(n{-}1)$} (xnc1);
		
		\draw[->] (x0c1)   edge[trans] node[above]{$0, \text{push}(0)$} (x0dots);
		\draw[->] (x0dots) edge[trans] node[above]{$0, \text{push}(0)$} (x0ck1);
		\draw[->] (x0ck1)  edge[trans] node[above right]{$0, \text{push}(0)$} (ck);
		
		\draw[->] (x1c1)   edge[trans] node[above]{$1,\text{push}(1)$} (x1dots);
		\draw[->] (x1dots) edge[trans] node[above]{$1,\text{push}(1)$} (x1ck1);
		\draw[->] (x1ck1)  edge[trans] node[above]{$1,\text{push}(1)$} (ck);
		
		\draw[->] (xnc1)   edge[trans] node[above]{\tiny $n{-}1,\text{push}(n{-}1)$} (xndots);
		\draw[->] (xndots) edge[trans] node[above]{\tiny $n{-}1,\text{push}(n{-}1)$} (xnck1);
		\draw[->] (xnck1)  edge[trans] node[below right]{$n{-}1,\text{push}(n{-}1)$} (ck);
		
	\end{tikzpicture}
\caption{The transitions of the PDA $D_0$ from the initial state $(\lozenge,0)^{d_0}$.}
\label{fig:D0}
\end{figure}	

Having described the PDA $D_0$, we now describe the remaining machines. For each $i \in \{1,\dots,k-1\}$, the machine $D_i$ is an NFA which upon reading any letter from a state $(x,c)^{d_i}$ with $x \in \{0,\dots,n-1\}$ and $c < k-1$ simply moves to $(x,c+1)^{d_i}$. The following lemma now follows immediately from the construction of these gadgets. Its proof is exactly similar
to the proof of the Prologue theorem (Theorem~\ref{thm:prologue}).

\begin{lemma}\label{lem:Di}
	Let $w$ be some word. Then, there is a run of $D_i$ on $w$ starting from
	some initial state $x_i$ and ending at some final state $x_i'$ for every $i$
	if and only if $w = y^k$ for some node $y$, 
	$x_0 = x_0' = \lozenge$ and for each $i > 0$, $x_i = x_i'$ is a node. Moreover,
	in every such collection of runs, if the stack of $D_0$ initially stores a word $\gamma$, then at the end it stores the word $\gamma y^k$.
\end{lemma}

The second type of gadgets that we will construct will use the $D_i$ gadgets that we just constructed
and the $B_i^j$ gadgets from sub-part II. To this end, for each $i$, we construct a gadget $E_i$
by composing the gadgets $D_i, B_i^1, B_i^2, \dots, B_i^{k-1}$. Intuitively, the idea is that 
each NFA $E_i$ will initially begin with some node $x_i$ stored in its state.
Then the gadget $D_i$ will serve the purpose of pushing $k$ many copies of some node $y$ onto the stack.
Then, $B_i^1$ will ensure that this node $y$ is a neighbor of $x_1$,
$B_i^2$ will ensure that $y$ is a neighbor of $x_2$ and so on. Let us now state this formally.

Recall that for any two nodes $x,y$ we had defined the word $N(x,y) = x\#\overline{y}\#x$.
Given any node $y$ and $k-1$ many nodes $x_1,\dots,x_{k-1}$, let $E(x_1,\dots,x_{k-1},y)$ be the word $$y^k@N(x_1,y)@N(x_2,y)@\dots @N(x_{k-1},y)$$ From the definition of composition of gadgets
and from the lemmas that we have proved regarding the gadgets $D_i$ and $B_i^j$ (Lemmas~\ref{lem:Di} and~\ref{lem:Bij}), the following lemma follows.

\begin{lemma}\label{lem:Ei}
	Let $w$ be some word. Then, there is a run of $E_i$ on $w$ 
	starting from some initial state $x_i$ and ending at some final state $x_i'$ for every $i$ if and only if for each $i > 0$, $x_i = x_i'$ is a node, 
	there is a node $y$ that is a neighbor of each $x_i$ such that
	$w = E(x_1,\dots,x_{k-1},y)$ and $x_0 = x_0' = \lozenge$. Moreover, in every such collection
	of runs, if the stack of $E_0$ initially stores $\gamma$, then at the end it stores $\gamma y$.	
\end{lemma}

Before we move on to the description of the third type of gadgets, we note the size of each $E_i$. Each $E_i$ is obtained by composing $D_i, B_i^1, \dots, B_i^{k-1}$.
Each $D_i$ has $O(nk)$ states and can be constructed in $O(n^2k)$ time.
Each $B_i^j$ has $O(n)$ states and can be constructed in $O(n^2)$ time.
It follows that 
\begin{proposition}[Size of $E_i$]\label{prop:size-Ei}
	Each $E_i$ has $O(nk)$ states and can be constructed in $O(n^2k)$ time.
\end{proposition}

Now, the third type of gadgets is obtained by simply repeating each gadget $E_i$ for $k-1$ times.
Formally, we create a gadget $F_i$, by first creating $k-1$ copies of $E_i$ and then 
composing them all together. 

Given nodes $x_1,\dots,x_{k-1}$ and $y_1,\dots,y_{k-1}$, let $F(x_{1},\dots,x_{k-1},y_1,\dots,y_{k-1})$
be the word given by 
$$E(x_1,\dots,x_{k-1},y_1)@E(x_1,\dots,x_{k-1},y_2)@\dots @E(x_1,\dots,x_{k-1},y_{k-1})$$
By Lemma~\ref{lem:Ei}, we immediately have the following lemma.

\begin{lemma}\label{lem:Fi}
	Let $w$ be some word. Then, there is a run of $F_i$ on $w$
	starting from some initial state $x_i$ and ending at some final state $x_i'$ for every $i$ 
	if and only if for each
	$i > 0$, $x_i = x_i'$ is a node, 
	there are nodes $y_1,\dots,y_{k-1}$ that are all neighbors of each $x_i$,
	$w = F(x_1,\dots,x_{k-1},y_1,\dots,y_{k-1})$ and 
	$x_0 = x_0' = \lozenge$. Moreover, in every such collection
	of runs,  if the stack of $F_0$ initially stores $\gamma$, then at the end it stores $\gamma y_{1} y_2 \dots y_{k-1}$.
\end{lemma}

We will now describe the fourth type of gadgets. Their job is to pop the topmost $k-1$ letters on the
stack and store them in the NFAs. To this end, for each $i \in \{0,\dots,k-1\}$, we will 
create a machine $H_i$, which will have as its states $(x,c)^{h_i}$ for $x \in \{0,\dots,n-1\} \cup \{\lozenge\}$ and $c \in \{0,\dots,k\}$. 

The machine $H_0$ is a PDA which operates according to the following transitions:
\begin{itemize}
	\item From $(\lozenge,0)^{h_0}$, upon reading $\#$, it moves to $(\lozenge,1)^{h_0}$.
	\item For any $c \in \{1,\dots,k-1\}$, from $(\lozenge,c)^{h_0}$, upon reading a letter of the form $\overline{z}$,
	it pops $z$ from the stack and moves to $(\lozenge,c+1)^{h_0}$.
\end{itemize}

Intuitively, this gadget first reads $\#$, then forces the next $k-1$ input letters to be the topmost $k-1$ letters in the stack. See Figure~\ref{fig:H0} for a representation of $H_0$.

We now describe the remaining NFAs. For each $i \in \{1,\dots,k-1\}$, the machine $H_i$ is an NFA which operates according to the following transitions:
\begin{itemize}
	\item For any node $x$, from $(x,0)^{h_i}$, upon reading $\#$, it moves to $(\lozenge,1)^{h_i}$.
	\item For any node $z$ and any $c \in \{1,\dots,k-i-1\}$, from $(\lozenge,c)^{h_i}$, upon reading a letter of the form $\overline{z}$, it moves to $(\lozenge,c+1)^{h_i}$.
	\item For any node $z$, from $(\lozenge,k-i-1)$, upon reading a letter of the form $\overline{z}$, it moves to $(z,k-i)^{h_i}$.
	\item For any nodes $z,x'$ and any $c \in \{k-i,\dots,k-1\}$, from $(z,c)^{h_i}$, upon reading a 
	letter of the form $\overline{x'}$, it moves to $(z,c+1)^{h_i}$.
\end{itemize}

Intuitively, this gadget first reads $\#$ and forgets the node that it currently stores. Then, it remembers the $(k-i)^{th}$ input letter that appears after the occurrence of $\#$. See Figure~\ref{fig:Hi} for a representation of $H_i$.

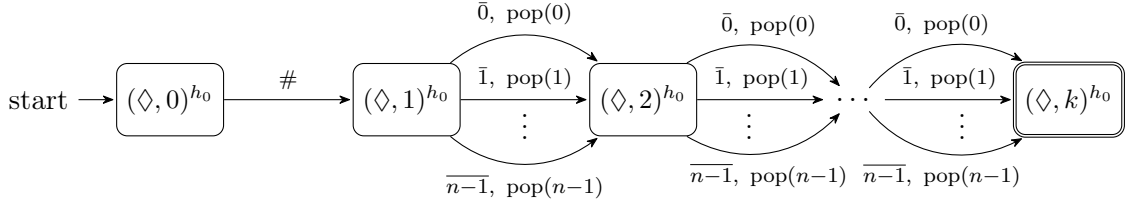
\begin{figure}[t]
	\centering
	\usetikzlibrary{automata, positioning, arrows.meta, bending}
	
	\begin{tikzpicture}[
		>={Stealth[round]},
		shorten >=1pt,
		auto,
		node distance=2.0cm and 1.7cm,
		every state/.style={
			draw, rounded corners=4pt,
			align=center, font=\small, rectangle,
		},
		trans/.style={font=\scriptsize},
		dots node/.style={draw=none, minimum width=0pt, minimum height=0pt}
		]
		
		\node[state, initial] (c0)  {$(\lozenge, 0)^{h_0}$};
		\node[state, right=of c0]   (c1)  {$(\lozenge, 1)^{h_0}$};
		\node[state, right=of c1]   (c2)  {$(\lozenge, 2)^{h_0}$};
		\node[dots node, right=of c2] (dots) {$\dots$};
		\node[state, right=of dots, accepting]   (ck)  {$(\lozenge, k)^{h_0}$};
		
		\foreach \src/\tgt in {c1/c2, c2/dots, dots/ck}{
			\draw[->] (\src) to[bend left=40]
			node[above, trans]{$\bar{0},\ \text{pop}(0)$} (\tgt);
			\draw[->] (\src) edge[trans]
			node[above]{$\bar{1},\ \text{pop}(1)$} (\tgt);
			\draw[->] (\src) to[bend right=40]
			node[below, trans]{$\overline{n{-}1},\ \text{pop}(n{-}1)$} (\tgt);
			\node[dots node] at ($(\src)!0.5!(\tgt)+(0,-0.25)$) {$\vdots$};
		}
	
		\draw[->] (c0) edge[trans]
		node[above]{$\#$} (c1);	
	\end{tikzpicture}
\caption{The transitions of the PDA $H_0$ from the initial state $(\lozenge,0)^{h_0}$.}
\label{fig:H0}
\end{figure}	

\begin{figure}[t]
	\centering
	\usetikzlibrary{automata, positioning, arrows.meta, bending}
	
\begin{tikzpicture}[
	>={Stealth[round]},
	shorten >=1pt,
	auto,
	node distance=1.5cm and 0.85cm,
	every state/.style={
		draw, rounded corners=4pt,
		align=center, font=\small, rectangle,
	},
	trans/.style={font=\scriptsize},
	dots node/.style={draw=none, minimum width=0pt, minimum height=0pt}
	]
	
	\node[state, initial] (xc0)   {\footnotesize $(x,0)^{h_i}$};
	\node[state, right=of xc0]    (xc1)   {\footnotesize$(\lozenge,1)^{h_i}$};
	\node[dots node, right=of xc1] (dots1) {$\dots$};
	\node[state, right=of dots1]  (xck1i) {\footnotesize $(\lozenge,k{-}i)^{h_i}$};
	\node[state, right=of xck1i]    (zki1)  {\footnotesize $(z,k{-}i{+}1)^{h_i}$};
	\node[dots node, right=of zki1] (dots2) {$\dots$};
	\node[state, right=of dots2, accepting]  (zk)    {\footnotesize $(z,k)^{h_i}$};
	
	\draw[->] (xc0) edge[trans]
	node[above]{$\#$} (xc1);
	
	\draw[->] (xc1) to[bend left=30]
	node[above, trans]{$\bar{0}$} (dots1);
	\draw[->] (xc1) edge[trans]
	node[above]{$\bar{1}$} (dots1);
	\draw[->] (xc1) to[bend right=30]
	node[below, trans]{$\overline{n{-}1}$} (dots1);
	\node[dots node] at ($(xc1)!0.5!(dots1)+(0.15,-0.1)$) {$\vdots$};
	
	\draw[->] (dots1) to[bend left=30]
	node[above, trans]{$\bar{0}$} (xck1i);
	\draw[->] (dots1) edge[trans]
	node[above]{$\bar{1}$} (xck1i);
	\draw[->] (dots1) to[bend right=30]
	node[below, trans]{$\overline{n{-}1}$} (xck1i);
	\node[dots node] at ($(dots1)!0.5!(xck1i)+(-0.2,-0.15)$) {$\vdots$};
	
	\draw[->] (xck1i) edge[trans] node[above]{$\bar{z}$} (zki1);
	
	
	\draw[->] (zki1) to[bend left=30]
	node[above, trans]{$\bar{0}$} (dots2);
	\draw[->] (zki1) edge[trans]
	node[above]{$\bar{1}$} (dots2);
	\draw[->] (zki1) to[bend right=30]
	node[below, trans]{$\overline{n{-}1}$} (dots2);
	\node[dots node] at ($(zki1)!0.5!(dots2)+(0.3,-0.1)$) {$\vdots$};
	
	\draw[->] (dots2) to[bend left=30]
	node[above, trans]{$\bar{0}$} (zk);
	\draw[->] (dots2) edge[trans]
	node[above]{$\bar{1}$} (zk);
	\draw[->] (dots2) to[bend right=30]
	node[below, trans]{$\overline{n{-}1}$} (zk);
	\node[dots node] at ($(dots2)!0.5!(zk)+(-0.1,-0.1)$) {$\vdots$};
	
\end{tikzpicture}
\caption{For any $i > 0$ and any $x, z \in \{0,\dots,n-1\}$, the transitions of the NFA $H_i$ from the initial state $(x,0)^{h_i}$.}
\label{fig:Hi}
\end{figure}
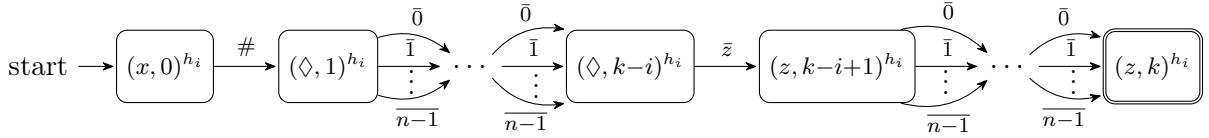	

From the construction, the following lemma easily follows. The proof is exactly similar to the proof 
of the Prologue theorem.
\begin{lemma}\label{lem:Hi}
	Let $w$ be some word. Then, there is a run of $H_i$ on $w$ starting from some initial state $x_i$
	and ending at some final state $x_i'$ for every $i$ if and only if $w = \# \overline{y_{k-1}}\dots \overline{y_1}$ for some nodes $y_{k-1},\dots,y_1$, the stack of $H_0$ at the beginning of the run is of the form $\gamma y_1 y_2 \dots y_{k-1}$, $x_0 = x_0' = \lozenge$ and for each $i > 0$, $x_i' = y_i$. Moreover, in every such collection of runs, the stack of $H_0$ at the end stores $\gamma$.
\end{lemma}

Now, let us create the final gadget for this sub-part. For each $i$, we compose the gadget $F_i$ and $H_i$
to get a new gadget $L_i$. For any collection of nodes $x_1, \dots, x_{k-1}, y_1,\dots, y_{k-1}$ let
$L(x_1,\dots,x_{k-1},y_1,\dots,y_{k-1})$ be the word given by 
$$E(x_1,\dots,x_{k-1},y_1,\dots,y_{k-1}) \ @ \ \# \ \overline{y_{k-1}} \ \overline{y_{k-2}} \dots \overline{y_1}$$

From Lemmas~\ref{lem:Fi} and~\ref{lem:Hi}, we immediately get the following theorem.
\begin{theorem}[The Exploration Theorem]\label{thm:exploration}
	Let $w$ be some word. Then, there is a run of $L_i$ on $w$ starting from 
	some initial state $x_i$ and ending at some final state $y_i$ for every $i$ 
	if and only if for each $i > 0$, $x_i$ and $y_i$ are nodes,
	$w = L(x_1,\dots,x_{k-1},y_1,\dots,y_{k-1})$, 
	each node in $\{y_1,\dots,y_{k-1}\}$ is a neighbor of each node in $\{x_1,\dots,x_{k-1}\}$ and
	$x_0 = y_0 = \lozenge$. Moreover, in every such collection
	of runs, if the stack of $L_0$ initially stores $\gamma$, then at the end it stores $\gamma$ as well.
\end{theorem}

We conclude with a discussion on the size of each $L_i$.
Each $L_i$ is obtained by composing $F_i$ and $H_i$. Each $F_i$ has $O(k)$ copies
of $E_i$ and so by Proposition~\ref{prop:size-Ei}, it follows that each $F_i$ has $O(nk^2)$ states and can be constructed in $O(n^2k^2)$ time. It is easy to see that
each $H_i$ has $O(nk)$ states and can be constructed in $O(n^2k)$ time.
Hence, it follows that
\begin{proposition}[Size of $L_i$]\label{prop:size-Li}
	Each $L_i$ has $O(nk^2)$ states and can be constructed in $O(n^2k^2)$ time.
\end{proposition}

This completes the final gadget that we wanted to construct for the third sub-part. Now let us see how to put together the gadgets from all three sub-parts together.

\paragraph*{Wrapping up Part 1. } Now, we wrap up the first part by combining the gadgets from its three sub-parts into one
gadget. To this end, for each $i$, we will create a gadget $X_i$ by composing $A_i, B_i$ and $L_i$.
Intuitively, $X_i$ first executes $A_i$ which begins with a collection of $k-1$ nodes
$x_1,\dots,x_{k-1}$ and then sets up the stack in a specific way. Then $X_i$ executes $B_i$ which verifies that these nodes $x_1,\dots,x_{k-1}$
form a $(k-1)$-clique. Finally, $X_i$ executes $L_i$ which finds another collection of nodes
$y_1,\dots,y_{k-1}$ such that every node in this new collection is a neighbor of every node
in $\{x_1,\dots,x_{k-1}\}$.

To make this intuition more precise, for any collection of nodes $x_1,\dots,x_{k-1}, y_1,\dots,y_{k-1}$, let  $W(x_1,\dots,x_{k-1},y_1,\dots,y_{k-1})$ be the word given by 
$$A(x_1,\dots,x_{k-1})@B(x_1,\dots,x_{k-1})@L(x_1,\dots,x_{k-1},y_1,\dots,y_{k-1})$$
%
Now, from the Setup Theorem (Theorem~\ref{thm:setup}), the Check Theorem (Theorem~\ref{thm:check}) and the Exploration Theorem (Theorem~\ref{thm:exploration}), we have the 
following main result.

\begin{theorem}[Part 1 Theorem]\label{thm:part-1}
	Let $w$ be some word. Then, there is a run of $X_i$ on $w$ starting from
	some initial state $x_i$ and ending at some final state $y_i$ for every $i$ if and only if
	for each $i > 0$, $x_i$ and $y_i$ are nodes,
	$w = W(x_1,\dots,x_{k-1},y_1,\dots,y_{k-1})$,
	$x_1,\dots,x_{k-1}$ is a $(k-1)$-clique, each node in $\{y_1,\dots,y_{k-1}\}$ is a neighbor of 
	each node in $\{x_1,\dots,x_{k-1}\}$ and $x_0 = y_0 = \lozenge$. Moreover, in every such collection
	of runs, if the stack of $X_0$ initially stores $\gamma$, then at the end it stores $\gamma$ as well.
\end{theorem}

We conclude this part by a discussion on the size of $X_i$. By propositions~\ref{prop:size-Ai},\ref{prop:size-Bi} and \ref{prop:size-Li}, it follows that

\begin{proposition}[Size of $X_i$]\label{prop:size-xi}
	Each $X_i$ has $O(nk^2)$ states and can be constructed in $O(n^2k^2)$ time.
\end{proposition}

\subsubsection{Parts 2 and 3} 

The second and third parts are exactly like the first part.
More precisely, for each $i$, we create two more copies of $X_i$ and call them $Y_i$ and $Z_i$.
$Y_i$ and $Z_i$ are the gadgets for the second and the third part. 

Intuitively, the first set of gadgets $X_0,\dots,X_{k-1}$ begins with a collection of nodes $x_1,\dots,x_{k-1}$, verifies
it to be a $(k-1)$-clique, then finds another collection $y_1,\dots,y_{k-1}$ such that
each node in this new collection is a neighbor of each node in $x_1,\dots,x_{k-1}$, then
it ends with the nodes $y_1,\dots,y_{k-1}$ stored in the states of the NFAs $X_1,\dots,X_{k-1}$.
Now, if we compose each $Y_i$ with $X_i$, they will verify that $y_1,\dots,y_{k-1}$ is a $(k-1)$-clique, then they will find another collection of nodes $z_1,\dots,z_{k-1}$ such that each node in this new collection
is a neighbor of each node in $y_1,\dots,y_{k-1}$ and end with the nodes $z_1,\dots,z_{k-1}$ stored in the states of $Y_1,\dots,Y_{k-1}$. Hence, if we now compose each $Z_i$ with $Y_i$, we can verify that $z_1,\dots,z_{k-1}$ 
is a $(k-1)$-clique and also find another collection of nodes $x_1',\dots,x_{k-1}'$ such that
each node in this new collection is a neighbor of each node in $z_1,\dots,z_{k-1}$. If we then check that each $x_i' = x_i$ (which will be done by the Epilogue gadget) then we are guaranteed that
the $x_i$'s, $y_i$'s and $z_i$'s together form a $3(k-1)$-clique. Hence, we will first compose 
the $X_i$'s, $Y_i$'s and $Z_i$'s together and then finally compose these new gadgets with the Prologue and Epilogue gadgets.

For each $i$, let $M_i$ be the gadget obtained by composing $X_i, Y_i$ and $Z_i$. By the Part 1 theorem (Theorem~\ref{thm:part-1}) applied thrice, we immediately get the following theorem

\begin{theorem}[Main Parts Theorem] \label{thm:main-parts}
	Let $w$ be some word. Then, there is a run of $M_i$ on $w$ starting from
	some initial state $x_i$ and ending at some final state $x_i'$ for every $i$ if and only if
	\begin{itemize}
		\item 	For each $i > 0$, $x_i$ and $x_i'$ are nodes and
		there exists nodes $y_1,\dots,y_{k-1}, z_1,\dots,z_{k-1}$ such that $w$ is precisely equal to the word
		$$\text{W}(x_1,\dots,x_{k-1},y_1,\dots,y_{k-1})@\text{W}(y_1,\dots,y_{k-1},z_1,\dots,z_{k-1}) @\text{W}(z_1,\dots,z_{k-1},x_1',\dots,x_{k-1}')$$
		\item The sets $X := \{x_1,\dots,x_{k-1}\}, Y := \{y_1,\dots,y_{k-1}\}$ and  $Z := \{z_1,\dots,z_{k-1}\}$ are all $(k-1)$-cliques.
		\item Every node in $X$ is connected to every node in $Y$, every node in $Y$ is connected
		to every node in $Z$ and every node in $Z$ is connected to every node in $X' = \{x_1',\dots,x_{k-1}'\}$.
		\item $x_0 = x_0' = \lozenge$.
	\end{itemize}
	Moreover, in every such collection of runs, if the stack of $M_0$ initially stores $\gamma$, then at the end it stores $\gamma$ as well.
\end{theorem}

Since $M_i$ is constructed by composing three copies of $X_i$, by Proposition~\ref{prop:size-xi}, we get the following bound on the size of each $M_i$.

\begin{proposition}[Size of $M_i$]\label{prop:size-mi}
	Each $M_i$ has $O(nk^2)$ states and can be constructed in $O(n^2k^2)$ time.
\end{proposition}

This completes the main three parts of the reduction. Now we move on to the epilogue.

\subsubsection{Part 4: Epilogue} 
Recall the description of the epilogue part: At the end of the prologue, the $P_i$ gadgets have pushed
$k-1$ nodes $x_1,\dots,x_{k-1}$ into the stack. If we then deploy the $M_i$ gadgets, by the Main Parts theorem, they will end with some nodes $x_1',\dots,x_{k-1}'$ stored in the NFAs $M_1,\dots,M_{k-1}$.
Furthermore, the stack content at the end will be exactly the same 
as it was at the beginning, i.e., the stack content will be $x_{k-1},\dots,x_1$.  
Hence, to check if $x_i' = x_i$ for each $i$, we only need
to pop the stack one element at a time and check that the $i^{th}$ element popped is the node
stored in the $(k-i)^{th}$ NFA. Equivalently, it suffices to check that the $(k-i)^{th}$ element
popped is the node stored in the $i^{th}$ NFA. This is what the Epilogue gadget will accomplish now.

Formally, for each $i \in \{0,\dots,k-1\}$, we will construct a gadget $Q_i$, which will have as
its states $(x,c)^{q_i}$ for $x \in \{0,\dots,n-1\} \cup \{\lozenge\}$ and $c \in \{0,\dots,k-1\}$.

The machine $Q_0$ is a PDA, which upon reading a letter of the form $\overline{z}$ from a state $(\lozenge,c)^{q_0}$ with $c < k-1$, pops $z$ from the stack and moves to $(\lozenge,c+1)^{q_0}$. Intuitively, this gadget simply pops $k-1$ nodes from the stack. See Figure~\ref{fig:q0} for a representation of $Q_0$.

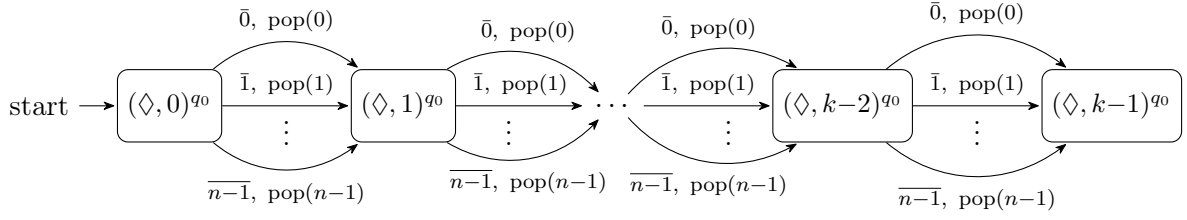
\begin{figure}[t]
	\centering
	\usetikzlibrary{automata, positioning, arrows.meta, bending}
	\begin{tikzpicture}[
		>={Stealth[round]},
		shorten >=1pt,
		auto,
		node distance=2.0cm and 1.7cm,
		every state/.style={
			draw, rounded corners=4pt,
			align=center, font=\small, rectangle,
		},
		trans/.style={font=\scriptsize},
		dots node/.style={draw=none, minimum width=0pt, minimum height=0pt}
		]
		
		\node[state, initial] (c0)  {$(\lozenge, 0)^{q_0}$};
		\node[state, right=of c0]   (c1)  {$(\lozenge, 1)^{q_0}$};
		\node[dots node, right=of c1] (dots) {$\dots$};
		\node[state, right=of dots]  (ck2) {$(\lozenge, k{-}2)^{q_0}$};
		\node[state, right=of ck2]   (ck1) {$(\lozenge, k{-}1)^{q_0}$};
		
		\foreach \src/\tgt in {c0/c1, c1/dots, dots/ck2, ck2/ck1}{
			\draw[->] (\src) to[bend left=40]
			node[above, trans]{$\bar{0},\ \text{pop}(0)$} (\tgt);
			\draw[->] (\src) edge[trans]
			node[above]{$\bar{1},\ \text{pop}(1)$} (\tgt);
			\draw[->] (\src) to[bend right=40]
			node[below, trans]{$\overline{n{-}1},\ \text{pop}(n{-}1)$} (\tgt);
			\node[dots node] at ($(\src)!0.5!(\tgt)+(0,-0.25)$) {$\vdots$};
		}
	\end{tikzpicture}
\caption{The transitions of the PDA $Q_0$ from the initial state $(\lozenge,0)^{q_0}$.}
\label{fig:q0}
\end{figure}

For each $i \in \{1,\dots,k-1\}$, the machine $Q_i$ is an NFA, which upon reading a letter of the form $\overline{z}$ from a state $(x,c)^{q_i}$ with $c < k-1$, moves to $(x,c+1)^{q_i}$ if $c \neq k-i-1$ and otherwise moves to $(x,c+1)^{q_i}$ if and only if $c = k-i-1$ and $x = z$.
Intuitively, the gadget $Q_i$ ensures that the $(k-i)^{th}$ letter that is read is the same
as the one stored in its state. See Figure~\ref{fig:qi} for a representation of $Q_i$.

\begin{figure}[t]
	\centering
	\usetikzlibrary{automata, positioning, arrows.meta, bending}
	\begin{tikzpicture}[
		>={Stealth[round]},
		shorten >=1pt,
		auto,
		node distance=2.0cm and 1.15cm,
		every state/.style={
			draw, rounded corners=4pt,
			align=center, font=\small, rectangle,
		},
		trans/.style={font=\scriptsize},
		dots node/.style={draw=none, minimum width=0pt, minimum height=0pt}
		]
		
		\node[state, initial] (xc0)   {$(x,0)^{q_i}$};
		\node[dots node, right=of xc0] (dots1) {$\dots$};
		\node[state, right=of dots1]  (xck1i) {$(x,k{-}i{-}1)^{q_i}$};
		\node[state, right=of xck1i]  (xki)   {$(x,k{-}i)^{q_i}$};
		\node[dots node, right=of xki] (dots2) {$\dots$};
		\node[state, right=of dots2]  (xck1)  {$(x,k{-}1)^{q_i}$};
		
		
		\draw[->] (xc0) to[bend left=30]
		node[above, trans]{$\bar{0}$} (dots1);
		\draw[->] (xc0) edge[trans]
		node[above]{$\bar{1}$} (dots1);
		\draw[->] (xc0) to[bend right=30]
		node[below, trans]{$\overline{n{-}1}$} (dots1);
		\node[dots node] at ($(xc0)!0.55!(dots1)+(0,-0.2)$) {$\vdots$};
		
		\draw[->] (dots1) to[bend left=30]
		node[above, trans]{$\bar{0}$} (xck1i);
		\draw[->] (dots1) edge[trans]
		node[above]{$\bar{1}$} (xck1i);
		\draw[->] (dots1) to[bend right=30]
		node[below, trans]{$\overline{n{-}1}$} (xck1i);
		\node[dots node] at ($(dots1)!0.4!(xck1i)+(0,-0.2)$) {$\vdots$};
		
		\draw[->] (xck1i) edge[trans] node[above]{$\bar{x}$} (xki);
		
		\draw[->] (xki) to[bend left=30]
		node[above, trans]{$\bar{0}$} (dots2);
		\draw[->] (xki) edge[trans]
		node[above]{$\bar{1}$} (dots2);
		\draw[->] (xki) to[bend right=30]
		node[below, trans]{$\overline{n{-}1}$} (dots2);
		\node[dots node] at ($(xki)!0.6!(dots2)+(0,-0.2)$) {$\vdots$};
		
		\draw[->] (dots2) to[bend left=30]
		node[above, trans]{$\bar{0}$} (xck1);
		\draw[->] (dots2) edge[trans]
		node[above]{$\bar{1}$} (xck1);
		\draw[->] (dots2) to[bend right=30]
		node[below, trans]{$\overline{n{-}1}$} (xck1);
		\node[dots node] at ($(dots2)!0.4!(xck1)+(0,-0.2)$) {$\vdots$};
		
	\end{tikzpicture}
\caption{For any $i > 0$ and any $x \in \{0,\dots,n-1\}$, the transitions of the NFA $Q_i$ from the initial state $(x,0)^{q_i}$.}
\label{fig:qi}
\end{figure}
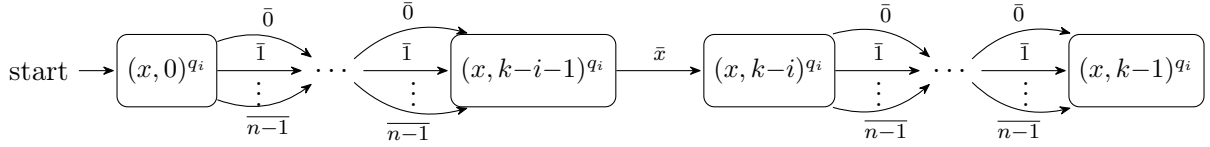

From the construction of the gadgets, we immediately get the following theorem.

\begin{theorem}[Epilogue Theorem]\label{thm:epilogue}
	Let $w$ be some word. Then, there is a run of $Q_i$ on $w$ between some initial
	state $x_i'$ and some final state $x_i$ for every $i$ if and only if $w = \overline{x_{k-1}'}\dots \overline{x_1'}$,
	$x_0' = x_0 = \lozenge$, for each $i > 0$, $x_i' = x_i$ is a node and the stack of $Q_0$
	at the beginning of the run is of the form $\gamma x_1' x_2' \dots x_{k-1}'$. Moreover, in every such collection of runs, the stack of $Q_0$ at the end stores $\gamma$.
\end{theorem}

Note that we immediately get the following bound on the size of $Q_i$.
\begin{proposition}[Size of $Q_i$]\label{prop:size-qi}
	Each $Q_i$ has $O(nk)$ states and can be constructed in $O(n^2k)$ time.
\end{proposition}

Now, it is time to construct the final machines $\mach_0,\dots,\mach_{k-1}$. For each $i$, we construct $\mach_i$ by composing the three gadgets
$P_i, M_i$ and $Q_i$. From the Prologue, the Main Parts and the Epilogue Theorems (Theorems~\ref{thm:prologue},~\ref{thm:main-parts} and~\ref{thm:epilogue}), we immediately get the 
following theorem, which establishes the correctness of the reduction.

\begin{theorem}[Correctness of the Reduction]
	Let $w$ be some word. Then, there is an accepting run of $\mach_i$ on $w$ for every $i$ 
	if and only if there exists a $3(k-1)$-clique in the graph $G$.
\end{theorem}

By Propositions~\ref{prop:size-pi},~\ref{prop:size-mi} and~\ref{prop:size-qi}, it follows that 
\begin{proposition}\label{prop:size-machi}
	Each $\mach_i$ has $O(nk^2)$ states and can be constructed in $O(n^2k^2)$ time.
\end{proposition}

We also note that the PDA $\mach_0$ is deterministic, in the sense that, for every state $q$ and every letter $a$ it has at most one outgoing transition from $q$ labelled by $a$. Note that this can be naturally converted into a complete machine that has exactly one outgoing transition for each letter, by adding sink states. However, for the sake of brevity, we do not do it here. 

This completes the desired reduction from the $3(k-1)$-Clique problem, except for the fact that the input and stack alphabets of our machines are linear in size. We now show how their sizes can be reduced at the cost of a $O(\log n)$ blowup in the state space.

\subsection{Reducing The Alphabet Size}

Now, we will show how to reduce the input and stack alphabet size in the machines $\mach_i$
so that it becomes one of constant size. Note that the input alphabet of each $\mach_i$ is $\{0,\dots,n-1\} \cup \{\overline{0},\dots,\overline{n-1}\} \cup \{@,\#\}$.
Similarly the stack alphabet of $\mach_0$ is $\{0,\dots,n-1\}$. Hence, if we encode the elements
in the set $\{0,\dots,n-1\} \cup \{\overline{0},\dots,\overline{n-1}\}$ by means of words over a constant alphabet and force the machines $\mach_i$ to work over this new encoding, we would get equivalent machines
over just constant-sized alphabets. This is what we shall do now.

Without loss of generality, we can assume that $n$ is a power of 2. 
For any node $\ell \in \{0,\dots,n-1\}$, let msbf$(\ell)$ (resp. lsbf$(\ell)$) be the
most significant bit first encoding of $\ell$ (resp. least significant bit first encoding
of $\ell$) over $\{0,1\}$. In this way, we can represent each node $\ell$ by $\log n$ sized words
over $\{0,1\}$. Now, we shall modify the machines $\mach_0,\dots,\mach_k$ so that instead of reading nodes, i.e., letters from the set $\{0,\dots,n-1\} \cup \{\overline{0},\dots,\overline{n-1}\}$, they read $\log n$ sized words over the alphabet $\{0,1\}$ and interpret them as nodes. A naive way of
doing this would be to introduce a $\log n$ sized gadget for each transition of each machine
which replaces reading a node with reading its $\log n$ sized msbf encoding.
Furthermore, whenever it wants to push some node into the stack, it pushes the input
as it is read, i.e., it pushes the msbf encoding and whenever it wants to pop some node from the stack,
it pops the lsbf encoding of that node.

The problem with this approach is the following: The number of transitions in each machine $\mach_i$ is $O(n^2k^2)$ and so this naive way would increase the number of states of each $\mach_i$ to $O(n^2 k^2 \log n)$, which is undesirable.
However, it turns out that by modifying this naive idea by a bit, we can arrive at the desired machines with just an extra $O(n \log n)$ states.
The modification is simply to combine all the ``naive gadgets'' going between any two pair of
states into one ``smart gadget'', which saves a lot of states and allows us to reuse the gadgets.
We now describe this formally by first making a series of observations regarding the gadgets
that we have constructed.

Let $p$ be some state in some machine $\mach_i$. Recall that $p$ is a 3-tuple, the first
of which is either some node of the graph $G$ or the symbol $\lozenge$.
Now, by examining all the gadgets that we have constructed, we have the following first observation.
\begin{quote}
	\textbf{Observation 1: } Suppose there is some outgoing transition from $p$ that reads $z$ for some node $z \in \{0,\dots,n-1\}$. 
	Then, all the outgoing transitions only read letters from the set $\{0,\dots,n-1\}\}$,
	i.e., they only read nodes of the graph $G$. Furthermore, exactly one of the following cases always applies: 
	\begin{itemize}
		\item For each node $y$ of the graph, there is exactly one outgoing transition from $p$ reading $y$ and all such outgoing transitions lead to the same state.
		\item $p$ stores some node $x$, there is exactly one outgoing transition from $p$ and this outgoing transition reads the letter $x$.
		\item $p$ stores $\lozenge$ and for each node $y$ of the graph, there is exactly
		one outgoing transition from $p$ reading $y$. 
	\end{itemize}
	
	Moreover, if the underlying machine is the PDA $\mach_0$, then 
	\begin{itemize}
		\item Either all the outgoing transitions from $p$ do not push into/pop from the stack.
		\item Or all the outgoing transitions from $p$ push the input letter onto the stack.\qaqed
	\end{itemize}
\end{quote}

Our second observation is a dual to the above observation for letters of the form $\overline{a}$.
\begin{quote}
	\textbf{Observation 2: } Suppose there is some outgoing transition from $p$ that reads $\overline{z}$ for some node $z \in \{0,\dots,n-1\}$. 
	Then, all the outgoing transitions only read letters from the set $\{\overline{0},\dots,\overline{n-1}\}$. Furthermore, exactly one of the following cases always applies:
	\begin{itemize}
		\item For each node $y$ of the graph, there is exactly one outgoing transition from $p$ reading $\overline{y}$ and all such outgoing transitions lead to the same state.
		\item $p$ stores some node $x$, there is exactly one outgoing transition from $p$ and this outgoing transition reads the letter $\overline{x}$.
		\item $p$ stores $\lozenge$ and for each node $y$ of the graph, there is exactly
		one outgoing transition from $p$ reading $\overline{y}$. 
	\end{itemize}
	
	Moreover, if the underlying machine is the PDA $\mach_0$, then all the outgoing transitions
	from $p$ reading a letter of the form $\overline{y}$ pop the node $y$ from the stack.\qaqed 
\end{quote}

Intuitively, these observations mean that whenever $p$ stores a node of the graph $G$, then we do not have
to introduce a separate gadget for every outgoing transition of $p$ reading $z$ or $\overline{z}$ for some node $z$.
Instead, we can club together all these gadgets into one gadget that goes to the same state.
Hence, for ``most'' states, we only need a $\log n$ sized gadget.
As for the other remaining states, they all store $\lozenge$ and there are only constantly many of them (here we use the fact that $k$ is a constant) and so we can afford to introduce a separate gadget
for each such state, which will only increase the number of states to $O(n \log n)$ overall.
We now move on to the formal aspects.

Let $p$ be some state of some machine $\mach_i$. Suppose there is at least one outgoing transition of $p$ that reads some node $z$. We now consider each of the three cases given by Observation 1.

\paragraph*{Case 1. } In this case, for each node $y$, 
we have exactly one outgoing transition reading 
$y$ and all of these outgoing transitions move to the same state $q$ 
(and we have no other outgoing transitions). 
Intuitively, this means that it does not matter which node is being read as long as we are sure
that the input being read is indeed a node (and not $\#$ or $@$).
With this in mind, we can replace all such outgoing transitions with the following gadget between $p$ and $q$. First, we add $\log n$ many states $p := (p,0), (p,1), \dots, (p,\log n) = q$
and then from state $(p,i)$ we move to $(p,i+1)$ upon reading either $0$ or $1$. 
Furthermore, if the machine $\mach_i$ is $\mach_0$, i.e., the PDA, then we know that one of the following two conditions must hold.
\begin{itemize}
	\item Either all the outgoing transitions from $p$ do not change the stack. In this case, in the new 
	gadget as well, no stack operations are performed.
	\item Or all outgoing transitions from $p$ push the input letter that is read
	into the stack. In this case, in the new gadget as well, we push all of the input letters that are read,
	i.e., while moving between $(p,i)$ and $(p,i+1)$ we push either 0 or 1 if the input letter
	that is read is either 0 or 1 respectively. This ensures that every action of pushing a node
	into the stack is replaced by the action of pushing its msbf representation into the stack.
\end{itemize}

See Figure~\ref{fig:case-1} for a representation of this gadget.
\begin{figure}[t]
	\centering
	\begin{subfigure}[t]{0.2\textwidth}
		\centering
		\begin{tikzpicture}[
			>={Stealth[round]},
			shorten >=1pt,
			auto,
			node distance=2.0cm and 1cm,
			every state/.style={
				draw, rounded corners=4pt,
				align=center, font=\small, rectangle,
			},
			trans/.style={font=\scriptsize},
			dots node/.style={draw=none, minimum width=0pt, minimum height=0pt}
			]
			
			\node[state] (p) {$p$};
			\node[state, right=of p] (q) {$q$};
			
			\draw[->] (p) to[bend left=40]
			node[above, trans]{$0$} (q);
			\draw[->] (p) edge[trans]
			node[above]{$1$} (q);
			\draw[->] (p) to[bend right=40]
			node[below, trans]{$n{-}1$} (q);
			\node[dots node] at ($(p)!0.5!(q)+(0,-0.25)$) {$\vdots$};
			
		\end{tikzpicture}
		\caption{Case 1: Outgoing transitions from $p$ without stack updates.}
		\label{fig:gadget-orig-nopush}
	\end{subfigure}
	\hfill
	\begin{subfigure}[t]{0.7\textwidth}
		\centering
		\begin{tikzpicture}[
			>={Stealth[round]},
			shorten >=1pt,
			auto,
			node distance=2.0cm and 0.8cm,
			every state/.style={
				draw, rounded corners=4pt,
				align=center, font=\small, rectangle,
			},
			trans/.style={font=\scriptsize},
			dots node/.style={draw=none, minimum width=0pt, minimum height=0pt}
			]
			
			\node[state] (p0) {$p \ {=}\ (p,0)$};
			\node[state, right=of p0] (p1) {$(p,1)$};
			\node[dots node, right=of p1] (dots) {$\dots$};
			\node[state, right=of dots] (pln1) {$(p,\log n{-}1)$};
			\node[state, right=of pln1] (pln) {$q \ {=} \ (p,\log n)$};
			
			\draw[->] (p0) to[bend left=30]
			node[above, trans]{$0$} (p1);
			\draw[->] (p0) to[bend right=30]
			node[below, trans]{$1$} (p1);
			\draw[->] (p1) to[bend left=30]
			node[above, trans]{$0$} (dots);
			\draw[->] (p1) to[bend right=30]
			node[below, trans]{$1$} (dots);
			\draw[->] (dots) to[bend left=30]
			node[above, trans]{$0$} (pln1);
			\draw[->] (dots) to[bend right=30]
			node[below, trans]{$1$} (pln1);
			\draw[->] (pln1) to[bend left=30]
			node[above, trans]{$0$} (pln);
			\draw[->] (pln1) to[bend right=30]
			node[below, trans]{$1$} (pln);
			
		\end{tikzpicture}
		\caption{Case 1: Replacement of outgoing transitions from $p$ without stack updates.}
		\label{fig:gadget-repl-nopush}
	\end{subfigure}
	
	\bigskip
	
	\begin{subfigure}[t]{0.2\textwidth}
		\centering
		\begin{tikzpicture}[
			>={Stealth[round]},
			shorten >=1pt,
			auto,
			node distance=2.0cm and 1.2cm,
			every state/.style={
				draw, rounded corners=4pt,
				align=center, font=\small, rectangle,
			},
			trans/.style={font=\scriptsize},
			dots node/.style={draw=none, minimum width=0pt, minimum height=0pt}
			]
			
			\node[state] (p) {$p$};
			\node[state, right=of p] (q) {$q$};
			
			\draw[->] (p) to[bend left=40]
			node[above, trans]{\tiny{$0,\ \text{push}(0)$}} (q);
			\draw[->] (p) edge[trans]
			node[above]{\tiny{$1,\ \text{push}(1)$}} (q);
			\draw[->] (p) to[bend right=40]
			node[below, trans]{\tiny{$n{-}1,\ \text{push}(n{-}1)$}} (q);
			\node[dots node] at ($(p)!0.5!(q)+(0,-0.25)$) {$\vdots$};
			
		\end{tikzpicture}
		\caption{Case 1: Outgoing transitions from $p$ with stack updates.}
		\label{fig:gadget-orig-push}
	\end{subfigure}
	\hfill
	\begin{subfigure}[t]{0.7\textwidth}
		\centering
		\begin{tikzpicture}[
			>={Stealth[round]},
			shorten >=1pt,
			auto,
			node distance=2.0cm and 0.8cm,
			every state/.style={
				draw, rounded corners=4pt,
				align=center, font=\small, rectangle,
			},
			trans/.style={font=\scriptsize},
			dots node/.style={draw=none, minimum width=0pt, minimum height=0pt}
			]
			
			\node[state] (p0) {$p \ {=} \ (p,0)$};
			\node[state, right=of p0] (p1) {$(p,1)$};
			\node[dots node, right=of p1] (dots) {$\dots$};
			\node[state, right=of dots] (pln1) {$(p,\log n{-}1)$};
			\node[state, right=of pln1] (pln) {$q \ {=} \ (p,\log n)$};
			
			\draw[->] (p0) to[bend left=30]
			node[above, trans]{$0,\ \text{push}(0)$} (p1);
			\draw[->] (p0) to[bend right=30]
			node[below, trans]{$1,\ \text{push}(1)$} (p1);
			\draw[->] (p1) to[bend left=30]
			node[above, trans]{$0,\ \text{push}(0)$} (dots);
			\draw[->] (p1) to[bend right=30]
			node[below, trans]{$1,\ \text{push}(1)$} (dots);
			\draw[->] (dots) to[bend left=30]
			node[above, trans]{$0,\ \text{push}(0)$} (pln1);
			\draw[->] (dots) to[bend right=30]
			node[below, trans]{$1,\ \text{push}(1)$} (pln1);
			\draw[->] (pln1) to[bend left=30]
			node[above, trans]{$0,\ \text{push}(0)$} (pln);
			\draw[->] (pln1) to[bend right=30]
			node[below, trans]{$1,\ \text{push}(1)$} (pln);
			
		\end{tikzpicture}
		\caption{Case 1:  Replacement of outgoing transitions from $p$ with stack updates.}
		\label{fig:gadget-repl-push}
	\end{subfigure}
	\caption{Case 1: Replacement of outgoing transitions from $p$.}
	\label{fig:case-1}
\end{figure}
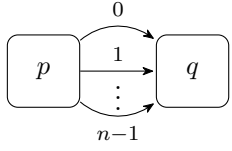
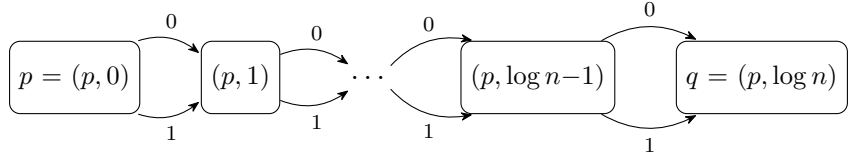
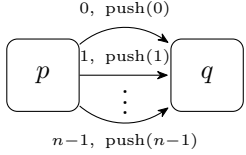
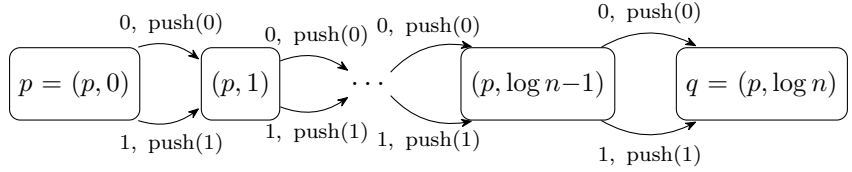

\paragraph*{Case 2. } In this case, $p$ stores some node $x$, there is exactly one outgoing transition
of $p$ and this transition reads the node $x$. Let $q$ be the state to which this transition goes to. Intuitively, this means
that we can only read the node $x$ from this state and so exactly one gadget between $p$ and $q$ suffices here. With this in mind, we can replace this outgoing transition with the following gadget between
$p$ and $q$. First, we add $\log n$ many states $p := (p,0), (p,1), \dots, (p,\log n) = q$
and then from state $(p,i)$ we move to $(p,i+1)$ upon reading the $(i+1)^{th}$ bit in the msbf encoding of $x$. This ensures that any way of using this gadget would be forced to read the msbf encoding of $x$.

In addition to the above modifications, if the machine $\mach_i$ is $\mach_0$, i.e., the PDA, then we know that one of the following two conditions must hold.
\begin{itemize}
	\item Either the (unique) outgoing transition from $p$ does not change the stack. 
	In this case, in the new  gadget as well, no stack operations are performed.
	\item Or the outgoing transition from $p$ pushes the input node $x$ into the stack. In this case, in the new gadget as well, we push the all of the input letters that are read.
\end{itemize}

See Figure~\ref{fig:case-2} for a representation of this gadget.
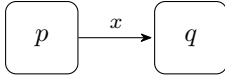
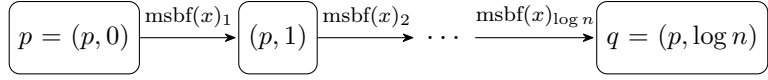
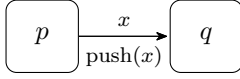
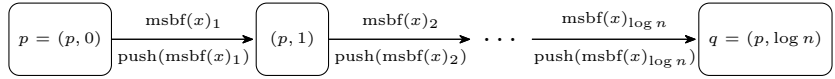
\begin{figure}[t]
	\centering
	\begin{subfigure}[t]{0.25\textwidth}
		\centering
		\begin{tikzpicture}[
			>={Stealth[round]},
			shorten >=1pt,
			auto,
			node distance=2.0cm and 1cm,
			every state/.style={
				draw, rounded corners=4pt,
				align=center, font=\small, rectangle,
			},
			trans/.style={font=\scriptsize},
			dots node/.style={draw=none, minimum width=0pt, minimum height=0pt}
			]
			
			\node[state] (p) {$p$};
			\node[state, right=of p] (q) {$q$};
			
			\draw[->] (p) edge[trans]
			node[above]{$x$} (q);
			
		\end{tikzpicture}
		\caption{Case 2: The unique outgoing transition from $p$ without stack updates.}
		\label{fig:gadget-orig-nopush-2}
	\end{subfigure}
	\hfill
	\begin{subfigure}[t]{0.7\textwidth}
		\centering
		\begin{tikzpicture}[
			>={Stealth[round]},
			shorten >=1pt,
			auto,
			node distance=2.0cm and 1.3cm,
			every state/.style={
				draw, rounded corners=4pt,
				align=center, font=\small, rectangle,
			},
			trans/.style={font=\scriptsize},
			dots node/.style={draw=none, minimum width=0pt, minimum height=0pt}
			]
			
			\node[state] (p0) {$p \ {=}\ (p,0)$};
			\node[state, right=of p0] (p1) {$(p,1)$};
			\node[dots node, right=of p1] (dots) {$\dots$};
			\node[state, right=1.6 cm of dots] (pln) {$q \ {=} \ (p,\log n)$};
			
			\draw[->] (p0) to
			node[above, trans]{$\text{msbf}(x)_1$} (p1);
			\draw[->] (p1) to[]
			node[above, trans]{$\text{msbf}(x)_2$} (dots);
			\draw[->] (dots) to[]
			node[above, trans]{$\text{msbf}(x)_{\log n}$} (pln);

		\end{tikzpicture}
		\caption{Case 2: Replacement of the unique outgoing transition from $p$ without stack updates.
		Here $\text{msbf}(x)_i$ denote the $i^{th}$ bit of the msbf($x$).}
		\label{fig:gadget-repl-nopush-2}
	\end{subfigure}
	
	\bigskip
	
	\begin{subfigure}[t]{0.25\textwidth}
		\centering
		\begin{tikzpicture}[
			>={Stealth[round]},
			shorten >=1pt,
			auto,
			node distance=2.0cm and 1.2cm,
			every state/.style={
				draw, rounded corners=4pt,
				align=center, font=\small, rectangle,
			},
			trans/.style={font=\scriptsize},
			dots node/.style={draw=none, minimum width=0pt, minimum height=0pt}
			]
			
			\node[state] (p) {$p$};
			\node[state, right=of p] (q) {$q$};
			
			\draw[->] (p) edge[trans]
			node[above]{$x$} (q);
			\draw[->] (p) edge[trans]
			node[below]{$\text{push}(x)$} (q);
			
		\end{tikzpicture}
		\caption{Case 2: The unique outgoing transition from $p$ with stack updates.}
		\label{fig:gadget-orig-push-2}
	\end{subfigure}
	\hfill
	\begin{subfigure}[t]{0.7\textwidth}
		\centering
		\begin{tikzpicture}[
			>={Stealth[round]},
			shorten >=1pt,
			auto,
			node distance=2.0cm and 1.9cm,
			every state/.style={
				draw, rounded corners=4pt,
				align=center, font=\tiny, rectangle,
			},
			trans/.style={font=\tiny},
			dots node/.style={draw=none, minimum width=0pt, minimum height=0pt}
			]
			
			\node[state] (p0) {$p \ {=} \ (p,0)$};
			\node[state, right=of p0] (p1) {$(p,1)$};
			\node[dots node, right=of p1] (dots) {$\dots$};
			\node[state, right=2.2 cm of dots] (pln) {$q \ {=} \ (p,\log n)$};
			
			\draw[->] (p0) to[]
			node[above, trans]{$\text{msbf}(x)_1$} (p1);
			\draw[->] (p0) to[]
			node[below, trans]{$\text{push}(\text{msbf}(x)_1)$} (p1);
			
			\draw[->] (p1) to[]
			node[above, trans]{$\text{msbf}(x)_2$} (dots);
			\draw[->] (p1) to[]
			node[below, trans]{$\text{push}(\text{msbf}(x)_2)$} (dots);
			
			\draw[->] (dots) to[]
			node[above, trans]{$\text{msbf}(x)_{\log n}$} (pln);
			\draw[->] (dots) to[]
			node[below, trans]{$\text{push}(\text{msbf}(x)_{\log n})$} (pln);
			
		\end{tikzpicture}
		\caption{Case 2:  Replacement of the unique outgoing transition from $p$ with stack updates.}
		\label{fig:gadget-repl-push-2}
	\end{subfigure}
	\caption{Case 2: Replacement of the unique outgoing transition from $p$.}
	\label{fig:case-2}
\end{figure}

\paragraph*{Case 3. } In this case, $p$ stores $\lozenge$ and for every node $y$,
there is exactly one outgoing transition from $p$ reading $y$. Let $q_y$ be the state 
to which this transition goes to. We now replace this outgoing transitions with the following gadget between $p$ and $q_y$. First, we add 
$\log n$ many states $p := (p,0)^y, (p,1)^y, \dots, (p,\log n)^y := q_y$ and then
from state $(p,i)^y$ we move to $(p,i+1)^y$ upon reading the $(i+1)^{th}$ bit in the msbf encoding
of $y$. Furthermore, if the machine $\mach_i$ is $\mach_0$, i.e., the PDA, then we know that one of the two following conditions must hold.
\begin{itemize}
	\item Either the (unique) outgoing transition reading $y$ from $p$ in does not change the stack.
	In this case, in the new  gadget as well, no stack operations are performed.
	\item Or the outgoing transition reading $y$ from $p$ pushes the input node $y$ into the stack. In this case, in the new gadget as well, we push all of the input letters that are read.
\end{itemize}

See Figure~\ref{fig:case-3} for a representation of this gadget. This completes the replacements for all the cases given by Observation 1.

\begin{figure}[t]
	\centering
	\begin{subfigure}[t]{0.24\textwidth}
		\centering
		\begin{tikzpicture}[
			>={Stealth[round]},
			shorten >=1pt,
			auto,
			node distance=2.0cm and 1cm,
			every state/.style={
				draw, rounded corners=4pt,
				align=center, font=\small, rectangle,
			},
			trans/.style={font=\scriptsize},
			dots node/.style={draw=none, minimum width=0pt, minimum height=0pt}
			]
			
			\node[state] (p) {$p$};
			\node[state, right=of p] (q) {$q_y$};
			
			\draw[->] (p) edge[trans]
			node[above]{$y$} (q);
			
		\end{tikzpicture}
		\caption{Case 3: The unique outgoing transition from $p$ reading $y$ without stack updates.}
		\label{fig:gadget-orig-nopush-3}
	\end{subfigure}
	\hfill
	\begin{subfigure}[t]{0.7\textwidth}
		\centering
		\begin{tikzpicture}[
			>={Stealth[round]},
			shorten >=1pt,
			auto,
			node distance=2.0cm and 1.3cm,
			every state/.style={
				draw, rounded corners=4pt,
				align=center, font=\small, rectangle,
			},
			trans/.style={font=\scriptsize},
			dots node/.style={draw=none, minimum width=0pt, minimum height=0pt}
			]
			
			\node[state] (p0) {$p \ {=}\ (p,0)^y$};
			\node[state, right=of p0] (p1) {$(p,1)^y$};
			\node[dots node, right=of p1] (dots) {$\dots$};
			\node[state, right=1.6 cm of dots] (pln) {$q_y \ {=} \ (p,\log n)^y$};
			
			\draw[->] (p0) to
			node[above, trans]{$\text{msbf}(y)_1$} (p1);
			\draw[->] (p1) to[]
			node[above, trans]{$\text{msbf}(y)_2$} (dots);
			\draw[->] (dots) to[]
			node[above, trans]{$\text{msbf}(y)_{\log n}$} (pln);

		\end{tikzpicture}
		\caption{Case 3: Replacement of the unique outgoing transition from $p$ reading $y$ without stack updates.
			Here $\text{msbf}(y)_i$ denote the $i^{th}$ bit of the msbf($y$).}
		\label{fig:gadget-repl-nopush-3}
	\end{subfigure}
	
	\bigskip
	
	\begin{subfigure}[t]{0.24\textwidth}
		\centering
		\begin{tikzpicture}[
			>={Stealth[round]},
			shorten >=1pt,
			auto,
			node distance=2.0cm and 1.2cm,
			every state/.style={
				draw, rounded corners=4pt,
				align=center, font=\small, rectangle,
			},
			trans/.style={font=\scriptsize},
			dots node/.style={draw=none, minimum width=0pt, minimum height=0pt}
			]
			
			\node[state] (p) {$p$};
			\node[state, right=of p] (q) {$q_y$};
			
			\draw[->] (p) edge[trans]
			node[above]{$y$} (q);
			\draw[->] (p) edge[trans]
			node[below]{$\text{push}(y)$} (q);
			
		\end{tikzpicture}
		\caption{Case 3: The unique outgoing transition from $p$ reading $y$ with stack updates.}
		\label{fig:gadget-orig-push-3}
	\end{subfigure}
	\hfill
	\begin{subfigure}[t]{0.72\textwidth}
		\centering
		\begin{tikzpicture}[
			>={Stealth[round]},
			shorten >=1pt,
			auto,
			node distance=2.0cm and 1.9cm,
			every state/.style={
				draw, rounded corners=4pt,
				align=center, font=\tiny, rectangle,
			},
			trans/.style={font=\tiny},
			dots node/.style={draw=none, minimum width=0pt, minimum height=0pt}
			]
			
			\node[state] (p0) {$p \ {=} \ (p,0)^y$};
			\node[state, right=of p0] (p1) {$(p,1)^y$};
			\node[dots node, right=of p1] (dots) {$\dots$};
			\node[state, right=2.2 cm of dots] (pln) {$q_y \ {=} \ (p,\log n)^y$};
			
			\draw[->] (p0) to[]
			node[above, trans]{$\text{msbf}(y)_1$} (p1);
			\draw[->] (p0) to[]
			node[below, trans]{$\text{push}(\text{msbf}(y)_1)$} (p1);
			
			\draw[->] (p1) to[]
			node[above, trans]{$\text{msbf}(y)_2$} (dots);
			\draw[->] (p1) to[]
			node[below, trans]{$\text{push}(\text{msbf}(y)_2)$} (dots);
			
			\draw[->] (dots) to[]
			node[above, trans]{$\text{msbf}(y)_{\log n}$} (pln);
			\draw[->] (dots) to[]
			node[below, trans]{$\text{push}(\text{msbf}(y)_{\log n})$} (pln);
			
		\end{tikzpicture}
		\caption{Case 3:  Replacement of the unique outgoing transition from $p$ reading $y$ with stack updates.}
		\label{fig:gadget-repl-push-3}
	\end{subfigure}
	\caption{Case 3: Replacement of the unique outgoing transition from $p$ reading $y$.}
	\label{fig:case-3}
\end{figure}
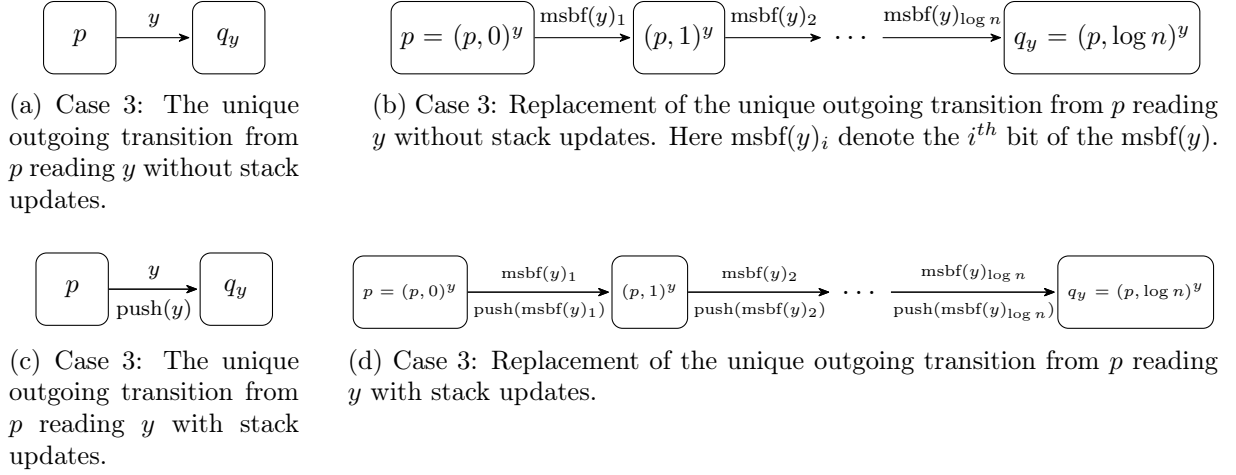

Now, suppose there is some outgoing transition from $p$ that reads some letter of the form
$\overline{z}$ for some node $z$. We now consider each of the three cases given by Observation 2 and design appropriate gadgets for each of them. Intuitively, the gadgets that we shall describe are exactly the same as the ones before, except they will read the lsbf encoding of the node, rather than the msbf encoding. We now move on to the formal aspects.

\paragraph*{Case 1. } In this case, let $q$ be the unique state to which all outgoing transitions of $p$
go to. We replace these transitions by adding $\log n$ many states $p := (p,0), (p,1), \dots, (p,\log n) = q$ and then moving from state $(p,i)$ to $(p,i+1)$ upon reading either $0$ or $1$. 
Furthermore, if the machine $\mach_i$ is $\mach_0$, i.e., the PDA, then 
all outgoing transitions from $p$ reading a letter of the form $\overline{y}$ pop the node $y$ from the stack. Correspondingly, in the new gadget as well, we pop all of the input letters that are read.
This ensures that every action of popping a node
from the stack is replaced by the action of popping its lsbf representation from the stack.

\paragraph*{Case 2. } In this case, $p$ stores some node $x$, there is exactly one outgoing transition of $p$ and this transition reads $\overline{x}$. Let $q$ be the unique
state to which this transition goes to. We replace this transition by adding $\log n$ many states $p := (p,0), (p,1), \dots, (p,\log n) = q$
and then moving from state $(p,i)$ to $(p,i+1)$ upon reading the $(i+1)^{th}$ bit in the lsbf encoding of $x$. 
Furthermore, if the machine $\mach_i$ is $\mach_0$, i.e., the PDA, then 
this outgoing transition from $p$ pops $x$ 
from the stack. Correspondingly, in the new gadget as well, we pop all of the input letters that are read.

\paragraph*{Case 3. } In this case, for every node $y$, let $q_y$ be the (unique) state that $p$ moves to upon reading $\overline{y}$. We replace this outgoing transition from $p$ by adding $\log n$ many states 
$p := (p,0)^y, (p,1)^y, \dots, (p,\log n)^y = q_y$
and then moving from state $(p,i)^y$ to $(p,i+1)^y$ upon reading the $(i+1)^{th}$ bit in the lsbf encoding of $y$. 
Furthermore, if the machine $\mach_i$ is $\mach_0$, i.e., the PDA, then
this outgoing transition from $p$ pops $y$ 
from the stack. Correspondingly, in the new gadget as well, we pop all of the input letters that are read.

This completes our transformation. Call the new machines as $\mach_0',\mach_1',\dots,\mach_{k-1}'$.
Given some word $\gamma \in \{0,\dots,n-1\}^*$, let $\text{msbf}(\gamma)$ (resp. $\text{lsbf}(\gamma)$) denote the word obtained by replacing
each node in $\gamma$ with its msbf (resp. lsbf) representation. The following proposition is immediate from construction.

\begin{proposition}
	The following are true.
	\begin{itemize}
		\item Suppose there is a step of the form $(p,\gamma) \xrightarrow{z} (q,\eta)$ 
		(resp. $(p,\gamma) \xrightarrow{\overline{z}} (q,\eta)$) in $\mach_0$ for some 
		states $p,q$, some node $z$ and some words $\gamma, \eta \in \{0,\dots,n-1\}^*$. Then
		there is a run of the form $(p,\text{lsbf}(\gamma)) \xrightarrow{\text{msbf}(z)} (q,\text{lsbf}(\eta))$
		(resp. $(p,\text{lsbf}(\gamma)) \xrightarrow{\text{lsbf}(z)} (q,\text{lsbf}(\eta))$)
		in $\mach_0'$.
		\item Suppose there is a run of the form $(p,\text{lsbf}(\gamma)) \xrightarrow{w} (q,\eta')$ in $\mach_0'$ such that $\gamma \in \{0,\dots,n-1\}^*$ and
		$q$ is the first state in this run after $p$ that belongs to $\mach_0$. Then $w$ must either 
		be $\text{msbf}(z)$ or $\text{lsbf}(z)$ for some node $z$ and there must be a step of the
		form $(p,\gamma) \xrightarrow{z} (q,\eta)$ (resp. $(p,\gamma) \xrightarrow{\overline{z}} (q,\eta)$ ) where $\eta \in \{0,\dots,n-1\}^*$ and $\text{lsbf}(\eta) = \eta'$.
		
		\item Suppose there is a step of the form $p \xrightarrow{z} q$ 
		(resp. $p \xrightarrow{\overline{z}} q$) in $\mach_i$ for some states $p,q$, some node $z$ and some $i > 0$. Then
		there is a run of the form $p \xrightarrow{\text{msbf}(z)} q$
		(resp. $p \xrightarrow{\text{lsbf}(z)} q$)
		in $\mach_i'$.
		\item Suppose there is a run of the form $p \xrightarrow{w} q$ in $\mach_i'$ for some $i > 0$ such that
		$q$ is the first state in this run after $p$ that belongs to $\mach_i$. Then $w$ must either 
		be $\text{msbf}(z)$ or $\text{lsbf}(z)$ for some node $z$ and there must be a step of the
		form $p \xrightarrow{z} q$ or $p \xrightarrow{\overline{z}} q$ respectively.
	\end{itemize}
\end{proposition}

Given a word $w$ over $\{0,\dots,n-1\} \cup \{\overline{0},\dots,\overline{n-1}\} \cup \{\#,@\}$,
let $\text{tr}(w)$ be the word that is obtained from $w$ by replacing every letter in $\{0,\dots,n-1\}$ 
with its msbf representation and replacing every letter in $\{\overline{0},\dots,\overline{n-1}\}$ with its lsbf representation. Using this proposition it is then easy to see that

\begin{theorem}
	If $w$ is accepted by all of the machines $\mach_0,\mach_1,\dots,\mach_{k-1}$
	then $\text{tr}(w)$ is accepted by all of the machines $\mach_0',\mach_1',\dots,\mach_{k-1}'$.
	Conversely, if $w'$ is accepted by all of the machines $\mach_0',\mach_1',\dots,\mach_{k-1}'$,
	then $w' = \text{tr}(w)$ for some $w$ such that $w$ is accepted by all of the machines $\mach_0,\mach_1,\dots,\mach_{k-1}$.
\end{theorem}

Now, let us analyze the size of each $\mach_i'$. For each state $p$ of $\mach_i$ that does not store
$\lozenge$, we have added only $\log n$ more states. For each state $p$ of $\mach_i$ that stores $\lozenge$, we have added $n \log n$ more states. Note that the number of states of $\mach_i$ that
store $\lozenge$ is a constant (depending only on $k^2$, which itself is a fixed constant). By Proposition~\ref{prop:size-machi} and by using the fact that $k$ is a fixed constant,
it follows that the total number of states in each $\mach_i'$ is $O(n \log n)$. 
Furthermore, by Proposition~\ref{prop:size-machi} it also follows 
that each machine $\mach_i'$
can be constructed in time $O(n^2 \log n)$.
Finally, it can also be verified that the PDA $\mach_0'$ is deterministic, in the sense
that for every state $q$ and every letter $a$, there is at most one outgoing transition
from $q$ labelled by $a$.

Now, suppose we can solve the \mainpdanfa non-emptiness problem
in time $O(M^{\omega (k-1)-\epsilon})$ where $M$ is the maximum number of states among all the given machines and $\epsilon$ is some number strictly bigger than $0$. Then, we can solve the $3(k-1)$-Clique problem in $O(n^{\omega (k-1) - \epsilon/2})$ time as follows: Given a graph $G$, first construct the machines $\mach_0',\dots,\mach'_{k-1}$ (each having $O(n \log n)$ states). Then, run the algorithm for the \mainpdanfa non-emptiness problem on $\mach_0',\dots,\mach'_{k-1}$ and return the answer of this algorithm. By the correctness of the reduction, it follows that this is a correct algorithm for deciding the $3(k-1)$-Clique problem. Furthermore, its running time is $O(n^2 \log n + (n \log n)^{\omega (k-1) - \epsilon})$.
Since $(\log n)^{\omega (k-1) - \epsilon}$ grows asymptotically slower than $n^{\epsilon/2}$ for any $\epsilon > 0$, it follows that we can solve the given instance of the $3(k-1)$-Clique problem in time $O(n^{\omega (k-1) - \epsilon + \epsilon/2})  = O(n^{\omega (k-1) - \epsilon/2})$, which contradicts the $3(k-1)$-Clique hypothesis.
Similarly, we can argue for an $O(n^{3(k-1)})$ lower bound for combinatorial algorithms. 
Theorem~\ref{thm:constant-alphabet-lower-bound} now follows.

	\section{2NPDA$(k)$ and \mainpdanfa Non-Emptiness}
\label{sec:2npdak}

In the previous sections, we have shown lower bounds for the \mainpdanfa non-emptiness problem based on the number of states 
of the underlying machines. 
However, this does not preclude the possibility that an $O(N^{3k-\epsilon})$ time algorithm exists for this problem (for some $\epsilon > 0$), \emph{where $N$ is the total number of bits required to encode all the machines}.
Note that $N$ might be quadratic in the number of states. Indeed, in \emph{dense} machines, the number of transitions is quadratic in the number of states.
Even when $k = 1$, i.e., PDA non-emptiness, no algorithm is known that runs in time $O(N^{3-\epsilon})$ for any $\epsilon > 0$.
Unfortunately, standard existing hypotheses
in fine-grained complexity theory seem to be insufficient for explaining this hardness aspect. 
Furthermore, it is known that, unless breakthrought results in circuit complexity appear, 
perhaps the most well-known hypothesis of fine-grained complexity theory (namely, the strong exponential-time hypothesis, SETH)
cannot be used to help explain this hardness \cite{ChistikovMS22}. 

Recently, a new NFA acceptance hypothesis was introduced~\cite{BringmannGKL24-j}.
Assuming this hypothesis,
no algorithm (combinatorial or otherwise) can solve PDA non-emptiness for dense PDA
in time $O(n^{3-\epsilon})$ for any $\epsilon > 0$ where $n$ is the number of states~\cite{BringmannGKL24-j}. 
However, even this result does not explain the absence of $O(N^{3-\epsilon})$ time algorithms where $N$ is the number of bits of the input.
Also, it is not clear how to extend (or use) that hypothesis to also help explain the absence
of faster algorithms for the \mainpdanfa non-emptiness problem.

In this paper we propose a new hypothesis, called the  \emph{\TwoNPDAkhyp hypothesis},
on the computational complexity of the 2NPDA$(k)$ language recognition problem (see below).
Based on this hypothesis, we prove that there can be no algorithm running in time $O(N^{3k-\epsilon})$
for the \mainpdanfa non-emptiness problem. To corroborate this hypothesis,
we provide reductions between the 2NPDA$(k)$ language recognition problem and other problems in language theory and program analysis. 
We now move on to describing the 2NPDA$(k)$ language recognition problem and the associated hypothesis.

\subsection{Two-Way Multihead Nondeterministic Pushdown Automata}

A \emph{two-way $k$-head nondeterministic pushdown automaton} ($\TwoNPDA(k)$) \cite{HarrisonI68,Ibarra73}
is a machine that 
consists of a finite set of control states, 
a read-only input tape, a pushdown store, i.e., a stack, 
and $k$~heads that read the input tape.
There is a single initial control state and a subset of states marked as accepting.
Based on the control state, the top of the stack, and the letters read by the heads,
the machine can nondeterministically pick a transition that
updates its control state, replaces the top symbol on the pushdown store with a (possibly empty) string, and moves each head to the left or right.
An input word $w$ over the machine's input alphabet is placed on the input tape between designated ``end of tape'' markers $\lmarker$ and $\rmarker$.

The machine starts from its initial state with an empty pushdown store, and applies its transitions.
A run on an input word is a sequence of transitions starting from the initial state consistent with the word.
A run is accepting if it leads to an accepting state; 
we say the word is \emph{accepted} by the machine.
Without loss of generality, we can assume above that a word is accepted in a final state 
with all heads scanning the right end marker and the pushdown store being empty. 
The \emph{language} of the machine is the set of all words that it accepts.

A $\TwoDPDA(k)$ is a $\TwoNPDA(k)$ machine in which the transition relation is deterministic:
there is at most one outcome for any state, top of stack, and letters being read by the heads.
Building upon this intuition, we now give a formal description of $\TwoNPDA(k)$ machines, following \cite{HarrisonI68} and \cite{Ibarra73}.

\subsubsection*{Formal Description of $\TwoNPDA(k)$ Machines}\label{subsec:2npda(k)-description}

A $\TwoNPDA(k)$ machine~\cite{HarrisonI68,Ibarra73} is a tuple
$\twonpda = (Q, \Sigma, \Gamma, \delta, q_0, \br F)$, where 
$Q$ is a finite set of states, 
$\Sigma$ are $\Gamma$ are finite alphabets of input and stack symbols, respectively, 
$q_0 \in Q$ is the initial state, 
$F\subseteq Q$ is the set of final states, 
and
$\delta \subseteq 
Q \times \Sigma^k \times \Gamma \times 
Q \times \Gamma^* \times \set{-1, 0, +1}^k$ is the transition relation.
We assume that $\Sigma$ contains two designated ``end of tape'' symbols $\lmarker$ and $\rmarker$ such that any head of $\twonpda$ cannot move left if it reads $\lmarker$ and cannot move right it if reads $\rmarker$. More precisely, suppose $(q,\vec a, \alpha, q', \gamma, \vec d) \in \delta$ is a transition 
with $\vec{a} = (a_1,\dots,a_k), \vec{d} = (d_1,\dots,d_k)$. Then, if $a_i = \lmarker$ (resp. $a_i = \rmarker$) for any $i$, then we require that $d_i \neq -1$ (resp. $d_i \neq 1$).
As usual, we also assume that $\Gamma$ contains a designated ``end of stack'' symbol~$Z_0$
such that any transition $(q, \vec a, Z_0, q', \gamma, \vec d)\in \delta$ satisfies $\gamma = Z_0$ or $\gamma = Z_0 \gamma'$ for $\gamma' \in (\Gamma \setminus \{Z_0\})^*$.
Thus, no transition of $\twonpda$ replaces $Z_0$ on the stack with a different symbol.
Similarly, we also require that no transition pushes $Z_0$
on the stack when the top of the stack is not $Z_0$, i.e., if $(q,\vec a, \alpha, q', \gamma, \vec d) \in \delta$ is a transition such that $Z_0$ occurs in $\gamma$, then $\alpha = Z_0$ and $\gamma = Z_0 \gamma'$ for $\gamma' \in (\Gamma \setminus \{Z_0\})^*$.

As mentioned before, the $\TwoNPDA(k)$ machine $\twonpda$ has a finite control (states from $Q$) 
and $k$~reading heads.
In a single transition, the machine simultaneously reads
$k$~symbols (elements of $\Sigma$) from the input tape
and also reads the top symbol (an element of $\Gamma$) from the pushdown store.
Based on the transition relation $\delta$, the machine moves by changing the control state, 
replacing the top symbol of the pushdown store by a finite string of symbols 
(possibly the empty string), 
and moving each of its input heads at most one symbol left or right (some heads can remain at the same position). We now formalize this.

A configuration of the $\TwoNPDA(k)$~$\twonpda$ is a triple $(q, w_1 \sigma_1 w_2\ldots \sigma_k w_{k+1}, \gamma)$, where 
$q\in Q$, $w_i\in\Sigma^*$ for each $i\in\set{1,\ldots,k+1}$, $\sigma_1 \ldots \sigma_k$ is a permutation of $\set{1,\ldots, k}$, and $\gamma\in\Gamma^+$.  
(Here, we assume that $\{1,\ldots,k\} \cap \Sigma = \emptyset$.)
Such a configuration represents the situation where the state is $q$, the word $w_1\ldots w_{k+1}$ is on the input tape, and the $k$ heads are at those positions 
on the input tape 
that are preceded in the configuration by $\sigma_1$ to $\sigma_k$. 
In other words, if $\sigma_i = \ell$, then the $\ell^{th}$ head of $\twonpda$ observes 
the first letter of the word~$w_{i+1} \ldots w_{k+1}$. 
Intuitively, $w_{i+1}$ is the word on the input tape between 
the cell observed by the $i^{th}$ leftmost head and 
the cell observed by the $(i+1)^{th}$ leftmost head (excluding the latter). 
Notice that $w_1 = \eps$ if and only if one of the heads observes the leftmost cell with the endmarker $\lmarker$;
and the suffix $w_{k+1}$ is never empty.

We write $c = (q', s, \gamma \alpha) \rightarrow (q', s', \gamma \gamma') = c'$
whenever there is a transition 
$(q, \vec a, \alpha, q', \gamma', \vec d) \in \delta$ with the following properties:
In $c$, each head $i$, observes some letter~$a_i$ such that $\vec a = (a_1, \ldots, a_k)$.
Furthermore, compared to $c$, in $c'$ each head must have moved by $d_i$ positions, where $\vec d = (d_1, \ldots, d_k)$
and only letters from $\Sigma$ are counted, not $\{1, \ldots, k\}$.
Our requirement that any head cannot move to the left of $\lmarker$ and cannot move to the right of $\rmarker$ ensures that none of the heads ``fall off'' the input word.
Note that the input tape is not changed, only the scan positions of each head may change. 
We write $\rightarrow^* $ for the reflexive and transitive closure of $\rightarrow$. 

The initial configuration of $\twonpda$ on a word $w \in (\Sigma \setminus\set{\lmarker,\rmarker})^*$
is $(q_0, \sigma_1\sigma_2\ldots \sigma_k\lmarker w\rmarker, Z_0)$
where each $\sigma_i = i$. An accepting configuration of $\twonpda$ on a word $w \in (\Sigma \setminus\set{\lmarker,\rmarker})^*$
is a configuration of the form $(q, \lmarker w \sigma_1\sigma_2 \ldots \sigma_k\rmarker, Z_0)$
where $q \in F$ and $\sigma_1 \ldots \sigma_k$ is a permutation of $\set{1,\ldots, k}$.
A \emph{run} of the $\TwoNPDA(k)$~$\twonpda$ on a word $w \in (\Sigma \setminus\set{\lmarker,\rmarker})^*$ is a sequence of configurations $C_0, C_1, \dots, C_k$
such that $C_0$ is the initial configuration of $\mach$ on $w$ and each $C_i \rightarrow C_{i+1}$.
A run is said to be accepting if the last configuration of that run is an accepting configuration.
A word $w \in (\Sigma \setminus\set{\lmarker,\rmarker})^*$ is \emph{accepted} by $\twonpda$
if there is an accepting run of $\twonpda$ on~$w$.
The language $\langof{\twonpda}$ of $\TwoNPDA(k)$~$\twonpda$ is the set of all words in $(\Sigma\setminus\set{\lmarker, \rmarker})^*$ that it accepts.




The following is the central decision problem for us:
\begin{mdframed}
\textbf{$\boldTwoNPDA\bm{(k)}$ Language Recognition}\\[1ex]
        \Fix $\TwoNPDA(k)$ $M$.\\
	\Input A word $w$.\\
	\Decide Is $w$ is accepted by $M$?
\end{mdframed}
This actually specifies a family of decision problems, one for each $\TwoNPDA(k)$ machine $M$. When referring to a problem from this family, we write ``\textbf{$\bm{M}$-language recognition}'', for any fixed machine $M$.
We are now ready to state the $\TwoNPDAhyp(k)$ hypothesis, $k \ge 1$. 
\begin{mdframed}
\textbf{$\TwoNPDAhyp\bm{(k)}$ Hypothesis}\\[1ex]
There is a fixed $\TwoNPDA(k)$ machine~$M$ such that the $M$-language recognition problem cannot be solved in time $O(|w|^{3k-\epsilon})$ for any $\epsilon > 0$. 
\end{mdframed}

This is an extension of the $\TwoNPDAhyp(1)$ hypothesis that was introduced by Neal~\cite{Neal}
and Heintze and McAllester~\cite{HeintzeMcAllester}. They successfully used the $\TwoNPDAhyp(1)$ hypothesis to explain the lack of sub-cubic algorithms for many problems in program analysis.


\begin{remark}
The class of languages accepted by $\TwoNPDA(k)$ machines has been studied both in language theory and in complexity theory.
Ibarra~\cite{Ibarra73} proved that the hierarchy is strict: for each $k \ge 1$, the class of languages of $\TwoNPDA(k)$ machines is a strict subset of the class of languages of $\TwoNPDA(k+1)$ machines.
Miyano~\cite{Miyano83} showed that each such class has a hardest language.
Cook~\cite[Corollary~1]{Cook71char} showed that the union of $\TwoNPDA(k)$ languages for all $k \ge 1$ precisely captures the class $\PTIME$.
\end{remark}

\subsection{Equivalences of $\TwoNPDAclass(k)$ with Other Problems}\label{subsec:equivalences}

We now show that the $\TwoNPDA(k)$ language recognition problem is linear-time equivalent to
a collection of other problems from program analysis and language theory.
To this end, let us define a \emph{linear-time reduction} from a problem~$\Pi_1$ to a problem~$\Pi_2$
as a linear-time algorithm $f$ that takes as input an instance $x$ of~$\Pi_1$ and
outputs an instance $f(x)$ of~$\Pi_2$ such that $x$ is a yes-instance of $\Pi_1$ if and only
if $f(x)$ is an yes-instance of $\Pi_2$.
Problems $\Pi_1$ and $\Pi_2$ are linear-time equivalent if
there exist linear-time reductions from $\Pi_1$ to $\Pi_2$ and
                                 from $\Pi_2$ to $\Pi_1$.

We first consider the following decision problem, where
$L$ is a context-free language (CFL), i.e., a language $L$ recognized by some PDA.
\begin{mdframed}
\textbf{CFL $k$-Intersection Reachability}\\[1ex]
\Fix
CFL~$L$.\\
\Input\negthickspace
$k$~NFAs $\fsa_1, \ldots, \fsa_k$
over the same alphabet as $L$.
\\
\Decide
Does the intersection $\langof{\fsa_1} \cap \ldots \cap \langof{\fsa_k}$
contain a word from~$L$?
\end{mdframed}
Note that
$k \ge 1$ is fixed,
$L$ is fixed,
and
the common input alphabet of the $k$~NFAs is also fixed.

Similar to $\TwoNPDA(k)$ language recognition, the above definition actually specifies a family of problems, one for
each $k \ge 1$ and each language $L$. 
When referring to a problem from this family, we will write ``\textbf{$\bm{(L, k)}$-intersection reachability}''.
The problem of DCFL $k$-intersection reachability for DFAs
is the special case of CFL $k$-intersection reachability 
when the language $L$ is a DCFL, i.e., a context-free language recognized by a deterministic PDA
and all of the NFAs $\fsa_1, \ldots, \fsa_k$ are actually DFAs.
The general case of the problem will also be referred to as 
CFL $k$-intersection reachability for NFAs.
For NFAs, note the subfamily with $k=1$ is exactly the well-known CFL reachability problem~\cite{Yannakakis90}.

Having introduced this problem, in the rest of this section, we show (informally speaking) that
for each $k \ge 1$, the following problems are linear-time equivalent:
\begin{itemize}
	\item $\TwoNPDA(k)$ language recognition,
	\item DCFL $k$-intersection reachability for DFAs,
	\item CFL $k$-intersection reachability for NFAs,
	\item \mainpdanfa non-emptiness,
	\item DPDA $\cap$ DFA$^{k-1}$ non-emptiness.
\end{itemize}
Note that $\TwoNPDA(k)$ language recognition is parameterised not only by~$k$ but also by the automaton.
Likewise, CFL $k$-intersection reachability is also parameterised by the CFL.
Thus, we cannot prove that \emph{each} problem from one family is linear-time equivalent to
the \mainpdanfa non-emptiness problem for the same~$k$.
However, what we show is that the \emph{hardest} (most difficult)
language recognition problem for $\TwoNPDA(k)$ is linear-time equivalent to
the \mainpdanfa non-emptiness problem, as illustrated below in the formal version of the above-mentioned equivalences.


\begin{theorem}[Linear-Time Equivalences]\label{th:nonemp=langrec}
	Let $k \ge 1$ be any fixed number.
	\begin{enumerate}
		\item For every $\TwoNPDA(k)$~\aut, there is a DCFL~$L$
		and a linear-time reduction from the \aut-language recognition
		problem to the $(L,k)$-intersection reachability problem for DFAs.
		\item For every CFL~$L$, 
		there is a linear-time reduction from the $(L,k)$-intersection reachability problem for NFAs (or DFAs)
		to
		the \mainpdanfa non-emptiness
		(or \mainpdadfa non-emptiness, respectively).
		Moreover, the reduction produces a DPDA if
		$L$ is a DCFL and the first given NFA is a DFA.
		\item There is a fixed $\TwoNPDA(k)$~$\aut_k$ such that
		\mainpdanfa non-emptiness
		has a linear-time reduction to $\langof{\aut_k}$.
	\end{enumerate}
\end{theorem}

Recall that the stack alphabet in the \mainpdanfa non-emptiness problem has constant size (see Section~\ref{subsec:intersect}).
Indeed, in Theorem~\ref{th:nonemp=langrec}, the reduction from $(L,k)$-intersection reachability produces a PDA whose stack alphabet only depends on $L$.

For $k=1$, a similar triangle of reductions appears in~\cite{Chaudhuri} and later in~\cite{ChistikovMS22}. Decision problems for one-way and two-way machines were originally connected by Hopcroft and Ullman in~\cite{hopcroft1967approach} and later also studied by Rubtsov and Vyalyi~\cite{RubtsovV2026}.

Note that by traversing through this sequence of equivalences, it follows that
there is a fixed $\TwoNPDA(k)$ $\aut_k$ such that, for every $\TwoNPDA(k)$ $\aut$,
the $\aut$-language recognition problem is linear-time reducible to the
$\aut_k$-language recognition problem. Hence, this proves
the existence of a ``hardest'' $\TwoNPDA(k)$ language, in terms of time complexity. 

It is known that there exist ``level by level'' reductions between the 
$\mainpdanfa$ non-emptiness problem
to the deterministic time hierarchy within~$\PTIME$ 
\cite{SwernofskyW15,WeharThesis,OliveiraW20}.
It follows from Theorem~\ref{th:nonemp=langrec} that $\TwoNPDAclass(k)$ has similar reductions.

The equivalences of Theorem~\ref{th:nonemp=langrec}
are with respect to linear-time reductions.
For measuring the computational complexity relative to the number of states and transitions 
in the automata, we can immediately draw the following consequences:

\begin{corollary}
Fix $k \ge 1$ and a finite alphabet $\Sigma$.
\begin{enumerate}
\item
The $\TwoNPDA(k)$ hypothesis is false if and only if
there is $\epsilon > 0$ for which
the \mainpdanfa non-emptiness problem has an algorithm
with running time $O(m^{3 k - \epsilon})$, where
$m$ is the maximum number of transitions in the PDA and NFAs.
The same holds for the special case of the 
\maindpdadfa non-emptiness problem.
\item
There exists a deterministic context-free language~$L$ such that,
unless the $\TwoNPDA(k)$ hypothesis is false,
there exists no $\epsilon > 0$ for which
the $(L,k)$-intersection reachability problem for DFAs has an algorithm
with running time $O(n^{3 k - \epsilon})$, where
$n$ is the maximum number of states of the DFAs.
\end{enumerate}
\end{corollary}

We next prove the three claims of Theorem~\ref{th:nonemp=langrec}.

\subsection{From $\TwoNPDA(k)$ language recognition to DCFL $k$-intersection reachability for DFAs}

Let \aut be a $\TwoNPDA(k)$ machine.
For every $w \in \Sigma^*$,
we show how to construct a DPDA $\pu$ and DFAs $\fsa_1, \ldots, \fsa_k$ such that
\begin{equation*}
	w \in \langof{\aut}
	\quad\text{if and only if}\quad
	\langof{\pu} \cap \langof{\fsa_1} \cap \ldots \cap \langof{\fsa_k} \ne \emptyset.
\end{equation*}
In our construction, the DPDA~\pu will be independent of the input word~$w$
and determined solely by the automaton~\aut;
thus, the DCFL $L$ from the theorem statement will be chosen as $L = \langof{\pu}$.
The DFAs $\fsa_1, \ldots, \fsa_k$ will be constructed in time linear in the length of~$w$.

The input alphabet of~$\pu$, $\fsa_1, \ldots, \fsa_k$ is the set $\delta$
of transitions of the $\TwoNPDA(k)$~\aut.
The language $\langof{\pu} \cap \langof{\fsa_1} \cap \ldots \cap \langof{\fsa_k}$ will consist of
all accepting runs of~\aut on input~$w$.
Indeed, a sequence $\rho = t_1 \ldots t_n \in \delta^*$ is an accepting run
if and only if the following three conditions are satisfied:
\begin{itemize}
	\item
	Transitions of $\rho$ trace a path in the finite graph on the states of~\aut
	from the initial state to a final state.
	\item
	Stack movements prescribed by the sequence~$\rho$ are valid, that is, the sequence of pushes and pops
	specified by the sequence $\rho$ constitutes a valid computation of the underlying stack.
	\item
	For each $i \in [1, k]$, the letters on the input tape that are read by the $i$th head
	in the sequence~$\rho$ are compatible with the input tape containing the word~$\lmarker w \rmarker$,
	where $\lmarker$ and $\rmarker$ are endmarker symbols.
\end{itemize}
In short, the DPDA~\pu checks the first two conditions, and
each DFA~$\fsa_i$ checks the third condition for $i$. We now expand upon the formal details of the reduction.

\paragraph*{Construction of the DPDA~\pu and the DCFL~$L$. }
The deterministic pushdown automaton~\pu verifies the first two conditions in the list above and is completely determined by the $\TwoNPDA(k)$~\aut.
The set of its control states is equal to that of~\aut, call it~$Q$.
For every $t = (q, \vec{a}, \alpha, q', \gamma', \vec{d})\in \delta$,
\pu will have a transition from $q$ to $q'$ labelled by $t$ which will pop $\alpha$
from the stack and push $\gamma'$ into the stack.
This way, both the conditions in the list above are checked by~\pu.
Note that the DPDA~\pu completely ignores the symbol $a \in \Sigma$ and the head movements $d$ prescribed by~$t$. 

To complete the description of~\pu, we choose
the initial state and the set of final states to be the same as they are in~\aut.
By construction, \pu has no $\eps$-transitions.
As already announced, $L = \langof{\pu}$ is the sought DCFL.

\paragraph*{Construction of the DFAs~$\fsa_1, \ldots, \fsa_k$.}
Let $i \in [1, k]$.
The deterministic finite automaton~$\fsa_i$ verifies the third condition in the list
above for~$i$.
The set of control states of $\fsa_i$ is $\{0, 1, \ldots, |w|+1\}$, which we think
of as possible positions of the $i^{th}$ head of \aut over the input tape.
The initial state is~$0$, and the only final state is~$|w|+1$, in line
with the semantics of $\TwoNPDA(k)$.

We refer to the $\lmarker$ (resp. the $\rmarker$ symbols) as the $0^{th}$ (resp. the $(|w|+1)^{th}$ letter) of $w$. Now, let $t = (q, \vec{a}, \alpha, q', \gamma', \vec{d}) \in \delta$ with $\vec{a} = (a_1,\dots,a_k)$
and $\vec{d} = (d_1,\dots,d_k)$. Corresponding to $t$, for every $0 \le j \le |w|+1$, $A_i$ will have a transition from $j$ to $j+d_i$ labelled by $t$ if and only if the $j^{th}$ letter of $w$ is $a_i$. 
This way the DFA $A_i$ checks the third condition in the list above for~$i$.

To ensure that the transition function of each DFAs is total, we
add a rejecting sink state and direct all missing transitions towards it.

\paragraph*{Running time of the reduction. } By construction it is immediately seen that $w \in \langof{\aut}$ if and only if $	\langof{\pu} \cap \langof{\fsa_1} \cap \ldots \cap \langof{\fsa_k} \ne \emptyset$. Regarding the running time, it is easy to see that \pu is determined solely by~\aut. Furthermore, the DFAs $\fsa_1, \ldots, \fsa_k$ depend on~$w$ and have $|w|+2$ states each. Their input alphabet is $\delta$, which is again independent of~$w$, and so all of these machines can be constructed in time that is linear in~$w$. This completes the proof of the reduction.


\subsection{From CFL $k$-intersection reachability to \mainpdanfa non-emptiness}

Let $L$ be a fixed context-free language, i.e., $L$ is recognized by some fixed PDA $P_0$.
We will now reduce the $(L,k)$-intersection reachability problem to the \mainpdanfa non-emptiness problem.

To this end, let $\fsa_1,\dots,\fsa_k$ be $k$~NFAs. Without loss of generality, we can assume that
if at all there is a DFA among these $k$~NFAs, then $\fsa_1$ is one of them. 
Now, the reduction first produces a PDA~$P$ for the language
$L \cap \langof{\fsa_1} = \langof{P_0} \cap \langof{\fsa_1}$
by utilising the standard product construction
(see, e.g., Hopcroft, Motwani, and Ullman's textbook~\cite[Section~7.3.4]{HopcroftMotwaniUllman}).
The set of control states of PDA~$P$ is the Cartesian product
of the sets of control states of~$P_0$ and $\fsa_1$.
Since $P_0$ is fixed, the description size of~$P$ is linear in the description
size of~$\fsa_1$.

The reduction then outputs the PDA~$P$ and the NFAs $\fsa_2, \ldots, \fsa_k$,
which together form the input to the \mainpdanfa non-emptiness problem.
The correctness and running time analysis of the reduction are immediate.

We remark that, if $L$ is a DCFL and $\fsa_1$ is a DFA, then
the product PDA~$P$ is in fact a DPDA~\cite[Theorem~3.1]{GinsburgG66}.

\subsection{From \mainpdanfa non-emptiness to $\TwoNPDA(k)$ language recognition}

\begin{figure*}
\begin{mdframed}
	\begin{algorithmic}[1]
		\State push (the encoding of) the bottom-of-stack symbol of $\fsa_0 = \pu$ onto the stack
		\For{$i = 1$ to $k$:}
		position head $i$ to initial state of $\fsa_{i-1}$
		\EndFor
		\While{true}
		\While{$*$} \Comment{nondeterministic choice: skip or repeat}
		\State move head~$1$ to an outgoing $\eps$-transition $t_0$ in the encoding of $\fsa_0$
		\State execute $t_0$ in $\fsa_0$
		\EndWhile
		\For{$i = 1$ to $k$}
		move head~$i$ to an outgoing transition $t_{i-1}$ in the encoding of $\fsa_{i-1}$
		\EndFor
		\For{$i = 2$ to $k$}
		check that $t_0$ and $t_{i-1}$ read the same input letter $a \in \LangAlph$
		\EndFor
		\For{$i = 1$ to $k$}
		execute $t_{i-1}$ in $\fsa_{i-1}$
		\EndFor
		\While{$*$} \Comment{nondeterministic choice: skip or repeat}
		\State move head~$1$ to an outgoing $\eps$-transition $t_0$ in the encoding of $\fsa_0$
		\State execute $t_0$ in $\fsa_0$
		\EndWhile
		\If{all of $\fsa_0, \fsa_1, \ldots, \fsa_{k-1}$ are accepting:}
		\myaccept
		\EndIf
		\EndWhile
	\end{algorithmic}
\end{mdframed}
\caption{Pseudocode of a $\TwoNPDA(k)$~$\aut_k$ that
that accepts~$w \in \EncAlph^*$ if and only if
$\langof{\pu} \cap \langof{\fsa_1} \cap \ldots \cap \langof{\fsa_{k-1}} \ne \emptyset$.}
\label{fig:2npda:constr}
\end{figure*}

The input to the \mainpdanfa non-emptiness problem
is a concatenation of the string encoding the PDA~\pu
and strings encoding the NFAs~$\fsa_1, \ldots, \fsa_{k-1}$, with delimiters
separating one from another.
The encodings use a fixed alphabet, which we denote by~\EncAlph;
then the input is some $w \in \EncAlph^*$.
In particular, all letters of the input alphabet of $\pu, \fsa_1, \ldots, \fsa_{k-1}$,
denoted by~\LangAlph, and of the stack alphabet of~\pu, denoted by~$\Gamma$,
are encoded by words from~$\EncAlph^*$.
We describe a fixed $\TwoNPDA(k)$~$\aut_k$ that accepts~$w \in \EncAlph^*$ if and only if
$w$~encodes some PDA~\pu and NFAs~$\fsa_1, \ldots, \fsa_{k-1}$ that accept
some word in common, i.e., if and only if $\langof{\pu} \cap \langof{\fsa_1} \cap \ldots \cap \langof{\fsa_{k-1}} \ne \emptyset$.

The description of $\TwoNPDA(k)$~$\aut_k$ will only depend on~$k$
but not on the PDA~\pu or the NFAs~$\fsa_1, \ldots, \fsa_{k-1}$.
The reduction is linear: in fact, it is just a matter of encoding the list $\pu, \fsa_1, \ldots, \fsa_{k-1}$
as a word $w \in \EncAlph^*$. The stack alphabet of $\aut_k$ is $\EncAlph \cup \{Z_0\}$, where $Z_0$ is the ``end of stack'' symbol.

The idea is for $\aut_k$ to guess a word in this intersection and
to simulate, on the fly, $k$~accepting runs on this word in lockstep:
one in the PDA~\pu and $k-1$ in the NFAs~$\fsa_1, \ldots, \fsa_{k-1}$.
In a nutshell, $\aut_k$ uses each one of its heads for keeping track of the states of each machine in these accepting runs,
and the stack for storing the content of the stack of the PDA~\pu.
The pseudocode in Figure~\ref{fig:2npda:constr} 
summarises the construction and is meant to be
seen as the program of the (fixed) $\TwoNPDA(k)$~$\aut_k$.
(For convenience of notation, we use $\fsa_0$ to mean the PDA $\pu$).
We now detail each step of the simulation.

\paragraph*{Initialising the stack. }
The description of the PDA $\fsa_0$ contains the encoding of its bottom-of-stack symbol
as a string over~\EncAlph, which is separated from other parts of the input word~$w$ of~$\aut_k$
by delimiters (also coming from~\EncAlph).
At the beginning of the simulation, $\aut_k$ locates this encoding using one of its heads
and pushes it onto the stack.
The head then returns to the left endmarker.

\paragraph*{Positioning head~$i$ to the initial state of $\fsa_i$. }
As already mentioned, $\aut_k$ uses one head for each of the machines $\fsa_0, \fsa_1, \ldots, \fsa_{k-1}$.
To keep track of the current control state of $\fsa_{i-1}$, the $i^{th}$ head of $\aut_k$ is used.
We can assume that the string encoding the machine~$\fsa_{i-1}$~--- be that PDA~\pu or NFA~$\fsa_{i-1}$~--- %
includes a list of the states of $\fsa_{i-1}$, in which the initial state comes first.
The $\TwoNPDA(k)$~$\aut_k$ moves the $i^{th}$ head right from the left endmarker,
skipping encodings of the first $i-1$ automata completely, and stopping over the first element in
the list of control states of the $i^{th}$ automaton, $\fsa_{i-1}$.

\paragraph*{Guessing and executing transitions. }
This phase of the simulation consists of the three \textbf{for}~loops in the pseudocode above,
as well as the two nondeterministic \textbf{while}~loops handling $\eps$-transitions
in the PDA~$\fsa_0$. We discuss the \textbf{for}~loops first.

In the first loop,
moving each head of~$\aut_k$ to an outgoing transition within the encoding of~$\fsa_{i-1}$
is non-deterministic: $t_{i-1}$ is guessed.
The implementation is self-explanatory, except for the following detail:
As an invariant of the simulation, we require that, in between iterations of
the main \textbf{while} loop in the pseudocode,
head~$i$ of $\aut_k$ is positioned over (the encoding of) the current control state of $\fsa_{i-1}$,
call it $q_{i-1}$,
within the list of all states of $\fsa_{i-1}$.
When the next transition~$t_{i-1}$ of~$\fsa_{i-1}$ is guessed,
the $\TwoNPDA(k)$~$\aut_k$ must check that $t_{i-1}$ departs from~$q_{i-1}$.
To this end, $\aut_k$~first pushes the encoding of $q_{i-1}$ on the stack.
It then moves head~$i$ to the encoding of~$t_{i-1}$, thus guessing~$t_{i-1}$.
At this point $\aut_k$~pops from the stack to check the match of the control state.
If the check fails, the nondeterministic branch rejects.

In the second loop, the goal is to ensure that the guessed transitions~$t_0, t_1, \ldots, t_{k-1}$
all read the same input letter $a \in \LangAlph$ from the input.
Recall that letters of the alphabet~\LangAlph are encoded by words over~\EncAlph.
(In fact, this is why it is not necessarily possible to guess~$a$
upfront and store it in the control state of~$\aut_k$.)
To perform the check, each $i^{th}$ head of~$\aut_k$ locates the encoding of the input letter
within the description of the corresponding transition~$t_{i-1}$.
The heads then move in synchrony to check equality of the letters.
As above, if the check fails, the nondeterministic branch of the computation rejects.

In the third loop, the transitions~$t_0, t_1, \ldots, t_{k-1}$ are executed:
\begin{itemize}
	\item
	For NFA~$\fsa_1, \ldots, \fsa_{k-1}$, it suffices, using head~$i$, to push the
	encoding of the destination of the transition~$t_{i-1}$ onto the stack, then
	locate the list of control states of $\fsa_{i-1}$ and guess the position of
	the destination in that list.
	After that, the stack is popped to compare the destination as recorded on the stack
	(which is popped) with the state in the list,
	ensuring the invariant of the simulation.
	\item
	For PDA~$\fsa_0$, we also need to simulate the operations on the stack.
	Recall that the semantics of a PDA transition dictates that a stack symbol
	$\alpha \in \Gamma$ be popped from the top of the stack and replaced
	by a word $\gamma' \in \Gamma^*$.
	Again as previously, letters of the stack alphabet~$\Gamma$ are encoded
	using words from~$\EncAlph^*$.
	To perform the stack operations, the $\TwoNPDA(k)$~$\aut_k$ locates the encoding of~$\alpha$
	and starts popping the stack, checking that the symbols match the encoding of~$\alpha$.
	If the check fails, the nondeterministic branch of computation rejects
	(because the guessed transition is not available from the current configuration).
	Otherwise $\aut_k$~proceeds to push the encoding of~$\gamma'$.
	After these stack operations, $\aut_k$ goes on to update the current control
	state, as in the case of NFA.
\end{itemize}

The \textbf{for}~loops discussed above ensure that all of
$\fsa_0, \fsa_1, \ldots, \fsa_{k-1}$ synchronise on the input letters,
i.e., in effect $\aut_k$~guesses $k$~sequences of transitions that form
$k$~accepting runs.
However, unlike the NFA~$\fsa_1, \ldots, \fsa_{k-1}$, the PDA~$\fsa_0$
may have $\eps$-transitions.
These are taken care of by the two nondeterministic \textbf{while}~loops:
with the help of head~$0$, $\aut_k$ can simulate an arbitrary sequence of
$\eps$-transitions taken by~$\fsa_0$ before and after \LangAlph-transitions.

\paragraph*{Checking acceptance. }
When $\aut_k$ guesses the end of the word in
$\langof{\fsa_0} \cap \langof{\fsa_1} \cap \ldots \cap \langof{\fsa_{k-1}}$,
it pushes the encoding of the current control states
of $\fsa_0, \fsa_1, \ldots, \fsa_{k-1}$ onto the stack and
then moves the heads to locate these states in the corresponding lists
of final states in the input word~$w \in \EncAlph^*$.
The stack is popped to verify that all these states are indeed final.
After that, one of the heads locates the encoding of the bottom-of-stack
symbol of~$\fsa_0$ within~$w$.
By popping the stack, $\aut_k$~verifies that the simulated stack of~$\fsa_0$
contains this symbol only and, therefore, that $\fsa_0$ has reached
an accepting configuration.
If all checks succeed, $\aut_k$~accepts.
This completes the construction of~$\aut_k$ and hence also the proof of all the three claims of Theorem~\ref{th:nonemp=langrec}.

\subsection{Application: Hardest 2NPDA$(k)$ Languages}

We already observed that Theorem~\ref{th:nonemp=langrec} proves
the existence of a ``hardest'' $\TwoNPDA(k)$ language in terms of time complexity.
In particular, for each~$k$, there exists a \emph{fixed} $\TwoNPDA(k)$ $\aut_k$ such that
for any $\TwoNPDA(k)$~$\aut$, there is a linear-time reduction from the $\aut$-language recognition problem to the $\aut_k$-language recognition problem.
We can strengthen this result by replacing linear-time reductions
with homomorphisms, giving a new proof of the result of Miyano~\cite{Miyano83}.
More precisely, we prove the following result.

\begin{proposition}
\label{hardest}
For each~$k$, there exists a fixed $\TwoNPDA(k)$ $\hrd_k$ over some alphabet $\Sigma_k$
with the following property:
For every $\TwoNPDA(k)$~\aut over a finite alphabet~$\Sigma$
there is a homomorphism $h \colon \Sigma^* \to \Sigma_k^*$ such that,
for every $w \in \Sigma^+$,
we have $w \in \langof{\aut}$ if and only if $h(w) \in \langof{\hrd_k}$.
\end{proposition}

We note that there is a classical result on the existence of ``hardest'' context-free languages by Greibach~\cite{Greibach-hardest}. Furthermore, for $\TwoNPDAclass(1)$, such a language
was first obtained by Rytter~\cite{Rytter-hardest}.
Our hardest languages, $\langof{\hrd_k}$, 
are different from those of Miyano~\cite{Miyano83}.

We do not provide the entire proof of Proposition~\ref{hardest} as it rests on an application
of existing ideas, namely on a similar recent argument for the case of $k = 1$~\cite[Section~8]{ChistikovMS22}.
We provide an outline of the proof, sketching the argument.

Conceptually, our hardest language $\langof{\hrd_k}$ is based on the ``circular''
application of the three reductions of Theorem~\ref{th:nonemp=langrec}.
For a word $w \in \Sigma^*$, the homomorphism $h$ embeds in each morphic
image $h(a)$ with $a \in \Sigma$, a description of the entire $\TwoNPDA(k)$~$\aut$,
encoded using an appropriate but fixed alphabet~$\Sigma_k$.
Roughly speaking, this enables the new fixed $\TwoNPDA(k)$~$\hrd_k$ to simulate~$\aut$,
using the same approach as the pseudocode from Figure~\ref{fig:2npda:constr}.
Movements of each head of $\hrd_k$ between ``blocks'' $h(a)$ with $a \in \Sigma$,
will follow the movements of the corresponding head of~$\aut$ between individual
letters $a$ of the input word~$w$.
The stack of $\hrd_k$ will also mimic the stack of $\aut$.
Auxiliary movements and auxiliary stack operations will be required for the simulation,
which are a bit tedious to describe but present no challenge.

A more sophisticated element of the construction is the handling of the endmarkers.
Intuitively, since the left and right tape delimiters
$\lmarker$ and $\rmarker$ are not given to the morphism~$h$,
special treatment of these two letters is required: the automaton~$\hrd_k$
``bounces back'' to the main part of the tape upon hitting an endmarker
and uses a copy of the description of~\aut embedded in the first (or last)
letter of the tape to continue the simulation.
Extra care is necessary to ensure that~$\hrd_k$ can process the additional
information, namely that some of the heads of the simulated automaton~\aut
are over the endmarker instead of the first (respectively, last) letter
of the input word.
The technique of \cite[Section~8]{ChistikovMS22} can be used to this end.
%
This completes the proof outline, as well as a sketch of the construction
of the hardest language $\langof{\hrd_k}$.

	\section{Conclusion}

In this paper, we have shown a conditional lower bound of $O(n^{2k} |\Sigma| + n^{3k})$ on the running time for the \mainpdanfa non-emptiness problem, where $n$ is the maximum number of states of the given PDA and the NFAs and $\Sigma$ is the common alphabet of these machines. This lower bound is conditional on the (combinatorial) $3k$-Clique hypothesis and matches the running time of the known (combinatorial) algorithms for this problem, thereby providing a tight bound on its complexity. 
Furthermore, we have also shown a conditional lower bound of $n^{3(k-1)}$ for the case when the machines have a constant-sized input alphabet. Finally, to investigate the possibility of algorithms with running time
faster than $O(N^{3k})$ (where $N$ is the total bit size of the input), we also introduced a new hypothesis called the $\TwoNPDA(k)$ hypothesis. We then used this hypothesis to help explain the lack of such algorithms for the \mainpdanfa non-emptiness problem, as well as for other problems in language theory and automata theory.

        \section*{Acknowledgments}
        We thank Marvin K\"unnemann, Neha Rino, Alexander Rubtsov, Henry Sinclair-Banks, and Karol W\k{e}grzycki for useful discussions.
        A. R. Balasubramanian and Rupak Majumdar were sponsored in part by the Deutsche Forschungsgemeinschaft project 389792660 TRR 248---CPEC.
        Dmitry Chistikov is supported by the Engineering and Physical Sciences Research Council [EP/X03027X/1] and
        by the Centre for Discrete Mathematics and its Applications (DIMAP) and Department of Computer Science,
        at the University of Warwick.

	
	\bibliographystyle{IEEEtranS} 
	\bibliography{refs}
\appendix

\end{document}